\documentclass[opre,nonblindrev]{informs4} 

\OneAndAHalfSpacedXI
\usepackage{tikz}
\usetikzlibrary{arrows.meta}
\usepackage{natbib}
\bibpunct[, ]{(}{)}{,}{a}{}{,}%
\def\bibfont{\small}%
\usepackage{booktabs} 

\usepackage[ruled]{algorithm2e} 

\usepackage{ bbm, color, enumerate, amsbsy, amsfonts, amsopn, amssymb,  amsxtra, bezier, graphicx, latexsym, verbatim,
	pictexwd, supertabular, url, dsfont,leftidx, setspace}
\usepackage{mathtools,bm,mathrsfs} 
\usepackage{multirow}
\usepackage{footnote}
\makesavenoteenv{tabular}
\makesavenoteenv{table}

\def\opt{\textsc{OPT}}
\def\mh{\mathcal{H}}
\def\nopt{\textsc{NA-OPT}}
\def\aopt{\textsc{OPT}}
\def\aalg{\textsc{ALG}}
\def\sap{\textsc{SAP}}
	\def\gla{\textsc{GLA}}
\def\alg{\textsc{R}}
\def\salg{r}
\def\optc{\opt^c}
\def\agrd{\text{\textsc{A-}Greedy}}
\def\grd{\text{Greedy}}
\def\ac{A}

\def\wapx{w^{apx}}
\def\aapx{a^{apx}}

\def\apx{\text{Apx-}}

\def\fc{F^{c}}

\def\gc{\h{G}} 

\def\h#1{\widehat{{#1}}}

\def\curn{\mathcal{N}}
\def\curac{\mathcal{\ac}}

\def\mv{\mathcal{D}}
\def\wpa{\textsc{RAVO}}

\def\no{\mathscr{A}}
\def\n1{\h{\mathscr{A}}}

\SetAlFnt{\small}
\SetAlCapFnt{\small}
\SetAlCapNameFnt{\small}
\SetAlCapHSkip{0pt}
\IncMargin{-\parindent}
\TheoremsNumberedThrough     
\ECRepeatTheorems

\EquationsNumberedThrough

\begin{document}
		\TITLE{Optimality of Non-Adaptive Algorithms in 
			Online Submodular Welfare Maximization with Stochastic Outcomes} 
\ARTICLEAUTHORS{	\AUTHOR{	Rajan Udwani}
\AFF{Department of Industrial Engineering and Operations Research, University of California, Berkeley 
\EMAIL{rudwani@berkeley.edu}}}

\ABSTRACT{%
We generalize the problem of online submodular welfare maximization to incorporate various  stochastic elements that have gained significant attention in recent years. We show that a non-adaptive Greedy algorithm, which is oblivious to the realization of these stochastic elements, achieves the best possible competitive ratio among all polynomial-time algorithms, including adaptive ones, unless NP$=$RP. This result holds even when the objective function is not submodular but instead satisfies the weaker submodular order property. Our results unify and strengthen existing competitive ratio bounds across well-studied settings and diverse arrival models, showing that, in general, adaptivity to stochastic elements offers no advantage in terms of competitive ratio.

To establish these results, we introduce a technique that lifts known results from the deterministic setting to the generalized stochastic setting. The technique has broad applicability, enabling us to show that, in certain special cases, non-adaptive Greedy-like algorithms outperform the Greedy algorithm and achieve the optimal competitive ratio. We also apply the technique in reverse to derive new upper bounds on the performance of Greedy-like algorithms in deterministic settings by leveraging upper bounds on the performance of non-adaptive algorithms in stochastic settings.

}%

\maketitle


{\color{red}
}
\section{Introduction}

In the online bipartite matching (OBM) problem \citep{kvv}, we have a bipartite graph between a set of offline vertices (resources) $I$ and a set of online vertices (arrivals) $[T]=\{1,2,\cdots,T\}$. The edges incident to an online vertex are revealed when the vertex arrives and the online algorithm must irrevocably decide which unmatched resource (if any) to match the arrival with. The algorithm may have no knowledge of future arrivals (adversarial arrival model) or may have some distributional information (stochastic arrival model). The objective is to maximize the total number of matches (in expectation). In all arrival models, the goal is to design an online algorithm with the optimal competitive ratio, i.e., the best possible worst case performance in comparison to an optimal offline algorithm with complete knowledge of the entire graph.

OBM and its generalizations have found applications in a wide range of domains. Early applications focused on online search and display adverting \citep{msvv,displayad,survey}, and with the rise of online platforms, there has been a surge of interest in exploring new settings. For example, in online assortment optimization,  each arrival is shown an assortment (a set) of resources and the arrival randomly chooses at most one resource from the set \citep{negin}. In online two-sided assortment optimization, both the resources and the arrivals randomly choose from the options shown to them \citep{aouad2020online}. Motivated by applications in the sharing economy, recent works have also explored settings with reusable resources, where a resource matched to an arrival is used for some time and can be matched to another arrival once the usage period ends \citep{ms}. These generalizations increase the complexity of the problem but model realistic scenarios that arise in many online platforms and marketplaces. 

Many of these new settings share some key features: the algorithm selects an action from a set of available actions at each arrival, the selected action has a \emph{stochastic outcome} that affects one or more resources, and the objective is a \emph{non-linear} function of the outcomes. For instance, in online assortment optimization,
the feasible actions are the assortments the algorithm can show. After an assortment is selected, the arrival randomly chooses at most one resource, leading to a stochastic outcome. While the objective is linear in the number of times each resource is chosen (subject to resource capacity), the objective becomes non-linear in two-sided assortment optimization due to random selections by both arrivals and resources. Similarly, in online allocation of reusable resources, the dynamics of reusability can generally only be captured through a non-linear objective function.

While these features have been investigated within specific contexts and arrival models, a unified formulation that captures and strengthens algorithmic results across these diverse settings remains elusive. In the absence of such a framework, it is often unclear whether established results extend to new settings that modify, combine, or generalize existing problems. Furthermore, several fundamental questions remain unresolved even for specific scenarios, such as the competitive ratio of simple greedy algorithms across different arrival models, the potential for stronger guarantees under stochastic arrivals compared to adversarial ones, and the performance limits when realized outcomes are unobservable to the algorithm. In cases where these questions have been answered, the results are typically established with custom analyses, often making it unclear if the techniques and results can be generalized\footnote{For example, the analyses of the Greedy algorithm in settings studied by \cite{chan2009stochastic}, \cite{mehta}, and \cite{ms} rely on different ideas and benchmarks. 
 	Even when the analysis employs similar techniques such as the primal-dual method \citep{buchbind, devanur}, addressing various combinations of actions and stochastic outcomes seems to require novel algorithms and analytical techniques  \citep{brubach2, borodin2022prophet}.}.

This paper introduces a general framework that unifies a wide range of contemporary settings and settles the question of the optimal competitive ratio achievable in this model across a spectrum of arrival models. We establish that non-adaptive algorithms that remain oblivious to the realization of stochastic outcomes, are sufficient to achieve optimal performance guarantees relative to an adaptive offline benchmark. This implies that the added complexity of adaptivity does not improve the worst-case performance bounds, thereby simplifying the algorithmic requirements for optimality while simultaneously strengthening the state of the art for several specific scenarios. From a practical perspective, our results imply a (perhaps surprising) resilience of non-adaptivity, which is useful in applications where the outcome of an action may be revealed only after a considerable delay or may not be observable at all (see \cite{podcast, feng2025simple, barrientos2025online} for concrete examples). 
%

\subsection{Previous Work}
We start by reviewing previous work most relevant to our goal of developing a general model. Later, in Section \ref{sec:prelimprior}, we discuss prior work in various other settings that we aim to capture as special cases of our formulation. 
{\color{red}
}
\paragraph{Settings with Deterministic Outcomes.} The problem of online submodular welfare maximization (OSW) generalizes several settings with deterministic outcomes, including OBM, the Adwords problem \citep{msvv}, edge weighted matching with free disposal \citep{displayad}, and problems with a concave objective function \citep{concave}. We formally introduce the formulation in Section \ref{sec:newmodels}, but briefly, in this setting, each resource can be matched to multiple requests and the total value from matching resource $i\in I$ to a set of requests is given by a monotone submodular set function $f_i:2^{[T]} \to \mathbb{R}_+$. 

From an algorithmic perspective, OSW unifies several key theoretical results across the aforementioned settings.  Specifically, \cite{kapralov} showed that the Greedy algorithm that maximizes the marginal increase in objective value at each arrival, is $0.5$-competitive in the adversarial arrival model\footnote{As noted in \cite{kapralov}, the result for adversarial arrivals was originally established in earlier works by \cite{nemhauser1978analysis, lehmann2006combinatorial}} and $(1-1/e)$-competitive in the Unknown IID (UIID) model, where the arrival sequence is generated from independent and identically distributed samples drawn from an unknown probability distribution. Furthermore, they showed that, assuming NP$\neq$RP \footnote{In simple terms, this means that problems that can be verified quickly (NP) are not necessarily solvable efficiently with randomized algorithms (RP).}, no polynomial-time online (randomized) algorithm for OSW can achieve a better competitive ratio in the adversarial or the Known IID (KIID) models. In other words, under this widely accepted conjecture, these competitive ratio results are optimal for polynomial-time algorithms. More recently, \cite{rom1} and \cite{rom2} consider (a generalization of) OSW in the Random Order (RO) arrival model, where the adversary selects the problem instance but the online algorithm receives the arrivals in a uniformly random order. The state-of-the-art result in this model, given by \cite{rom2}, is that Greedy is at least $0.5096$-competitive in the RO model. Finding the optimal competitive ratio guarantee in the RO model remains an open problem.

Unfortunately, the OSW problem does not encompass the settings we are interested in. In particular, it cannot model environments with stochastic outcomes, nor problems such as online matching with reusable resources, where we show that the resulting objective function fails to be submodular.

\paragraph{Adaptive Algorithms for Stochastic Outcomes and Adversarial Arrivals.} \cite{chan2009stochastic} formulated a novel framework for stochastic inventory depletion problems and showed that an adaptive generalization of the Greedy algorithm (\agrd) is 0.5-competitive for any adversarial arrival setting that fits within the framework. Their framework captures a variety of settings including online two-sided (and one-sided) assortment optimization. Specifically, \cite{aouad2020online} applied the framework from \cite{chan2009stochastic} to prove that the \agrd\ algorithm is 0.5-competitive for online two-sided assortment optimization in the adversarial model. 

However, this framework does not provide insight into the competitive ratio of online algorithms in the UIID and RO models, nor does it address the performance of non-adaptive algorithms in any arrival model. Furthermore, unlike OSW, determining whether a given setting is included in the framework can be non-trivial, as inclusion depends on properties of the optimal offline solution (see \cite{aouad2020online} and \cite{sumida2024dynamic} for examples). 
\smallskip

\paragraph{Non-Adaptive Algorithms for Stochastic Outcomes.} 
In the realm of offline submodular maximization, \cite{asadpour2016maximizing} introduced a setting with stochastic outcomes and showed that a non-adaptive algorithm has the optimal approximation guarantee of $(1-1/e)$ for this setting. This means that, given an oracle for function values, no polynomial-time algorithm can achieve a better approximation ratio, even among adaptive algorithms. In the online setting, the optimal competitive ratio for non-adaptive algorithms remains an open problem in most settings with stochastic outcomes, with a few exceptions that we discuss in Section \ref{sec:prelimprior}. 

One might consider an online extension of the formulation in \cite{asadpour2016maximizing}. However, in their model, the outcome of each action is a scalar random variable, which prevents the framework from capturing settings such as online assortment optimization, where actions naturally yield vector-valued outcomes (see Section~\ref{sec:prelim} and also Remark~\ref{scalar} in Appendix~\ref{appx:prelim}). Moreover, even if their results could be generalized to handle vector outcomes, their analysis relies on a careful extension of techniques tailored to a specific algorithm for deterministic submodular maximization.

In contrast, our goal is to capture a broad range of settings across different arrival models, potentially involving objectives that are not submodular. Achieving this level of generality requires a fundamentally different approach. To this end, we develop a black-box framework that derives new guarantees for settings with stochastic outcomes by \emph{systematically lifting existing results} from corresponding deterministic settings.

 Our work is also related to the literature on stochastic probing, which studies offline problems where a decision-maker sequentially probes elements and must irrevocably select them subject to feasibility constraints. Probing reveals stochastic values (with known distributions) and may incur costs. Prior work has studied such problems under submodular objectives with vector-valued elements \citep{bradac2019near} and developed reductions that convert cost-free Greedy probing algorithms into costly settings while preserving approximation guarantees \citep{singla2018price}. While these works also involve stochasticity and reductions, they do not address reductions between stochastic and deterministic formulations or online arrival models, and thus differ fundamentally from our setting. At a high level, our approach is also reminiscent of \cite{sosa2}, which uses stochastic coupling to reduce the analysis of adaptive algorithms for online matching with stochastic rewards to deterministic online bipartite matching; however, that work focuses on a specific matching problem, whereas our framework applies more broadly and yields qualitatively different results using different techniques.

\subsection{Our Contributions} 

\paragraph{New Models.} We introduce the problem of online submodular welfare maximization with stochastic outcomes (OSW-SO), a natural generalization of the standard OSW problem that encompasses a broad array of stochastic settings. Extending this framework further, we propose the online \emph{submodular order} welfare maximization with stochastic outcomes (OSOW-SO) setting, where the objective function satisfies a weakened notion of submodularity called (weak) submodular order. We show that this formulation captures settings with reusable resources, where the resulting objective can be non-submodular. This is a new application domain for submodular order functions, which were previously introduced in the context of offline assortment optimization..  

\paragraph{Optimality of Non-adaptive Greedy.} A primary contribution of this work is establishing the optimality of non-adaptive policies in the presence of stochastic outcomes. For the standard OSW problem, the Greedy algorithm is known to achieve competitive ratios of  $0.5$, at least $0.5096$, and $(1-1/e)$ under adversarial, RO, and UIID arrival models, respectively. We show that \emph{non-adaptive} Greedy algorithm achieves the same guarantees for OSW-SO, measured against an \emph{adaptive} offline benchmark.  
Since OSW-SO generalizes OSW, these guarantees are tight among all polynomial-time algorithms in the adversarial and KIID arrival models, unless NP$=$RP. Our results reveal that, in these general settings, adaptivity to stochastic outcomes offers no worst-case theoretical advantage. Furthermore, our analysis strictly improves upon existing competitive ratios for several established special cases, as detailed in Table~\ref{tab}.  

We establish an analogous result in the more general OSOW-SO setting, where we show that non-adaptive Greedy is $0.5$-competitive in the adversarial arrival model. 
In stochastic arrival models, OSOW-SO reduces to OSW-SO.

\paragraph{Beyond Greedy: Performance Limits and Guarantees.}  In several special cases of OSW (deterministic outcomes),  Greedy-like algorithms surpass Greedy and achieve the optimal competitive ratio of $(1-1/e)$ in the adversarial model. 
We show that the \emph{non-adaptive} counterparts of these algorithms retain the same guarantees under stochastic outcomes, yielding new $(1-1/e)$ competitive ratio results for special cases of OSW-SO (see Table \ref{tab}). We also obtain complementary impossibility results in the \emph{reverse} direction: leveraging known limitations of non-adaptive algorithms, we show that a broad family of Greedy-like algorithms cannot achieve a competitive ratio exceeding $0.5$ for a generalization of OSW even when computational efficiency is ignored (e.g., when the algorithm has access to an oracle for NP-hard problems).

\paragraph{A Black-box Technique.} Central to our analysis is a black-box reduction technique that "lifts" results from deterministic settings to those with stochastic outcomes. The idea behind this method is as follows: given an instance of OSOW-SO, we construct an instance of OSOW (with deterministic outcomes), such that the adaptive offline solution in the original instance is a feasible solution in the new instance. Consequently, the offline optimum in the deterministic instance \emph{dominates} the adaptive offline optimum in the original stochastic setting. At the same time, the two instances are indistinguishable to non-adaptive Greedy algorithms, and the value achieved by Greedy is \emph{invariant} across the two constructions. Together, these properties—\emph{dominance} of the offline benchmark and \emph{invariance} of the algorithm’s performance—allow us to lower bound the competitive ratio in stochastic outcome settings using the corresponding competitive ratio established for deterministic outcomes. 

\paragraph{Non-monotone Objectives.}
Finally, we explore the limits of our approach by relaxing the monotonicity assumption. While general non-monotone submodular functions present significant technical challenges, we generalize a result of \cite{nuti} to show that an \emph{adaptive} algorithm achieves the optimal competitive ratio $0.25$ for online non-monotone submodular welfare maximization with stochastic outcomes under adversarial arrivals. Understanding the power of non-adaptive algorithms in the non-monotone regime remains a compelling open question.


\subsubsection*{Outline:} In Section \ref{sec:newmodels}, we introduce the new models and describe both the adaptive and non-adaptive Greedy algorithms, along with the various offline benchmarks used for comparison. Section \ref{sec:prelim} provides a formal discussion of specific settings captured by these models, reviewing prior work and highlighting the new results we derive for these settings. In Section \ref{sec:main}, we present our main results for the Greedy algorithm and explain the reduction technique. Section \ref{sec:gla} extends this technique to other algorithms. In Section \ref{sec:onsw}, we explore the setting of non-monotone submodular functions. Finally, we conclude with a summary of our findings in Section \ref{sec:conclusion}.

\section{Generalizing OSW} 
\label{sec:newmodels}
In this section, we introduce formulations of online submodular maximization with stochastic outcomes. We begin with notation and a generalized version of OSW. Let $F$ be a set function on ground set $\curac$. For $a\in\curac$ and $X\subseteq\curac$, define the marginal value
\[
F(a\mid X)=F(X\cup\{a\})-F(X).
\]
We typically use the following shorthand for a singleton set $\{e\}\subseteq T$: $S\cup e$ for $S\cup \{e\}$, $F(e)$ for $F(\{e\})$, $F(e\mid S)$ for $F(\{e\}\mid S)$, and $F(S\mid e)$ for $F(S\mid \{e\})$. The function $F$ is monotone if $F(B)\le F(A)$ for all $B\subseteq A$, and submodular if $F(e\mid A)\le F(e\mid B)$ for all $B\subseteq A$ and $e\notin A$. We say that $F$ is monotone submodular if it is both monotone and submodular.

\paragraph{OSW (Generalized).}
Let $\curac$ be the ground set of \emph{actions}, partitioned into disjoint sets $\{\ac_t\}_{t=1}^T$ revealed online. At each arrival $t$, the algorithm irrevocably selects exactly one action from $\ac_t$ to maximize $F(S)$, where  $S$ is the set of chosen actions. The algorithm may query $F(S)$ for any $S\subseteq\cup_{\tau\le t}\ac_\tau$.  We assume that the function $F:2^{\curac}\to \mathbb{R}_{+}$, where $\mathbb{R}_{+}$ is the set of non-negative real values, is monotone submodular with $F(\emptyset)=0$, and that each $\ac_t$ contains a null action $0_t$ satisfying $F(0_t\mid X)=0$ for all $X$. This generalized setting is also referred to in the literature as Online Submodular Maximization (OSM) \citep{rom2}. The original OSW formulation proposed by  \citet{kapralov} is a special case of OSM in which there is a set of resources $I$, $F=\sum_{i\in I} F_i$, and $\ac_t=\{(i,t)\}_{i\in I}$. Throughout this paper, we use the term OSW to refer to this generalized OSM formulation unless a distinction is explicitly required (such as in Section \ref{sec:upb}).

{\color{black} 

\subsection{Online Submodular Welfare Maximization with Stochastic Outcomes (OSW-SO)}
  
	In OSW-SO, actions have independent stochastic outcomes, and the objective is defined over realized outcomes. Let $N$ denote the ground set of outcomes corresponding to the action set $\curac$. Upon arrival $t$, the algorithm observes the outcome distributions of actions in $\ac_t$ but only observes the realized outcome of a chosen action after selection. The goal is to maximize expected reward defined by a monotone submodular function $f:2^N\to\mathbb{R}_+$.

\smallskip
	
	\noindent We next describe each component of the formulation in more detail.

	\paragraph{Outcomes.} Let $N(a)$ denote the set of all possible outcomes (the outcome set) of action $a\in \curac$. Let $p_{e}$ denote the probability that $a$ \emph{maps to} $e\in N(a)$, i.e., the probability that $e$ is the outcome of action $a$. The outcome of each action is \emph{independent} of other events in the instance. Since each action has exactly one outcome, we have 
	\[\sum_{e\in N(a)}p_e = 1.\] 
  Furthermore, the outcome sets of different actions are disjoint, i.e., $N(a)\cap N(a')=\emptyset$ for distinct actions $a,a'\in \curac$. For brevity, let $N(S)=\cup_{a\in S} N(a)$ denote the set of possible outcomes for the action set $S$, and note that $N=N(\curac)$. We use $N_t$ as a shorthand for $N(\ac_t)$.
    \smallskip

   \paragraph{Realized Mapping.} Let $\curn$ denote the set of all possible mappings from actions to outcomes, defined as
	\[\curn=\{P\,\mid\, P\subseteq N,\,\, |N(a)\cap P|= 1\,\, \forall a\in \curac\}.\]  
Each $P\in \curn$ corresponds to an injective mapping from $\curac$ to $N$, where each action $a\in \curac$ is mapped to a unique outcome $e\in N(a)$. The injectivity follows from the disjointedness of the outcome sets $N(a)$. We refer to a mapping $P\in \curn$ as a \emph{realized mapping} when every outcome in $P$ is realized, which occurs with probability (w.p.) $\prod_{e\in P} p_e$. We define the function $\gamma:2^N\to [0,1]$ by
	\[\gamma(P)=\underset{e\in P}{\Pi}\, p_{e}.\] 
	For any subset of actions $S\subseteq \curac$, let $\curn(S)$ denote the collection of all realizable outcome sets corresponding to $S$, namely,
	\[\curn(S)=\{P\cap N(S)\mid P\in \curn\}.\]
	Note that $\curn(\curac)=\curn$. From the definitions, the following identities hold for any $S\subseteq \curac$ and $e\in N(S)$:
	\[\sum_{P\in \curn(S)} \gamma(P)=1\quad \text{ and} \quad \sum_{P\in \curn(S)\mid P\ni e} \gamma(P)=p_e.\] 
 The first identity states that some outcome set in $\curn(S)$ is realized, while the second identity states that a given outcome $e\in N$ appears in the realized outcome set with probability $p_e$.
	\smallskip
	
	\paragraph{Objective.} 	The objective value (or reward) associated with a set of outcomes is given by a monotone submodular function $f:2^N\to \mathbb{R}_{+}$, with $f(\emptyset)=0$. This outcome level objective induces a reward over the set of actions.  Specifically, for a realized mapping $P\in \curn$, the total reward obtained by selecting a set of actions $S\subseteq \curac$ is $f(N(S)\cap P)$, where $N(S)\cap P$ is the set of outcomes realized by actions in $S$. Let $F:2^{\curac}\to \mathbb{R}_{+}$ denote the expected total reward of an action set $S$. Then,
	\begin{equation}\label{defF}	F(S)=\sum_{P\in \curn} \gamma(P)\,f\left(N(S)\cap P\right)\quad \forall S\subseteq \curac.
	\end{equation} 
We assume access to value oracles for both $f$ and $F$.  It is immediate that $F$ is monotone submodular whenever $f$ is monotone submodular. 
	\smallskip
	
	\paragraph{Null Actions.}
	Each $\ac_t$ contains a null action $0_t$ with $N(0_t)=\{0_t\}$ and $f(0_t\mid X)=0$ for all $X$. This assumption is without loss of generality under monotonicity.
		
	\smallskip
	
	\paragraph{Adversarial Arrivals.} 
	An instance consists of the sequence of action sets $\{\ac_t\}_{t\in[T]}$, outcome sets $\{N_t\}_{t\in[T]}$, probabilities $\{p_e\}_{e\in N}$, and functions $f$ and $F$, chosen by an oblivious adversary\footnote{Specifically, an oblivious adversary who does not have access to the random bits sampled by a randomized online algorithm.} and revealed online. At arrival $t$, the algorithm observes $\ac_t$, $N_t$, and $\{p_e\}_{e\in N_t}$, may query value oracles on previously revealed actions or outcomes, and must irrevocably select exactly one action from $\ac_t$. Adaptive algorithms observe realized outcomes immediately; non-adaptive algorithms do not. 	This distinction is discussed in more detail in Section \ref{sec:adapnonadap}.
	
	Competitive ratio definitions and extensions to stochastic arrival models (RO and UIID/KIID arrival models) are deferred to Sections~\ref{sec:offline} and~\ref{sec:stochastic}.
Before proceeding, we make a few important remarks regarding the OSW-SO formulation.
\begin{remark}[Overlapping Sets]
In the OSW-SO formulation, we assume that distinct actions have disjoint outcome sets and that distinct arrivals have disjoint action sets. These assumptions are made without loss of generality and solely for technical convenience in the analysis. Specifically, suppose two actions $a$ and $a'$ share a common outcome $e$. We can modify the instance by removing $e$ from $N(a')$ and replacing it with a new ``copy" $e'$, where $e'\in N(a')$, $p_{e'} = p_e$, and $f(e' \mid X) = f(e \mid X)$ for all $X \subseteq N.$
This transformation preserves the cardinality of $N(a')$ as well as the monotonicity and submodularity of $f$, and therefore allows us to assume that outcome sets of distinct actions are disjoint. An analogous duplication argument applies to actions associated with different arrivals, allowing us to assume that the action sets are disjoint across arrivals. Importantly, these transformations do not affect the behavior or running time of the algorithms under consideration, as the disjointness assumptions are used only for analytical purposes and are not required during algorithm execution. In particular, all algorithms can be applied directly to the original instance.
\end{remark}
\begin{remark}[Deterministic Outcomes and Induced Instance]
	In the absence of stochastic outcomes (i.e., in the deterministic setting), an instance $G$ of OSW-SO reduces to an instance of OSW with objective function $F$ on the set of actions $\curac$. Specifically, when $p_e\in\{0,1\}\,\, \forall e\in N$, each action has a fixed (deterministic) outcome because there is a unique mapping $P\in \curn$ such that $\gamma(P)=1$.  In this case,  it is sufficient to focus solely on the set of actions $\curac$ and the objective function $F$. Notably, every instance of OSW-SO naturally induces an instance of OSW with the objective function $F$ defined in \eqref{defF}.
	\end{remark}
We next discuss generalizations of OSW-SO under different structural assumptions on $f$ (and the induced function $F$). 
	\subsubsection{Submodular Order Objective Function.}  
	To capture settings where the objective is not submodular, such as online allocation of reusable resources (see Section \ref{sec:prelim}), we introduce a natural generalization of OSW where the objective $f$ belongs to the broader family of submodular order functions. We begin by introducing some key notation.

\paragraph{Submodular Order \citep{SOF}.}
Let $\pi$ be a total order over $N$. For sets $A,C\subseteq N$, write $C\succ_\pi A$ if every element of $C$ appears after every element of $A$ in $\pi$. For $B\subset A$, we say that the sets are \emph{$\pi$-nested} if $A\setminus B\succ_\pi B$. The order $\pi$ is a \emph{(weak) submodular order}\footnote{This is the more general of the two notions of submodular order introduced in \citet{SOF}. The \emph{strong} submodular order property, which requires that $f(C\mid A)\leq f(C\mid B)$ for \emph{every} $B\subseteq A$, cannot be used to capture allocation of reusable resources.} for $f$ if for all $\pi$-nested $B\subset A$ and all $C\succ_\pi A$,
\[
f(C\mid A)\le f(C\mid B).
\]
We assume that $f(\emptyset\mid A)=0\,\, \forall A\subseteq N$. If every total order $\pi$ is a submodular order for $f$, then $f$ is submodular.

	\smallskip

{\color{black}	\paragraph{Online Submodular Order Welfare Maximization with Stochastic Outcomes (OSOW-SO).} This is a generalization of OSW-SO where $f$ is monotone and has a (weak) submodular order $\pi$ over $N$ that satisfies the following \emph{arrival-consistency} property:
\[	\text{ Order $\pi$ is arrival consistent if $N_{t}\succ_{\pi} N_{t-1}$ for all  $t\in [T]$.}\]

An analogous definition applies to the action space: an order $\pi_{\curac}$ over $\curac$ is arrival-consistent if $A_t \succ_{\pi_{\curac}} A_{t-1}$ for all $t\in[T]$. This property is central to our competitive ratio analysis (see Section~\ref{sec:osow}). Notably, the online algorithms developed for OSOW-SO are order-oblivious, requiring no prior knowledge of the submodular order $\pi$; its existence is leveraged only in the theoretical analysis. Note that every instance of OSW-SO is also an instance of OSOW-SO because any order over $N$ is a (weak) submodular order for a submodular function $f$. For brevity, we hereafter refer to "(weak) submodular order" simply as "submodular order." Example~\ref{adapeg} provides a concrete illustration of this formulation.

\paragraph{Online Submodular Order Welfare Maximization (OSOW).} OSOW is a further generalization of OSW in which the objective function $F$ over ground set $\curac$ is a sum of finitely many monotone functions, each admitting an arrival-consistent submodular order. The specific submodular order may differ across the summands. Intuitively, this flexibility is natural because the relative ordering of actions corresponding to the same arrival is immaterial: any feasible solution selects one action per arrival. Notably, we do \emph{not} require $F$ itself to admit a submodular order, which would be a stronger assumption. 

Every instance of OSOW-SO naturally induces an instance of OSOW via the objective function $F$ defined in \eqref{defF}. Specifically, for every mapping $P\in\curn$, the function $\gamma(P)\,f(N(\cdot)\cap P)$ is monotone and admits an arrival-consistent submodular order over $\curac$. A formal proof is included in Appendix~\ref{appx:newmodels} (Lemma~\ref{ftoF2}). Consequently, $F$ is a sum of monotone functions with arrival-consistent submodular orders, as required by the OSOW formulation.}

\begin{remark}All our results for OSOW-SO continue to hold under the less restrictive assumption that $f$ is a sum of finitely many monotone functions, each admitting an arrival-consistent (Weak) submodular order over $N$. While we use the simpler definition of OSOW-SO for clarity of exposition, this more general formulation ensures that OSOW remains a special case of OSOW-SO.
	\end{remark}

\subsubsection{Non-Monotone Objective Function.}  
To capture settings with non-monotone objectives, we relax the assumption that the functions $f$ and $F$ are monotone and consider the problem of online \emph{non-monotone} submodular welfare maximization problem with stochastic outcomes or ONSW-SO. 
We use ONSW to denote the special case of ONSW-SO with deterministic outcomes. We describe the formulation in more detail and discuss prior work on ONSW in Section \ref{sec:onsw}. Until then, we assume that $f$ and $F$ are monotone functions.

	\subsection {Adaptive and Non-Adaptive Algorithms} \label{sec:adapnonadap}
	We only consider non-anticipatory algorithms, meaning that the algorithm can only observe the realization of an action after it irrevocably selected the action. Given a non-anticipatory online algorithm $\aalg$ for OSOW-SO, we classify it as either non-adaptive or adaptive.
	
	\paragraph{Non-Adaptive Algorithm.} \aalg\ is non-adaptive if its output is independent of the realized mapping. 
	Formally, \aalg\ is non-adaptive if, for each instance, it selects the set same set of actions for every action to outcome mapping. 
	Any algorithm designed for OSOW can be directly used as a non-adaptive algorithm for OSOW-SO by applying it to the induced instance of OSOW. 
	
	\paragraph{Adaptive Algorithm.} An algorithm is called adaptive if it is \emph{not} non-adaptive. The output of an adaptive algorithm can change based on the realized mapping.	We illustrate the distinction with an example below.
	
	\begin{example}\label{adapeg}
			Consider an instance of OSOW-SO, with two arrivals $\{1,2\}$ and action sets $\ac_1=\{a_{1}\}$ and $\ac_2=\{a'_{1},a_{2}\}$. Every action has two possible outcomes, \emph{success} or \emph{failure}, and let $e_1,e'_1,e_2$ denote the successful outcomes of actions $a_1,a'_1,a_2$ respectively. The probabilities of success for each action are as follows: $2p_{e_1}=p_{e'_1}=p_{e_2}=1$. Let the reward function $f$ be defined as follows: 
			\[f(e_1)=f(e'_1)=4f(e_2)=1, f(\{e'_1,e_1\})=f(\{e'_1,e_2\})=1, f(\{e_1,e_2\})=f(\{e_1,e'_1,e_2\})=1.25.\] Actions that result in failure have no value and can be ignored. While $f$ is not submodular -- since $f(e_2\mid e'_1)<f(e_2\mid\{e_1,e'_1\})$ -- it can be verified that $f$ is monotone and has a arrival-consistent submodular order $\pi=\{e_1,e_2,e'_1\}$. 
		
		Now, consider a non-adaptive algorithm that chooses action $a_1$ at arrival 1 and $a'_1$ at arrival 2. This solution has expected total reward \[F(\{a_1,a'_1\})=p_{e_1}f(e_1)+p_{e'_1}(1-p_{e_1})f(e'_1)=1,\] 
		where we use the fact that $f(e'_1\mid e_1)=0$. However, we can do better with an adaptive strategy of picking action $a_2$ at arrival 2 if $a_1$ succeeds and picking action $a'_1$ if $a_1$ fails. The adaptive algorithm has expected total reward \[p_{e_1}[f(e_1)+p_{e_2}f(e_2\mid \{e_1\})]+p_{e'_1}(1-p_{e_1})f(e'_1)=1.125.\]
		\end{example}
	
	Next, we discuss the adaptive and non-adaptive algorithms that, as we show later, achieve the optimal performance guarantees for both OSW-SO and OSOW-SO. 
 
	\subsubsection*{Adaptive and Non-Adaptive Greedy Algorithms.}
	Consider arrival $t\in [T]$ and let $\alg_t$ denote the set of actions chosen prior to $t$. Let $P_t=N(\alg_t)\cap P$ denote the set of realized outcomes of $\alg_t$. The Adaptive-Greedy algorithm, denoted as \agrd, selects the following action at $t$,
	\[\agrd: \quad \salg'_t=\argmax_{a\in \ac_t} \sum_{e\in N(a)}p_e\,f(e\mid P_t).  \]
	Action $\salg'_t$ maximizes the expected marginal increase in the total reward, conditioned on the set of realized outcomes $P_t$.  On the other hand, the non-adaptive Greedy algorithm, or simply \grd, selects the action that maximizes the marginal value of the expected total reward function $F$,
	\[\text{ Greedy: }\quad \salg_t=\argmax_{a\in \ac_t}\,\, F(a\mid \alg_t ). \]
	Clearly, \grd\ is oblivious to the set of realized outcomes.  
	Both Greedy and \agrd\ are well-defined for any function $f$ (and $F$) and do not require the knowledge of the submodular order of $f$ (or $F$). Additionally, when the outcomes are deterministic, \agrd\ and Greedy are identical. We examine the \grd\ algorithm in some well-studied special cases of OSW-SO in Appendix \ref{appx:sapeg}.
	
	\begin{remark}[Computing $\salg'_t$ and $\salg_t$]\label{compremark}
For our main results, we assume access to value oracles for both $f$ and $F$. With an oracle for $f$, one can compute $\salg'_t$ in at most $O\left(|\ac_t|\times \max_{a\in\ac_t} |N(a)|\right)=O(|\ac_t|\times |N_t|)$  time (which is polynomial), where $|\cdot|$ denotes the cardinality of a set. Similarly, with an oracle for $F$, $\salg_t$ can also be computed in 
	at most $O(|A_t|)$ time. However, in certain applications, we encounter the following challenges: $(i)$ The set $\ac_t$ is implicitly defined and $|\ac_t|$ can be exponentially large in the size of the instance, making it difficult to find the optimal action, and $(ii)$ Computing the exact value of $F(S)$ becomes challenging when $S$ contains more than a few elements. To address the first challenge, we find an approximately optimal action in $\ac_t$, rather than the optimal action. For the second challenge, we use a sample average approximation of $F(S)$ instead of computing the exact value. We discuss these approximations in more detail in Appendix \ref{appx:sap}. 
		\end{remark}
\begin{remark}[\agrd\ does not dominate]
		While it might seem intuitive that \agrd\ (weakly) dominates Greedy in every instance with stochastic outcomes, we demonstrate below that there are instances where Greedy has an expected reward up to twice as large as \agrd.

	\begin{example}\label{algequiv}\label{noneg}
	Consider an instance of OSW-SO, with three arrivals $\{1,2,3\}$ and action sets $\ac_1=\{a_1\}$, $\ac_2=\{a'_{1},a_2\}$ and $\ac_3=\{a''_{1}\}$. Each action has two possible outcomes, \emph{success} or \emph{failure}, and let $e_1,e'_1,e_2,e''_1$ represent the successful outcomes of actions $a_1,a'_1,a_2,a''_1$ respectively. Let $\epsilon$ be a small non-negative real value. The success probabilities for each action are as follows: $p_{e_1}=\epsilon$ and $p_{e'_1}=p_{e_2}=p_{e''_1}=1$. The reward values are defined as,  \[f(e_1)=f(e'_1)=f(e''_1)=f(\{e_1,e'_1,e''_1\})=1+\epsilon,\quad f(e_2)=1, \quad f(e_2\mid \{e_1,e''_1\})=1.\] 
	Actions that result in failure have no value and can be ignored. It can be verified that there is a monotone submodular function $f$ with these values. In fact, this is an instance of online matching with stochastic rewards (see Section \ref{sec:prelim}), where outcomes $e_1,e'_1, e''_1$ represent successful matches to vertex 1 (with reward $1+\epsilon$) and outcome $e_2$ represents a successful match to vertex 2 (with reward 1).

Now, let's compare the performance of \agrd\ and \grd. At arrival 1, \agrd\ chooses action $a_1$. If the action succeeds, it chooses actions $a_2$ and $a''_1$, resulting in total award $2+\epsilon$. If action $a_1$ fails, it chooses actions $a'_1$ and $a''_1$, resulting in total award $1+\epsilon$. The expected total reward of \agrd\ is,  $\epsilon (2+\epsilon)+(1-\epsilon)(1+\epsilon)= 1+2\epsilon.$ In contrast, \grd\ chooses actions $a_1,a_2,$ and $a''_1$ with expected total reward $F(\{a_1,a_2,a''_1\})=2+\epsilon.$ Note that \grd\ chooses $a_2$ because $F(a_2\mid a_1)=1$, while $F(a'_1\mid a_1)=1-\epsilon^2$. For $\epsilon\to 0$, \grd\ has has an expected reward that is twice that of \agrd.
	\end{example}
\end{remark}

	\subsection{Offline Benchmark and Competitive Ratio} \label{sec:offline}
%

We evaluate online algorithms against an optimal offline benchmark $\opt$. In the offline setting, the full instance $G$ is known in advance, but the realized mapping is not known \emph{a priori}\footnote{Comparisons against benchmarks with prior knowledge of realizations are standardly uninformative.}. The benchmark $\opt$ selects exactly one action from each set in $\{A_1,\ldots,A_T\}$ and may visit these sets in any order. It is non-anticipatory, observing the outcome of an action only after irrevocably selecting it, and may adapt both its action choices and visitation order based on observed outcomes. Equivalently, $\opt$ can be formulated as a stochastic dynamic program.

This benchmark strictly dominates the clairvoyant benchmark of \citet{chan2009stochastic}, which requires actions to be selected in arrival order. In certain settings, such as stochastic rewards with patience, our results hold against even stronger offline benchmarks (see Section~\ref{sec:prelim} and Appendix~\ref{appx:patience}).

%




\paragraph{Competitive Ratio.} Let $\opt(G)$ and $\aalg(G)$ denote the expected total reward of \opt\ and \aalg\ on instance $G$. 
In the adversarial arrival model, we evaluate the performance of \aalg\ by analyzing its competitive ratio, 
\[\inf_{G\in \mathcal{G}}\frac{\text{\aalg}(G)}{\opt(G)},\]
here $\mathcal{G}$ denotes the set of all instances. We say that $\aalg$ is $\beta$-competitive against \opt\ if the competitive ratio is $\beta$. We will define the competitive ratio in other arrival models in Section \ref{sec:stochastic}. 
To make it easy to compare an online algorithm with \opt, we use a deterministic upper bound on $\opt(G)$.

\subsubsection{Upper Bound on $\aopt(G)$.} Let $\Delta(\curac)$ represent the probability simplex over the set of actions $\curac$. Specifically, a vector $Y=(y_a)_{a\in \curac}\in \Delta(\curac)$ if and only if 
\[ 
\sum_{a\in \ac_t} y_a= 1\quad \forall t\in [T],\quad y_a\geq 0\quad \forall a\in \curac.\] 
We also define the set:
\[  \mathcal{X}=\{X\,\mid\, X\subseteq N,\,|X\cap N_t|=1\,\, \forall t\in [T]\},\] 
which is the set of all possible outcome sets of feasible action sets (one action per arrival). 
Using these definitions, we introduce a new function $F^c:\Delta(\curac)\to \mathbb{R}_{+}$, 
\[\fc(Y)=\max\left\{\sum_{X\in \mathcal{X}}\alpha(X) f(X) \mid \sum_{X\in \mathcal{X}} \alpha(X)= 1,\,\, \sum_{X\in \mathcal{X}\mid X\ni e} \alpha({X})\,  = \, p_e\,y_a\,\, \forall e\in N(a),\, a\in \curac\right\}.\] 
Here, $\fc(Y)$ is the maximum expected value of $f(X)$, where $X$ is a random set from $\mathcal{X}$ and the maximum is taken over all probability distributions such that each element  $e\in N(a)$ appears in $X$ with probability $p_e\, y_a$ (not necessarily independently). We now define the following quantity as an upper bound on $\opt(G)$:
\[\optc(G)\,\,=\,\,  \max_{Y\in\Delta(\curac)} F^c(Y).\]
\begin{lemma}\label{concdom}
For every instance $G$ of OSOW-SO (and OSW-SO) we have $\optc(G)\geq \opt(G).$
\end{lemma}
We include the proof in Appendix \ref{appx:newmodels}. The key idea is to show that the output of \opt\ induces a probability distribution over $\cal{X}$ that is a feasible solution for the optimization problem in $\optc$ and its objective value is equal to $\opt(G)$. The definition of $F^c(\cdot)$ is inspired by the concave closure of set functions \citep{vondrak2007submodularity}. A similar deterministic upper bound for the offline stochastic submodular maximization problem was proposed by \cite{asadpour2016maximizing}. 

\subsubsection{Non-Adaptive Offline Benchmark.} In the absence of stochastic outcomes, there is a deterministic mapping from $\curac$ to $N$ and the value of the optimal offline solution, denoted by \nopt, is given by,
\[\nopt(G)= \max_{a_t\in \ac_t\,\,\forall t\in [T]}\,\, F\left(\cup_{t\in [T]}\{a_{t}\}\right).\]
When there are no stochastic outcomes, we have $\opt(G)=\nopt(G)$. Recall that each instance of OSOW-SO induces an instance of OSOW with action set $\curac$ and objective function $F$. If $G$ is an instance of OSOW-SO, then $\nopt(G)$ is the value of the optimal \emph{non-adaptive} offline solution for OSOW-SO. 
Therefore, we have the relationship: 
\[\opt(G)\geq \nopt(G).\]
Here, ``NA" in \nopt\ stands for ``Non-Adaptive".


}

\section{New Results for Well-Studied Settings}\label{sec:prelim}

We study several well-known online allocation settings that fall outside the classical OSW framework but are captured by our OSOW-SO formulation. Using the algorithmic techniques developed for the general model, we derive new guarantees for each setting. We focus on the key features of each setting and defer details to Appendix~\ref{appx:prelim}, where we also formally show that all settings discussed here are special cases of OSOW-SO.

\subsubsection*{Online Matching with Stochastic Rewards.}
The model consists of a bipartite graph between resources $I$ and arrivals $[T]$. Upon arrival $t$, the algorithm selects an incident edge $(i,t)$, which succeeds independently with probability $p_{i,t}$. Each successful match of resource $i$ yields reward $r_i$, up to a capacity $c_i$, after which additional successes generate no reward. Resources may be matched arbitrarily many times\footnote{From an algorithmic viewpoint, this is equivalent to the setting where resource $i\in I$ can be successfully matched to at most $c_i$ arrivals and arrivals may be discarded without choosing any action.}, and the objective is to maximize expected total reward.

\subsubsection*{Stochastic Rewards with Patience.}
Each arrival $t$ is endowed with a patience parameter $k_t$, possibly random with a known distribution. The algorithm may attempt to match $t$ sequentially to distinct resources until a success occurs or $k_t$ attempts fail, subject to feasibility and ordering constraints. Outcomes are independent across attempts, and each feasible sequence of attempts constitutes an action whose outcome is a binary vector indicating the first successful resource (if any). Our guarantees hold against the stronger offline benchmark of \citet{borodin2022prophet}, which may interleave attempts across arrivals (Appendix~\ref{appx:patience}).

\subsubsection*{Online Assortment Optimization.}
Upon arrival $t$, a choice model $\phi_t$ is revealed, and the algorithm offers an assortment $U\subseteq I$. The arrival selects at most one resource $i\in U$ with probability $\phi_t(i,U)$. Each resource $i$ yields reward $r_i$ for its first $c_i$ selections\footnote{Under the weak substitutability assumption of \citet{negin}, this formulation is equivalent to models that restrict assortments to available resources, since any algorithm that offers unavailable resources can be transformed into an equivalent (adaptive) algorithm that only shows available resources to each arrival. See the Probability Matching algorithm in \cite{reuse} for details.}. Outcomes are vector-valued, and the objective is to maximize expected total reward. The offline version of this problem cannot be captured by the framework of \citet{asadpour2016maximizing} due to its inherently scalar-valued outcomes (see Appendix~\ref{appx:prelim}). 

\subsubsection*{Two-Sided Assortment Optimization.}
After arrivals choose resources, each resource $i$ independently selects at most one arrival from those that chose it according to a known choice model $\phi_i$. A reward $r_i$ is earned if such a selection occurs. As shown by \citet{aouad2020online}, under mild assumptions on $\phi_i$, the resulting objective is monotone submodular.

\subsubsection*{Online Matching with Reusable Resources.}
This model generalizes classic online bipartite matching by allowing resources to be rented for a fixed deterministic duration, after which they return and may be matched again. Each rental yields the same fixed reward. Although the resulting objective is not submodular, it admits an arrival-consistent submodular order (Appendix~\ref{appx:reuse}).

\subsection{Previous Work}\label{sec:prelimprior}

For non-adaptive algorithms, \citet{mehta} showed that Greedy is $0.5$-competitive for online matching with stochastic rewards, and \citet{deb} proved this ratio is optimal for non-adaptive algorithms. No algorithm can asymptotically exceed $(1-1/e)$, even in classical OBM with large capacities \citep{pruhs}. For online assortment optimization with reusable resources, \citet{reuse} obtained a $(1-1/e)-\delta(c_{\min})$ competitive non-adaptive algorithm, where $\delta(c_{\min})\to 0$ as the minimum resource capacity $c_{\min}=\min_{i\in I} c_i\to\infty$. LP-based non-adaptive algorithms achieve a $(1-1/e)$-competitive ratio in KIID models for stochastic rewards with patience and online assortment optimization \citep{brubach2,borodin2022prophet}. Next, we discuss known results for adaptive algorithms. 

For adaptive algorithms under adversarial arrivals, \citet{chan2009stochastic} established that \agrd\ is $0.5$-competitive for a framework encompassing stochastic rewards, assortment optimization, and two-sided assortments \citep{aouad2020online}. The same guarantee was shown for stochastic rewards with patience by \citet{brubach2} and \citet{borodin2022prophet}, and for assortments with reusable resources by \citet{ms}. In large capacity regimes, \citet{negin} showed that a Balance algorithm (defined in Section \ref{sec:gla}) is  $(1-1/e)$-competitive for assortment optimization, while \citet{aouad2020online} proved that no algorithm can exceed $0.5$ for two-sided assortments.

In the Random Order (RO) model, any algorithm that is $\alpha$-competitive under adversarial arrivals remains $\alpha$-competitive. \citet{borodin2020greedy} showed that \agrd\ achieves $(1-1/e)$ for stochastic rewards, while $0.5$ remained the best known ratio for assortments, two-sided assortments, and patience. In the large capacity regime, Balance attains an asymptotic guarantee of $0.75$ for assortments in the RO model \citep{negin}. Since the UIID model is weaker than RO, all RO guarantees carry over (see Appendix \ref{appx:ro2iid}). \cite{aouad2020online} showed that the \agrd\ algorithm is $(1-1/e)$-competitive for two-sided assortment optimization in the UIID model\footnote{Adaptive and non-adaptive approximations have also been studied for offline two-sided assortment optimization. See \cite{torrico2020multi,housni2024two} for more details.}, which remains the best known guarantee even for simpler settings such as stochastic rewards and online assortments. 
For stochastic rewards with patience, no competitive ratio better than  $0.5$ is known in the UIID model.  

\subsection{New Results}\label{sec:prelimnew}
	\begin{table}
	\centering
	\begin{tabular}{l|p{20mm}|p{20mm}|p{20mm}|p{20mm}}
		\toprule
		& \multicolumn{2}{c|}{\textbf{Adversarial}} &\multicolumn{1}{c|}{ \textbf{RO}} & \multicolumn{1}{c}{\textbf{UIID}} \\
		\cmidrule{2-3}
		\newline \textbf{Setting} & 	\textbf{General} & 	\textbf{Large Capacity} & & \\
		\midrule
		Stochastic Rewards &  {\small $0.5$}  & {\small $(1-1/e)$ }& {\small$(1-1/e)$ [A]}\newline \bm{$0.5096$}  & {\small$(1-1/e)$ [A]}\newline \bm{$(1-1/e)$}\\
		\hline
		Stochastic Rewards with Patience & {\small  $0.5$ [A]}\newline \bm{$0.5$} & {\small 0.5 [A]}\newline \bm{$(1-1/e)^*$} & {\small  $0.5$ [A]}\newline \bm{$0.5096^*$} & {\small  $0.5$ [A]} \newline \bm{$(1-1/e)^*$}\\
		\hline
		Assortment Optimization & {\small  $0.5$ [A]}\newline \bm{$0.5$} & {\small $(1-1/e)$ } & {\small  $0.5$ [A]} \newline \bm{$0.5096^*$}& {\small$(1-1/e)$ [A]} \newline \bm{$(1-1/e)$}\\
		\hline
		Two-Sided Assortment Optimization & {\small  $0.5$ [A]}\newline \bm{$0.5$} & {\small  $0.5$ [A]} \newline \bm{$0.5$} & {\small  $0.5$ [A]} \newline \bm{$0.5096^*$} & {\small$(1-1/e)$ [A]} \newline \bm{$(1-1/e)$}\\
		\bottomrule
	\end{tabular}
	\caption{Summary of competitive ratio results for non-adaptive algorithms, with newly established results highlighted in bold. We report the best known results for adaptive algorithms, denoted by {\small [A]}, when no prior non-adaptive results exist.  A $(^*)$ indicates settings where our results surpass the previous state-of-the-art for adaptive algorithms. Computational challenges common in assortment optimization and stochastic rewards with patience are ignored, and results are based on an exact implementation of the Greedy algorithm (see Remark \ref{sapremark} and Appendix \ref{appx:sap} for further discussion).}  \label{tab}
\end{table}
 Using our new results for OSW-SO, we derive several new results that are summarized in Table \ref{tab}. Notably, the $(1-1/e)$-competitive ratio in the large capacity regime (for adversarial arrivals) holds for a general formulation, termed Resource Allocation with Vector Outcomes (\wpa), which will be introduced in Section~\ref{sec:gla}. \wpa\ which captures all settings discussed here except two-sided assortment optimization, for which a $0.5$ upper bound exists \citep{aouad2020online}. Our results also improve upon the best previously known adaptive guarantees in several RO and UIID settings.

Beyond Table~\ref{tab}, we apply our black-box reduction in reverse to obtain an upper bound of $0.5$ on the competitive ratio of a broad class of Greedy-type algorithms for OSW with deterministic outcomes. This bound applies even to algorithms with access to computationally unbounded oracles and strengthens the complexity-theoretic limitations of \citet{kapralov}. The generality of OSOW-SO further allows us to handle hybrid models combining two-sided assortment optimization, stochastic rewards with patience, arbitrary feasibility constraints, and correlated outcomes within arrivals, while maintaining independence across arrivals. Finally, under adversarial arrivals, our results imply that Greedy is $0.5$-competitive for online allocation of reusable resources with deterministic usage durations in the presence of arbitrary stochastic outcomes, including combinations such as stochastic rewards with patience that have not previously been analyzed.

\section{Main Results}\label{sec:main}

In this section, we present our main results for OSOW-SO and OSW-SO. The first two subsections, \ref{sec:osow} and \ref{sec:pas}, focus on the adversarial arrival model. In Section \ref{sec:osow}, we analyze the Greedy algorithm for OSOW and establish a competitive ratio of  $0.5$ against $\nopt$. In Section \ref{sec:pas}, we introduce the reduction technique and use it to show that (non-adaptive) Greedy is also  $0.5$-competitive against $\optc$. Sections \ref{sec:rom} and \ref{sec:iid} describe the RO and UIID models, respectively, extending the reduction technique to both models.  We briefly discuss computational challenges associated with implementing Greedy and possible mitigation strategies in Remark~\ref{sapremark}, with further details provided in Appendix~\ref{appx:sap}.

\subsection{Optimal Competitive Ratio for OSOW}\label{sec:osow}
Without stochastic outcomes, the \grd\ algorithm is known to be  $0.5$-competitive for OSW. In contrast, no such guarantee is known for OSOW, where the objective function need not be submodular but is instead a sum of monotone submodular order functions. In this section, we establish a matching guarantee for the more general setting.
\begin{theorem}\label{advmain}
	Greedy is  $0.5$-competitive for OSOW in the adversarial model.
\end{theorem}
Since OSOW generalizes OSW, this guarantee is optimal: unless NP$=$RP, no polynomial-time algorithm can achieve a competitive ratio strictly better than  $0.5$ for OSOW.

To prove Theorem~\ref{advmain}, we introduce some notation and a key structural property of submodular order functions. Consider an instance $G$ of OSOW. Let $\salg_t$ denote the action selected by Greedy at arrival $t$ and let $\alg_t$ denote the set of actions chosen prior to arrival $t$. Let $\opt_{T+1}$ denote an optimal offline solution (set of actions) for instance $G$, and let $o_t=\opt_{T+1}\cap A_t$ represent the action selected at $t$ in the offline solution. Finally, let $\opt_t=\{o_1,\cdots,o_{t-1}\}$ denote the set of actions selected by the offline solution prior to arrival $t$. We note that in settings with deterministic outcomes, the optimal offline solution coincides with \nopt.  

{\color{black}The following lemma is the main technical ingredient in our analysis of the Greedy algorithm for OSOW. It follows from monotonicity, the submodular order property, and arrival-consistency of submodular orders, and can be viewed as a generalization of Corollary~2 in \cite{SOF}. Its proof is included in Appendix~\ref{appx:osow}.
\begin{lemma}\label{SOineq}
	Consider a collection of monotone functions \(\{F_1,\cdots,F_u\}\) defined on the ground set \(\mathcal{A}\). Suppose that each function admits an arrival-consistent submodular order (not necessarily the same order across functions). Then, the function $F \coloneqq \sum_{j\in[u]} F_j$ satisfies the following inequality:
	\[
	F(\opt_{T+1} \cup \alg_{T+1})
	\;\le\;
	F\!\left(\alg_{T+1}\right)
	\;+\;
	\sum_{t\in[T]}
	F\!\left(\{o_t\}\setminus\{\salg_t\}
	\;\middle|\;
	\cup_{\tau\in[t-1]} \{\salg_\tau\}\right).
	\]
\end{lemma}}	
\begin{proof}{Proof of Theorem \ref{advmain}.}
Clearly, the competitive ratio of \grd\ cannot exceed  $0.5$ because OSW is a special case of OSOW. We now show that $F(\alg_{T+1})\geq 0.5\,F(\opt_{T+1})$, which establishes that the competitive ratio of \grd\ is exactly  $0.5$.

First, observe that $F$ is monotone because it is a non-negative linear combination of monotone functions. We have
 \begin{eqnarray}
F (\opt_{T+1})&\leq &F\left( \opt_{T+1}\cup \alg_{T+1}\right),\nonumber\\
 	&\leq & F\left(\alg_{T+1}\right)+ \sum_{t\in [T]} F\left(\{o_{t}\}\backslash \{\salg_{t}\}\mid \cup_{\tau\in[t-1]} \{\salg_\tau\}\right),\nonumber\\
	 &= & F(\alg_{T+1})+ \sum_{t\in [T]} F(\{o_{t}\}\backslash \{\salg_{t}\}\mid \alg_t),\nonumber\\
 	 &\leq &F(\alg_{T+1})+ \sum_{t\in [T]} F(\salg_t\mid \alg_t),\label{fourth}\\
 	 &=& 2F(\alg_{T+1})\nonumber.
 	\end{eqnarray}
  The first inequality follows from the monotonicity of $F$. The second inequality follows from Lemma \ref{SOineq}. The third inequality uses the fact that $\salg_t$ is the greedy (i.e., marginally ``best") action at $t$, implying $F(\{o_t\}\backslash \{\salg_t\}\mid \alg_t)\leq F(\salg_t\mid \alg_t)$. Note that if $o_t=r_t$,  then $F(\{o_t\}\backslash \{\salg_t\}\mid \alg_t)=F(\emptyset\mid \alg_t)=0$. 

	\hfill\Halmos\end{proof}

{\color{black}  Using the same proof framework, we can show that \agrd\ is  $0.5$-competitive with respect to \opt\ in the presence of stochastic outcomes. In fact, later we will establish a more general result for an adaptive (randomized) algorithm in the non-monotone setting (see Section \ref{sec:cassam} and Remark \ref{agrd0.5} in Appendix \ref{appx:cassam}). However, comparing \grd\ with \opt\ requires a different approach. Theorem \ref{advmain} only guarantees that Greedy is  $0.5$-competitive against the non-adaptive offline solution, \nopt. Specifically, inequality \eqref{fourth} in the proof does not hold when comparing Greedy with \opt\ because \opt\ has the advantage of making adaptive choices at each arrival $t\in[T]$. 

\subsection{Optimal Competitive Ratio in the Presence of Stochastic Outcomes}\label{sec:pas}

In this section, we compare Greedy with $\optc$, which is an upper bound on \opt, 
and establish the following result.

\begin{theorem}\label{concave1}
	Greedy is  $0.5$-competitive for OSOW-SO in the adversarial model. 
\end{theorem}
 {\color{black} We show this result by \emph{reducing} the comparison  	between \grd\ and $\optc$ on an instance $G$ of OSOW-SO to a comparison	between \grd\ and \nopt\ on a new instance $\gc$ of OSOW. We embed the optimal offline solution to $G$ as a feasible solution in $\gc$, while ensuring that the output of Greedy remains the same for both instances. It is fairly easy to construct such an instance using an arbitrary objective function but we show that such a construction is also possible using a sum of monotone submodular order (and monotone submodular) objective functions. Later, we apply the same reduction technique to derive new results for other arrival models and various Greedy-like algorithms. Overall, this reduction technique enables us to lift competitive ratio results from classic deterministic settings to settings with stochastic outcomes. 
 	
 	In the following lemma, we formally state the properties of $\gc$ that we use to show Theorem \ref{concave1}. In the following, we use $\aalg(G)$ to denote the value of the \grd\ solution for an instance $G$.
 	\begin{lemma}\label{redprop}
 		For each instance $G$ of OSOW-SO, there exists an instance $\gc$ of OSOW with the following properties:
 		\begin{enumerate} [$(i)$]
 			\item \textbf{Dominance:} The optimal offline value of $\gc$ is an upper bound on $\optc(G)$, i.e., $\nopt(\gc)\geq \optc(G)$.
 			\item \textbf{Invariance:} The output of Greedy is the same for both instances which implies $\aalg(G)=\aalg(\gc)$. 
 		\end{enumerate}
 	\end{lemma}
In fact, we show that $\gc$ is an instance of OSW if $G$ is an instance of OSW-SO. For the moment, we assume the correctness of Lemma \ref{redprop} and use it to show Theorem \ref{concave1}.
 	\begin{proof}{Proof of Theorem \ref{concave1}.}
 		Given an instance $G$ of OSOW-SO, consider an instance $\gc$ of OSOW as stated in Lemma \ref{redprop}. We have the following chain of inequalities:
 		\begin{eqnarray*}
 				\frac{\aalg(G)}{\opt(G)}\geq \frac{\aalg(G)}{\optc(G)}\geq \frac{\aalg(G)}{\nopt(\gc)}=\frac{\aalg(\gc)}{\nopt(\gc)}\geq 0.5.
 		\end{eqnarray*}
 	The first inequality	follows from Lemma \ref{concdom}, while the second inequality and the subsequent equality follow from the dominance and invariance properties (parts $(i)$ and $(ii)$ of Lemma \ref{redprop}), respectively. The final inequality follows from Theorem \ref{advmain}. 
 		
 		\hfill\Halmos\end{proof}

 Next, we construct instance $\gc$. We then discuss the main properties of $\gc$ and outline the keys steps in the proof of Lemma \ref{redprop}.  The full proof of Lemma \ref{redprop} is deferred to Appendix \ref{appx:pas}.

 	\subsubsection{Construction of $\gc$.} 	Recall that each instance $G$ of OSOW-SO induces an instance of OSOW with  objective function $F$ defined in \eqref{defF}. We construct $\gc$ by augmenting this instance with additional actions and  extending the objective function $F$ to this enlarged action set. The extension is designed to $(i)$ capture the value of the optimal adaptive offline solution of $G$ (the dominance property) and $(ii)$ ensure that the value attained by the Greedy algorithm is unchanged (the invariance property).

 	For each $t\in [T]$, let $\h{a}_t$ denote a new action, which represents the ``decision'' of the optimal offline at arrival $t$ in instance $G$. Let $\h{\curac}=\{\h{a}_t\}_{t\in [T]}$ be the set of these new actions, and define the expanded ground set of actions as $\no=\curac\cup \h{\curac}$.  Since the instance $\gc$ has no stochastic outcomes,  it suffices to define the objective function $F$ directly on the expanded ground set $\no$. 

 We associate each new action $\h{a}_t$ with the outcome set
 	\[N(\h{a}_t)=N_t\quad \forall t\in [T].\] 
 Note that we allow these outcome sets to overlap with those of actions in $\ac_t,$ since this construction is used solely for defining the extended objective function $F$. 
 
 We now define the values of the new actions using the optimal solution of the offline benchmark $\optc$, ensuring that $F(\h{\curac})=\optc(G)$. Recall that \[\optc(G)=\max_{Y\in\Delta(\curac)} F^c(Y),\] and let $Y^c=(y^c_a)_{a\in \curac}$ satisfy
 	  \[Y^c=\argmax_{Y\in\Delta(\curac)} F^c(Y).\] 
 	Let $\alpha^c$ denote the optimal solution 
 	to the optimization problem defining $F^c(Y^c)$, so that
 	\[\optc(G)=\fc(Y^c)=\sum_{X\in \cal{X}} \alpha^c (X)\, f(X).\]
 		For any subset $S\subseteq \no$ such that $S_1=S\cap \curac$ and $S_2=S\cap \h{\curac}$. We extend the definition of $F$ to $\no$ as
 		\begin{eqnarray}
 			F(S)&\coloneqq\, &\sum_{P_1\in \curn}\sum_{P_2\in \mathcal{X}} \gamma(P_1) \, \alpha^c(P_2)\, f\left((N(S_1)\cap P_1)\cup (N(S_2)\cap P_2)\right)\label{def1}\\
 			&=&F(S_1)+\sum_{P_1\in \curn}\sum_{P_2\in \cal{X}} \gamma(P_1) \, \alpha^c(P_2)\, f\big(N(S_2)\cap P_2\,\mid\, N(S_1)\cap P_1\big).\label{def2}
 		\end{eqnarray}	
 		The first summation ranges over all possible mappings from actions to outcomes in the original instance $G$, while the second summation ranges over all realizable outcome sets of feasible action sets. Intuitively, $F(S)$ is the expected total reward obtained by the action set $S_1\cup S_2$, where we independently sample a realized mapping $P_1$ w.p.\ $\gamma(P_1)$ and a (partial) mapping $P_2$ w.p.\  $\alpha^c(P_2)$. Here, $\gamma$ corresponds to the product distribution over outcomes, whereas $\alpha^c$ is the distribution over outcome sets induced by the optimal offline solution. Note that when $S_2=\{\emptyset\}$, this definition coincides with the original of $F$ on the ground set $\curac$ given in \eqref{defF}. 
 		
 		In equation \eqref{def2}, we used the fact that $\sum_{P_2\in \cal{X}} \alpha^c(P_2)=1$ to separate the term $F(S_1)$ from the summation. Additionally, for every set $S\subseteq \h{\curac}$ of new actions, we have, 
 		\begin{eqnarray*}
 			F(S)&=&\sum_{P_2\in \cal{X}} \alpha^c(P_2)\,\, f\left(\underset{t\mid \h{a}_t\in S}{\cup} (N_t\cap P_2)\right),
 		\end{eqnarray*}
 		where $N(S)\cap P_2=\underset{t\mid \h{a}_t\in S}{\cup} N_t\cap P_2$.
 		In particular, we have the following:
 		\begin{eqnarray}
 			\nopt(\gc)\,\geq\, F(\h{\curac})\,=\, \sum_{P\in \cal{X}} \alpha^c(P) f(P)\,=\, \optc(G),\label{optineq}
 		\end{eqnarray}
 		where the inequality follows from the fact that the set of actions $\h{\curac}$ is a feasible solution for the instance $\gc$. 
 		This completes the definition of $\gc$. Note that inequality \eqref{optineq} provides the dominance property.   
 		To further illustrate this construction, we present a numerical example. For simplicity, we use \opt\ instead of $\optc$ to construct $\gc$ in this example.
 		
\begin{example} 
	Let us revisit Example \ref{adapeg}, where we consider an instance of OSOW-SO with two arrivals $\{1,2\}$ and action sets $\ac_1=\{a_{1}\}$ and $\ac_2=\{a'_{1},a_{2}\}$. Let $e_1,e'_1,e_2$ denote the successful outcomes of actions $a_1,a'_1,a_2$, respectively, with success probabilities $2p_{e_1}=p_{e'_1}=p_{e_2}=1$. Let the rewards be defined by $f(e_1)=f(e'_1)=4f(e_2)=1$, with marginal gains $f(e'_1\mid e_1)=0$ and $f(e_2\mid e_1)=0.25$. 
	
	In this setting, Greedy selects actions $a_1$ and $a'_1$. The adaptive optimal solution $\aopt$ also selects $a_1$ at arrival 1. If $a_1$ succeeds, \opt\ selects $a_2$ at arrival 2; otherwise, it selects $a'_1$. The outcomes of the actions selected by \opt\ lead to the following distribution $\alpha^c$ over the success outcomes: $\alpha^c(\{e_1,e_2\})=0.5$ and $\alpha^c(\{e'_1\})=0.5$. 
	
	We capture this adaptive solution in a deterministic instance $\gc$ by introducing two new actions $\{\h{a}_1, \h{a}_2\}$ with the following rewards:
	\begin{align*}
		F(\h{a}_1) &= \alpha^c(\{e_1,e_2\})\, f(e_1) = 0.5, \\
		F(\{\h{a}_1,\h{a}_2\}) &= \alpha^c(\{e_1,e_2\})\, f(e_1, e_2) + \alpha^c(\{e'_1\})\, f(e'_1) = 0.5(1.25) + 0.5(1) = 1.125.
	\end{align*}
	Observe that $F(\h{a}_1)$ equals the expected reward of \opt\ at arrival 1, and $F(\{\h{a}_1,\h{a}_2\})$ equals the total expected reward of \opt.  While $\gc$ is technically a deterministic instance, its objective $F$ is constructed as an expectation over ``imaginary" stochastic outcomes.  The distribution $\alpha^c$ ensures that the imaginary outcomes of actions $\h{a}_1$ and $\h{a}_2$ are perfectly negatively correlated.
	
	In instance $\gc$, Greedy selects action $a_1$ at the first arrival because $a_1$ and $\h{a}_1$ yield the same marginal reward ($0.5$). We claim that Greedy still picks $a'_1$ at arrival 2, rendering the Greedy solution unchanged from the original stochastic instance. To see this, consider the marginal reward for choosing $\h{a}_2$ at arrival 2, given that $a_1$ was selected at arrival 1:
	\begin{align*}
		F(\h{a}_2\mid a_1) &= \alpha^c(\{e_1,e_2\}) [p_{e_1}f(e_2\mid e_1) + (1-p_{e_1})f(e_2)] + \alpha^c(\{e'_1\}) [p_{e_1}f(e'_1\mid e_1) + (1-p_{e_1})f(e'_1)] \\
		&= 0.5[0.5 \times 0.25 + 0.5 \times 0.25] + 0.5[0.5 \times 0 + 0.5 \times 1] \\
		&= 0.375 < 0.5 = F(a'_1\mid a_1).
	\end{align*}
This confirms that the output of Greedy is invariant, even though the value of the optimal adaptive offline solution for $G$ is embedded within $\gc$. Intuitively, the imaginary outcomes of $\h{a}_2$ and $a_1$ are treated as independent; consequently, the correlation between $\h{a}_1$ and $\h{a}_2$ remains "invisible" to Greedy once it ignores $\h{a}_1$ at the first arrival. 
\end{example}

 \subsubsection{Properties of $\gc$ and Proof Sketch of Lemma \ref{redprop}.} Lemma \ref{redprop} states that: $(i)$ $\gc$ is an instance of OSOW, $(ii)$ The dominance property holds, and $(iii)$ The invariance property holds. The dominance property follows directly from the definition (see \eqref{optineq}), so we focus on proving $(i)$ and $(iii)$. Below, we outline the key steps; the complete proofs are deferred to Appendix \ref{appx:pas}.  
 
 	\begin{lemma}\label{gcprop1}
 		If $G$ is an instance of OSOW-SO then $\gc$ is an instance of OSOW. Similarly, if $G$ is an instance of OSW-SO then $\gc$ is an instance of OSW. 
 		\end{lemma} 	
 The lemma above strengthens claim $(i)$. The key insight is that the function $f\left((N(\cdot)\cap P_1)\cup (N(\cdot)\cap P_2)\right)$ is monotone and admits an arrival-consistent submodular order (or is submodular). 
 	
 	To prove the invariance property, we need to show that \grd\ does not select an action from the set $\h{\curac}$ on instance $\gc$. If \grd\ chooses an action from $\ac_t$ at every $t\in [T]$, then by induction, \grd\ produces the same output on $G$ and $\gc$. The key ingredient in this proof is Lemma \ref{key}.
 	\begin{lemma}\label{key}
	For any set function $f$ (not necessarily monotone or submodular order), the function $F$ defined in \eqref{def1} satisfies, 
	\[F(\h{a}_t\mid S)\leq \max_{a\in \ac_t} F(a\mid S)
	\qquad \forall t\in [T], S\subseteq \curac\backslash \ac_t.\] 
\end{lemma}
Note that if $S\subseteq \curac$ is the set of actions selected by \grd\ prior to arrival $t$, then by applying Lemma \ref{key}, the marginal value $F(\h{a}_t\mid S)$ of action $\h{a}_t$ is (weakly) dominated by the marginal value of some action in $\ac_t$. Thus, \grd\ will choose an action from $\ac_t$ at arrival $t$.

 	We now state the two main lemmas that we use to prove Lemma \ref{key}. For all $a\in \curac,\, S\subseteq \curac\backslash \{a\},$ and $e\in N(a)$, let
 	\[w_{e,S}=\sum_{P\in \curn(S)} \gamma(P)\,f(e\mid P),\]
 	denote the marginal reward of outcome $e\in N(a)$ when action set $S$ has been selected.
\begin{lemma}\label{useF}\label{key1}
	For any set function $f$ (not necessarily monotone or submodular), every arrival $t\in [T]$, action $a\in \ac_t$, and set $S\subseteq \curac\backslash \{a\}$, the function $F$ defined in \eqref{def1} satisfies, 
	\[
	F(a\mid S)=\sum_{e\in N(a)} p_e\, w_{e,S}.
	\]
\end{lemma}

\begin{lemma}\label{key2}
	For any set function $f$ (not necessarily monotone or submodular), every arrival $t\in [T]$ and set $S\subseteq \curac\backslash \ac_t$, the function $F$ defined in \eqref{def1} satisfies, 
	\[
	F(\h{a}_t\mid S)=\sum_{a\in \ac_t} y^c_a \left(\sum_{e\in N(a)} p_e\, w_{e,S}\right).
	\]
\end{lemma}

 	Recall that $\sum_{a\in \ac_t} y^c_a=1$. Using Lemma \ref{useF} and Lemma \ref{key2}, we have that $F(\h{a}_t\mid S)=\sum_{a\in \ac_t} y^c_a\, F(a\mid S)$. This shows that the marginal value of $\h{a}_t$ is a convex combination of the marginal values of actions in $\ac_t$. From this, we can conclude Lemma \ref{key}. 

\begin{remark}[Invariance Beyond \grd]\label{beyondg}
	The invariance property is the only component of our framework that depends on the specific online algorithm. Consequently, the reduction applies to any algorithm that satisfies an analogous invariance condition with respect to the constructed instance $\gc$. In Section~\ref{sec:gla}, we establish invariance for a broad class of non-adaptive Greedy-like algorithms that, at each arrival, maximize the marginal value of a (possibly randomized) set function different from $F$. In contrast, Section~\ref{sec:cassam} presents a non-Greedy algorithm for which invariance fails.
\end{remark}

\begin{remark}[Single-Arrival Problem and Approximate-Greedy]\label{sapremark}
	Let $\alg_t$ denote the set of actions selected by \grd\ prior to arrival $t$.  
	By Lemma~\ref{key1}, at arrival $t$ the \grd\ algorithm solves the following optimization problem:
	\[
	\text{Single-Arrival Problem (\sap):}\qquad 
	\argmax_{a\in \ac_t} \sum_{e\in N(a)} p_e\, w_{e,\alg_t}.
	\]
	In some applications, the weights (marginal rewards) $w_{e,\alg_t}$ are sums of exponentially many terms, making them challenging to compute exactly. Additionally, the \sap\ problem itself can be NP-hard in certain settings. In Appendix \ref{appx:sap}, we use our reduction framework to show that an Approximate-Greedy algorithm, which uses sample average approximations (SAA) of the weights and approximately optimal solutions to the \sap, is $\tfrac{\eta}{1+\eta}-O(\epsilon)$-competitive against $\optc$, where $\eta$ denotes the approximation ratio  for the \sap\ and $\epsilon$ is a tunable parameter that is inversely proportional to the number of samples used in the SAA. As a concrete implication, in the setting of stochastic rewards with \emph{stochastic} patience, our result implies that Approximate-Greedy is $\tfrac{1}{3}-O(\epsilon)$-competitive when $\eta=0.5$. This improves upon the $0.25$-competitive guarantee for \agrd\ established by \cite{brubach2} under the same approximation ratio for the \sap.
\end{remark}

\subsection{Beyond Adversarial Arrivals}\label{sec:stochastic}
In this section, we consider two alternative (and less pessimistic) models of generating the instance, namely, the Random Order (RO) model and the Unknown IID (UIID) model. In both models, OSOW-SO reduces to OSW-SO and we show new competitive ratio results for OSW-SO by using the reduction technique to lift state-of-the-art results for OSW. 
\subsubsection{The Random Order (RO) Model.}\label{sec:rom}

In the RO model, the adversary chooses the instance $G$ but the online algorithm receives arrivals in a randomly (re)ordered sequence. Let $\sigma:[T]\to [T]$ denote a uniformly random permutation of the arrivals and let \aalg\ denote an online algorithm. At arrival $t\in [T]$, $\aalg$ receives the set of actions $\ac_{\sigma(t)}$, possible outcomes $N_{\sigma(t)}$, and probabilities $\{p_e\}_{e\in N_{\sigma(t)}}$, and chooses exactly one action from $\ac_{\sigma(t)}$. Let $G_\sigma$ denote the reordered instance  
and let $E_{\sigma}[\cdot]$ denote expectation w.r.t.\ the randomness in $\sigma$. The competitive ratio of $\aalg$ is defined as,
\[\inf_{G\in \mathcal{G}} \frac{E_{\sigma}[\aalg(G_\sigma)]}{\opt(G)}\]
 
Note that $\opt(G)=\opt(G_\sigma)$ because the order of arrivals has no impact on the offline solution. Similarly, $\optc(G)=\optc(G_\sigma)$ for every permutation $\sigma$.

Since an adversary can also generate instances using the RO model, we know from Theorem \ref{concave1} that Greedy is \emph{at least}  $0.5$-competitive for OSW-SO in the RO model. However, since the arrivals are now randomly ordered, this result may not be optimal. 
Using the reduction technique, we show that Greedy maintains the same competitive ratio for both OSW and OSW-SO, even in the RO model.  

Let $\beta_{RO}$ denote the true competitive ratio of Greedy in the RO model for OSW. \cite{rom2} showed that $\beta_{RO}\geq 0.5096$. We establish the following new result.
 \begin{theorem}\label{rom}
 	Greedy is $\beta_{RO}$-competitive, where $\beta_{RO}\geq 0.5096$, for OSW-SO in the RO model. 
 \end{theorem}
Thus, a non-adaptive algorithm like \grd\ achieves the state-of-the-art competitive ratio for OSW-SO in the RO model. The proof is included in Appendix \ref{appx:rom}. As in the adversarial case, we leverage the construction $\gc$ and its properties (dominance and invariance) to establish this result. 
 
\subsubsection{The Unknown IID (UIID) Model.}\label{sec:iid}

 Let $G$ denote an instance of OSW-SO in the adversarial model. In the UIID model, $G$ is paired with a probability distribution $D$ over the set of arrival \emph{types} $[T]$. There are $T$ arrivals in the UIID model, with each arrival having a type independently sampled from the distribution $D$. In other words, the arrival sequence is generated by drawing $T$ independent and identically distributed samples from $D$.  In the UIID model, the distribution $D$ is \emph{unknown} to the online algorithm, wheres in the KIID model, the online algorithm has knowledge of the distribution.
The set of arrivals in $G$ corresponds to the set of arrival types in the UIID model. Specifically, an arrival of type $t$ has action set $A_t$, outcome set $N_t$, and probabilities $\{p_e\}_{e\in N_t}$. Notably, if two arrivals $t$ and $t'$ share the same type, the realized outcomes of the actions selected at $t$ and $t'$ are independent.

Since multiple arrivals may have the same type, the objective function 
in the UIID model is an extension of the set function $f$ (and $F$) to the $|N|$ (and $|\curac|$) dimensional lattice $\mathbb{Z}^{|N|}_{+}$ (and $\mathbb{Z}^{|\curac|}_{+}$). We assume that $f$ is monotone on this extended domain and satisfies a well known generalization of submodularity called the diminishing returns (DR) property\footnote{From \cite{kapralov}: A function $f$ has the property of diminishing marginal returns if for any two vectors $\bm{x},\bm{y}$ such that $\bm{x}\geq \bm{y}$ (coordinate-wise) and any unit basis vector $e_t=(0,\cdots,0,1,0,\cdots,0)$, we have, $f(\bm{x}+e_t)-f(\bm{x})\leq f(\bm{y}+e_t)-f(\bm{y})$. Further, $f$ is monotone if $f(\bm{x})\geq f(\bm{y})$.}. 

Let $\bm{H}$ denote a random instance generated according to $D$, and let
$E_D[\cdot]$ denote expectation with respect to the arrival sequence. Let $\mv_G$ denote the set of all possible distributions over $[T]$. The
competitive ratio of an online algorithm $\aalg$ in the UIID model is
defined as
\[
\inf_{G\in\mathcal{G},\, D\in\mathcal{M}_G}
\frac{E_D[\aalg(\bm{H})]}{E_D[\opt(\bm{H})]}.
\]
The UIID model is weaker than the RO model: any algorithm that is
$\beta_{\mathrm{RO}}$-competitive in the RO model is also
$\beta_{\mathrm{RO}}$-competitive in the UIID model (see
Appendix~\ref{appx:ro2iid}). Moreover, \cite{kapralov} showed that Greedy
is $(1-1/e)$-competitive for OSW in the UIID model and we show that this result also extends to our generalized OSW formulation (see Appendix~\ref{appx:iiddeterm}). By extending our
reduction technique, we obtain the same guarantee for OSW-SO.
\begin{theorem} \label{uiid}
	Greedy is $(1-1/e)$-competitive for OSW-SO in the UIID model.	
\end{theorem}
The proof of Theorem~\ref{uiid} appears in
Appendix~\ref{appx:iidso}. As before, the analysis reduces instances with
stochastic outcomes to deterministic ones. A key distinction in the
UIID model is that the offline optimum depends on the realized instance;
we address this by introducing a set of auxiliary actions that capture
all possible offline solutions. Since OSW-SO strictly generalizes OSW,
this result implies that Greedy achieves the optimal competitive ratio
among all polynomial-time (even adaptive) online algorithms unless
$\mathrm{NP}=\mathrm{RP}$, including in the KIID model. We note that while
Greedy is non-adaptive with respect to outcomes, it does adapt to the
realized sequence of arrival types.


 \begin{remark}[OSOW Reduces to OSW] 
In the natural extensions of OSOW to the RO and UIID models, every
realized arrival sequence induces a submodular order for $f$.
Equivalently, $f$ is submodular, and OSOW reduces to OSW in these
stochastic arrival models.
\end{remark}
\section{Beyond Greedy} \label{sec:gla}

	While Greedy achieves the state-of-the-art guarantee for the general OSW formulation, it is outperformed in several well-studied special cases by Greedy-like algorithms (\gla s) that attain the optimal competitive ratio \citep{kvv,msvv,displayad,goel,wholepage}. 
	
	In this section, we focus on the adversarial arrival model and use the reduction technique to establish that a non-adaptive Greedy-like algorithm (\gla) achieves the optimal competitive ratio of $(1-1/e)$ in the large capacity regime for several important special cases of OSW-SO. In addition, we apply the reduction technique in  reverse to establish new upper bounds on the competitive ratio of Greedy-like algorithms for Online Submodular Maximization (OSM). While prior hardness results rule out polynomial-time algorithms with competitive ratio exceeding  $0.5$, they leave open the possibility of beating this bound with computationally inefficient methods. Using the reduction technique, we show that this is impossible for a broad class of Greedy-like algorithms. 

\subsection{Reduction Technique for Special Cases of OSW-SO} 
Let $\Omega$-SO denote a subset of instances of OSW-SO. Let $\Omega\subseteq \Omega$-SO denote the corresponding subset consisting of instances with deterministic outcomes. Let $\aalg$ be a non-adaptive online algorithm for instances in $\Omega$-SO. We say that \aalg\ has competitive ratio $\beta$ for $\Omega$ if, 
\[\inf_{G\in \Omega}\frac{\aalg(G)}{\opt(G)}= \beta.\]
Clearly, if \aalg\ is $\beta$-competitive for $\Omega$-SO then it is at least $\beta$ competitive for $\Omega$.

Suppose now that \aalg\ satisfies the invariance property. Then, by using the reduction technique, for every
 instance $G$ of OSOW-SO there exists a corresponding instance $\gc$ of OSW such that the following holds,
\begin{equation}\label{gandgc}
	\frac{\aalg(G)}{\opt(G)}\geq \frac{\aalg(\gc)}{\opt(\gc)}.
\end{equation}
If $\Omega$ is sufficiently ``large'' so that $\gc\in\Omega$, then the right-hand side of inequality \eqref{gandgc} is lower bounded by $\beta$. In this case, $\aalg$ achieves the same competitive ratio for $\Omega$ and $\Omega$-SO. However, not every problem class satisfies this largeness condition. For example, the class of online bipartite matching with stochastic rewards is not sufficiently large: if $G$ is an instance of online bipartite matching with stochastic rewards, then under our construction in Section~\ref{sec:pas}, the corresponding instance $\gc$ is not an instance of online bipartite matching. At the same time, we want $\Omega$ to be sufficiently ``small'' so that it admits a Greedy-like algorithm with competitive ratio strictly better than $0.5$. In the following, we define a problem setting that generalizes many of
the models discussed in Section~\ref{sec:prelim}, admits a Greedy-like algorithm with competitive ratio $(1-1/e)$ for deterministic outcomes, and is nevertheless large enough to ensure that $\gc \in \Omega$ for every instance
$G \in \Omega$-SO.

{\color{black}
\subsubsection*{Resource Allocation with Vector Outcomes (\wpa).} Consider a setting where we have a set of resources $I$. Resource $i\in I$ earns a reward $r_i$ per unit of used capacity, up to a maximum of $c_i$ units. Each outcome $e\in N$ is a binary vector $e=(e_i)_{i\in I}$, where $e_i\in\{0,1\}$ indicates whether outcome $e$ will consume a unit of resource $i$. Given a set of outcomes $X\subseteq N$, the total reward obtained from resource $i$ is $f_i(X)=r_i\min\{c_i,\,\sum_{e\in X} e_i\}$.  The overall objective function is the sum of rewards across all resources, $f(X)=\sum_{i\in I} f_i(X).$ We impose no restrictions on the sequence of action sets $(\ac_t)_{t\in [T]}$ or the probabilities $\{p_e\}_{e\in N}$. It is easy to see that this setting is a special case of OSW-SO. Moreover, \wpa\ generalizes both online assortment optimization and stochastic rewards with patience. 

 Let \wpa-DO denote the special case of \wpa\ where all actions have deterministic outcomes. \wpa-DO itself is a special case of the whole-page optimization problem introduced by \cite{wholepage}.  

\subsubsection*{Greedy-like Algorithms (\gla s) for \wpa-DO.} 
 We consider a family of \gla s for \wpa-DO (in the adversarial model) that we collectively refer to as \emph{Balance}. This is a class of algorithms that favor actions generating the highest \emph{perturbed rewards}. Fix an arbitrary arrival $t\in [T]$. Let $\alg_t$ denote the set of actions selected prior to arrival $t$. Since actions have deterministic outcomes, we use the shorthand $a_i=1$ to indicate that action $a$ can consume a unit of resource $i$ and define $y_i(t)=\sum_{a\in \alg_t} a_i$ as the total capacity of resource $i$ consumed prior to arrival $t$. Each algorithm in this family is characterized by a deterministic perturbation function $u(c_i,y_i(t)):\mathbb{R}_{+}^2\to \mathbb{R}_{+}$.  
The algorithm selects the following action at arrival $t$:
\[\text{Balance:}\quad \argmax_{a\in \ac_t}\, \sum_{i\in I\mid a_i=1} r_i\, u\left(c_i,y_i(t)\right),\] 
where the perturbation function $u$ is typically monotone decreasing in $\tfrac{y_i(t)}{c_i}$, ensuring that a  larger fraction of used capacity results in a smaller perturbed reward. 
In particular, setting 
\[u\left(c_i,y_i(t)\right)=1-\exp\left(\frac{y_{i}(t)}{c_i}-1\right),\]
corresponds to the original Balance algorithm introduced by \cite{msvv} for the Adwords problem, where the exponential function naturally emerged from the analysis aimed at achieving the best possible competitive ratio. Since then, Balance has since been generalized beyond Adwords. In particular, \cite{wholepage} showed that Balance, using the same perturbation function, is $(1-1/e)$-competitive for a generalization of \wpa-DO in the fractional setting, which corresponds to the large capacity regime in the limit as $c_{\min}= \min_{i\in I}c_i \to +\infty$.  

Non-adaptive Balance is the natural generalization in which the algorithm selects the action with the highest \emph{expected} perturbed reward (see Appendix \ref{appx:gla}).
Recall that, \cite{negin} showed that an \emph{adaptive} Balance algorithm is asymptotically $(1-1/e)$-competitive for online assortment optimization, which is a special case of \wpa. However, analyzing the performance of non-adaptive Balance remained an open problem\footnote{\cite{reuse} studied a different non-adaptive algorithm and established an asymptotic competitive ratio of $(1-1/e)$ for online assortment optimization with reusable resources.}. We establish the following new result for non-adaptive Balance.
\begin{theorem}\label{thm:ib}
	Any (non-adaptive) Balance algorithm has the same asymptotic competitive ratio for \wpa\ and \wpa\emph{-DO}. 
\end{theorem}	
The proof is included in Appendix \ref{appx:gla}. As a direct corollary of this result, we conclude that non-adaptive Balance with $u(c,y)=1-\exp({\tfrac{y}{c}-1})$ is asymptotically $(1-1/e)$-competitive for online assortment optimization and stochastic rewards with patience.  Note that we leave aside the computational issues associated with the algorithm but it should be possible address them using the approximation methods discussed in  Appendix \ref{appx:sap}. 

{\color{black}
\subsection{New Competitive Ratio Upper Bounds}\label{sec:upb}
\cite{kapralov} showed that, unless NP$=$RP, no \emph{polynomial time} algorithm can achieve competitive ratio exceeding  $0.5$ for OSW (in the adversarial model). This result leaves open the possibility that one could surpass the Greedy algorithm using a 
\emph{computationally inefficient} online algorithm. 
In this section, we apply our reduction technique in the reverse direction to demonstrate that for the generalized OSW formulation (OSM), even if a computationally inefficient algorithm exists that surpasses a competitive ratio of $0.5$, it cannot belong to the broad class of Greedy-like algorithms defined below. Throughout this section (only), we explicitly use the term OSM to distinguish the generalized model from the more restricted OSW formulation originally proposed by \citet{kapralov} (see Section \ref{sec:newmodels}).

\subsubsection*{Greedy-like Algorithms (\gla s) for OSM.}
Consider an instance $G$ of OSW-SO with objective functions $f$ and $F$. Without loss of generality, assume that $f$ (and hence $F$) decomposes as a sum of component functions, $f=\sum_{i\in I} f_i$, where each $f_i$ is arbitrary and need not be monotone or submodular. By Definition~\eqref{defF}, this induces a corresponding decomposition of $F$ as $F=\sum_{i\in I} F_i$, where each $F_i$ is also arbitrary.
 This decomposition allows us to formulate a broad class of \gla s that captures many classical algorithms in special cases of OSW/OSM where the objective function naturally decomposes across a set $I$ of resources. A \gla\ for OSM operates as follows: at each arrival $t$, given the set of actions $\alg_t$ selected prior to $t$, the algorithm chooses an action that maximizes the value of \emph{a sum of perturbed marginal value functions} $\tilde{F}_i$:
\[
\arg\max_{a \in \ac_t}\; \sum_{i\in I} \tilde{F}_i(a \mid \alg_t).
\]

The functions $\tilde{F}_i$ may be randomized and are not assumed to be monotone or submodular. However, to enable the application of the reduction technique, we impose the following structural restriction: for any action $a$ and set $S$, the value $\tilde{F}_i(a \mid S)$ may depend only on the marginal contribution $F_i(a \mid S)$ and on the current state of the allocation, as summarized by the values of functions $\{F_j\}_{j\in J}$ on subsets of $S$. Formally, for each $i\in I$, there exists an arbitrary (possibly randomized) function $u_i$, not necessarily efficiently computable, such that
\begin{equation}\nonumber
\tilde{F}_i(a \mid S)
=
u_i\!\left(\{F_j(X)\}_{j\in I, X \subseteq S}\right)\, F_i(a \mid S).
\end{equation}
 This restriction is crucial for applying the reduction technique: without it the invariance property may fail. Under this structure, we establish the following upper bound on the performance of \gla s for OSM.
\begin{theorem}\label{thm:upb}
	The family of \gla s defined above has competitive ratio at most  $0.5$ for  OSM.
\end{theorem}
The full proof and algorithmic details appear in Appendix~\ref{appx:upb}. The argument builds on a result of \cite{deb}, which shows that when $c_{\min}=1$, no non-adaptive algorithm can achieve a competitive ratio exceeding  $0.5$ for online matching with stochastic rewards, a special case of OSW-SO. By applying the reduction technique in the reverse direction, we then extend this limitation to obtain a universal upper bound of $0.5$ on the competitive ratio of all \gla s for OSM.  Notably, this proof does not extend to the standard OSW setting of \cite{kapralov} because the reduction technique requires the flexibility of a general action set; in contrast, the action space in the OSW formulation is restricted to edges in a bipartite matching graph between resources and arrivals.

As discussed in Appendix \ref{appx:pg1}, this family of algorithms includes the classic Perturbed Greedy algorithms \citep{kvv,goel}, that achieve the optimal competitive ratio for OBM and several of its extensions 
without requiring a large inventory assumption. While it is straightforward to show that a natural extension of Perturbed Greedy cannot surpass a competitive ratio of  $0.5$ for \wpa-DO, Theorem~\ref{thm:upb} does not directly extend to this setting (see Appendix~\ref{appx:pg2} for more details). Consequently, determining the optimal competitive ratio attainable by the class of \gla s defined above for \wpa-DO remains an open problem.

}
\section{Non-Monotone Functions}\label{sec:onsw}\label{sec:cassam}

We now relax the monotonicity assumption on $f$ and $F$ and study online non-monotone submodular welfare with stochastic outcomes (ONSW-SO), where $f$ and $F$ are non-negative submodular but not necessarily monotone\footnote{Allowing negative values precludes any non-trivial competitive ratio for non-anticipatory algorithms \citep[Section~6.4]{asadpour2016maximizing}.}. We focus on adversarial arrivals to highlight the new challenges that arise for non-adaptive algorithms.

\paragraph{Previous Work.}
To our knowledge, no competitive-ratio guarantees were previously known for ONSW-SO. For ONSW with adversarial arrivals, \citet{nuti} showed that every deterministic algorithm can have arbitrarily poor performance, necessitating randomization, and proposed a randomized algorithm with competitive ratio $0.25$, which is optimal by an upper bound of \citet{buchbinder2019online}. Earlier, \citet{harshaw2022power} analyzed a randomized Greedy variant achieving competitive ratio $3-2\sqrt{2}\approx0.17$. For a discussion of the literature on ONSW with stochastic arrivals, we refer the reader to \cite{nuti}.

\paragraph{Our Contributions.}
We introduce an adaptive randomized algorithm for ONSW-SO, called \emph{Cascade Sampling} (Algorithm~\ref{cassam}), and show that it is $0.25$-competitive in the adversarial model. Since ONSW-SO generalizes ONSW, this guarantee is optimal.

Motivated by the optimality of non-adaptive algorithms for OSOW-SO, we investigate whether similar results hold for ONSW-SO. Two obstacles arise. First, a non-adaptive version of Cascade Sampling appears incompatible with our reduction due to a failure of the invariance property. Second, for non-monotone $f$, the objective function $F$ in the constructed instance $\gc$ need not remain submodular. These issues are discussed in Appendix~\ref{appx:onswnonadap}. We also present a simple non-adaptive algorithm satisfying invariance (Algorithm~\ref{sinsam}), though determining the optimal non-adaptive algorithm for ONSW-SO remains open.

\begin{remark}
	We assume the availability of a null action $0_t\in\ac_t$ at each arrival, disjoint action sets across arrivals, and disjoint outcome sets across actions. These are natural extensions of standard assumptions in the literature \citep{nuti}, though they are not without loss of generality for non-monotone objective functions. 
\end{remark}

\subsubsection*{Adaptive Algorithm for ONSW-SO.} 
Recall that, for ONSW, obtaining a non-trivial competitive ratio guarantee requires the use of randomized algorithms. 
\begin{algorithm}
	\SetAlgoNoLine
	Set $P_1=\emptyset$ \;
	\For{every arrival $t\in[T]$}{
		Compute marginal values $v_{a}=\sum_{e\in N(a)} p_e\, f(e\mid P_t)\quad \forall a\in \ac_t$\;
		Sort and order actions with non-negative values, $v_{a_1}\geq v_{a_2}\geq \cdots\geq v_{a_{m_t}}\geq 0$ \;
		Choose action $a_{\ell}$ w.p.\ $\frac{1}{2^{\ell}}$ and choose the null action w.p.\ $1-\sum_{\ell\in[m_t]}\frac{1}{2^\ell}$\;
		Observe realized outcome $e_t$ and set $P_{t+1}=P_{t}\cup \{e_t\}$\;
	}
	\caption{Cascade Sampling (Adaptive)}
	\label{cassam}
\end{algorithm}

At each arrival, Cascade Sampling (Algorithm \ref{cassam}) considers the set of all actions with non-negative marginal values and assigns a non-zero probability of selection to each action in this set. Specifically, if there are $m_t$ ($\geq 0$) actions with non-negative marginal value at arrival $t$, then the action with the $\ell$-th highest marginal value is selected w.p.\ $\frac{1}{2^{\ell}}$. The algorithm chooses the null action w.p.\ $1-\sum_{\ell\in[m_t]}\frac{1}{2^{\ell}}$. Note that $\sum_{\ell\in[m_t]}\frac{1}{2^{\ell}}<1$ for any finite $m_t$. In the absence of stochastic outcomes, Cascade Sampling reduces to the algorithm in \cite{nuti}.

\begin{theorem}\label{nonmadapt}
	Cascade Sampling (Algorithm \ref{cassam}) is $0.25$-competitive for ONSW-SO in the adversarial model.
\end{theorem}
The proof is provided in Appendix~\ref{appx:cassam} and follows the same high-level structure as the proof of Theorem~\ref{advmain}.  The underlying ideas are straightforward generalizations of the analysis in \cite{nuti}; however, the semantics and presentation differ in important ways. In particular, our proof adheres to the unified template of Theorem~\ref{advmain}, which allows us to integrate the new analysis into a single, coherent framework that applies uniformly to both monotone and non-monotone objectives, as well as to submodular and submodular-order functions. Recall that that no online algorithm for ONSW-SO can have competitive ratio higher than $0.25$.

\begin{remark}\label{redfail}
Our reduction technique does not apply to Cascade Sampling due to a potential failure of the invariance property. Specifically, a non-adaptive version of the algorithm may output a solution with a higher objective value on the constructed instance $\gc$ than on the original instance $G$. To illustrate this, consider an instance $G$ of OSW-SO (a special case of ONSW-SO) and let $\alg_t$ denote the set of actions chosen prior to time $t$ in a non-adaptive version of Cascade Sampling. At time $t$, the algorithm considers all the actions in the set $\{a \mid F(a \mid \alg_t) > 0, a \in \ac_t\}$ and selects each action with a non-zero probability. Recall that, in our construction of the instance $\gc$, we add a new action $\h{a}_t$ at time $t$ to capture the choice of the adaptive offline algorithm. If $F(\h{a}_t \mid \alg_t) > 0$, then $\h{a}_t$ is chosen by Cascade Sampling with non-zero probability, resulting in a higher net gain in objective value on the instance $\gc$ compared to the instance $G$.
\end{remark}



		\section{Conclusion}\label{sec:conclusion}
In recent years, there has been a surge in research on online bipartite resource allocation problems that are not captured by OSW. Despite this, the \agrd\ algorithm remains  $0.5$-competitive in most of these settings. To unify these new results and develop a standard model that generalizes OSW, we introduce the OSOW-SO and OSW-SO problems. We demonstrate that non-adaptive Greedy achieves the best possible (or best known) competitive guarantees for these problems across various arrival models. Prior to our work, no results were known for non-adaptive algorithms in many of the settings we generalize. Our findings also show that, in a general setting, adaptive algorithms do not outperform non-adaptive algorithms in terms of competitive ratio (unless NP$=$RP). Additionally, we extend our results and techniques to other algorithms and present improved lower and upper bounds on the competitive ratio of Greedy-like algorithms. Finally, we turn our attention to the ONSW setting, where the objective is a non-monotone function, and discuss the challenges in generalizing our techniques. We show that an adaptive algorithm achieves the optimal competitive ratio for ONSW-SO and leave the problem of finding the optimal non-adaptive algorithm as an open question.

	\section*{Acknowledgements}
We thank the anonymous referees for their extremely insightful feedback. We also thank Ali Aouad, Omar El Housni, Adam Elmachtoub, Daniel Freund, Will Ma, Calum MacRury, Daniela Saban, and Ziv Scully for valuable discussions related to this paper.

{\small
	\bibliographystyle{informs2014.bst}
\bibliography{new}
}
\begin{APPENDICES}
		\section{Missing Proofs from Section \ref{sec:newmodels}}\label{appx:newmodels}

	\begin{repeatlemma}[Lemma \ref{concdom}.]
		For every instance $G$ of OSOW-SO (and OSW-SO) we have $\optc(G)\geq \opt(G).$
	\end{repeatlemma}
	\begin{proof}{Proof.} 	
		To show that $\optc(G)\geq \opt(G)$, it suffices to show that there exists a vector $Y=(y_a)_{a\in \curac}\in \Delta(\curac)$, and a distribution $\alpha$ over subsets of $\cal{X}$, such that,
		\[\sum_{X\in \cal{X}}\alpha(X)f(X)=\opt(G),\quad \sum_{X\in \cal{X}}\alpha(X)=1,\quad \text{ and } \sum_{X\in \cal{X}\mid X \ni e} \alpha(X)=p_e\,y_a\,\, \forall e\in N(a),\, a\in \curac.\]
		
		Let $S$ denote the (random) set of actions chosen by \opt\ and let $Z\in \cal{X}$ denote the set of realized outcomes of $S$. Note that $Z\subseteq N(S)$. We define $y_a$ as the probability that action $a\in S$. Then, for each $e\in N(a)$, the probability that $e\in Z$ is given by $y_a\, p_e$. This follows from the fact that $\opt$ is non-anticipatory, meaning it chooses action $a$ with probability (w.p.) $y_a$ without knowing the outcome, and the outcome $e$ occurs w.p.\ $p_e$ independent of other events. 
		
		Now, define $\alpha(X)$ as the probability that $Z=X$.  With this definition, we have,
		\[\sum_{X\in \cal{X}}\alpha(X)=1\quad \text{ and }\quad \sum_{X\in \cal{X}\mid X \ni e}\alpha(X)=p_ey_a\,\, \forall e\in N(a),\, a\in \curac,\]
		which satisfies the desired conditions. Finally, since \opt\ has total reward $f(X)$ on sample paths where $Z=X$, we conclude that, $\opt(G)=\sum_{X\in \cal{X}} \alpha(X)\, f(X)$. 
		
		\hfill\Halmos\end{proof}
{\color{black}
		\begin{lemma}\label{ftoF2}
	If $f$ is monotone and has an arrival-consistent submodular order then, for every $P\in \curn$ and non-negative real value $\lambda$, the function  $\lambda\;f(N(\cdot)\cap P)$ on ground set $\curac$ is monotone with an arrival-consistent submodular order.
\end{lemma}
\begin{proof}{Proof.} By Remark 1 in \cite{SOF}, multiplying a monotone function with a submodular-order by a non-negative scalar preserves both monotonicity and the submodular-order property. Therefore, it suffices to focus on the function $f(N(\cdot)\cap P)$. Let $\pi$ denote the arrival-consistent submodular order of $f$.
	Fix an arbitrary $P\in \curn$ and let $\psi(S)=N(S)\cap P$. Since $\psi$ is an injective mapping from $\curac$ to $N$, there is a unique arrival-consistent order $\pi_{\curac}$ over $\curac$ such that for any two actions $a,a'$, we have  $a\succ_{\pi_{\curac}}a'$ if and only if $N(a)\cap P\succ_\pi N(a')\cap P$. More strongly, for any two action sets $B$ and $C$, we have $C\succ_{\pi_{\curac}} B$ if and only if $\psi(C)\succ_{\pi} \psi(B)$.
	
	Now, $f(\psi(\cdot))$ is monotone because $\psi(B)\subseteq \psi(B^+)$ for all $B\subseteq B^+\subseteq \curac$. And $\pi_{\curac}$ is an arrival-consistent submodular order for $f(\psi(\cdot))$ because for all $\pi_{\curac}$-nested sets $B\subseteq B^+$ and all sets $C\succ_{\pi_{\curac}} B^+$, we have that, $\psi(B)$ and $\psi(B^+)$ are $\pi$-nested and $\psi(C)\succ_{\pi} \psi(B^+)$. Thus,
	\[f(\psi(B^+\cup C))-f(\psi(B^+))=f(\psi(C)\mid \psi(B^+))\;\leq\; f(\psi(C)\mid \psi(B))=f(\psi(B\cup C))-f(\psi(B)),\]
	here the inequality follows from the submodular order property of $f$. 
	
	\hfill\Halmos		\end{proof}}
\section{Special Cases of OSOW-SO}\label{appx:prelim}
We begin by proving that the OSW-SO formulation captures all of the settings we discussed in Section \ref{sec:prelim}, with the exception of OBM with reusable resources, which, as shown in Appendix \ref{appx:reuse}, is a special case of OSOW. 

To establish this, we define an equivalent instance of OSW-SO for each setting. We start by describing the common features shared by these instances. Recall that all the settings we are considering involve a set of resources $I$. We focus on instances of OSW-SO where the objective function $f$ is a sum of monotone submodular functions $\{f_i\}_{i\in I}$. Each function $f_i$ is defined using a subset $N_i\subseteq N$ of outcomes, with only elements in $N_i$ contributing non-zero values to $f_i$. The sets $N_i$ may overlap for different resources. We now define the specific instance for each setting. 

\emph{Online matching with stochastic rewards:} For each arrival $t\in[T]$, the action set $\ac_t$ is the set of all edges incident to $t$. Action $(i,t)$ has two possible outcomes: success or failure. The outcome of success is represented by $e_{i,t,s}$, which is included in the set $N_i$. Given a subset of outcomes $X\subseteq N$, the reward from resource $i$ is captured by the monotone submodular function $f_i(X)=r_i\,\min\{c_i, |X\cap N_i|\}$. The overall objective function, $f=\sum_{i\in I} f_i$, is also a monotone submodular function.

\emph{Stochastic rewards with patience:} In this setting, a feasible action at arrival $t$ is an ordered sequence of resources to attempt matching with arrival $t$. Let $\ac_t$ denote the set of all feasible actions. Each action $a\in \ac_t$ has a vector outcome $e_{a,t}=(e_{a,i,t})_{i\in I}$, where $e_{a,i,t}=1$ if resource $i$ is successfully matched to $t$ and $e_{a,i,t}=0$ otherwise. If $e_{a,i,t}=1$, then we include outcome $e_{a,t}$ in the set $N_i$. The reward functions $f_i$ and objective function $f$ are the same as in the previous case. 

\emph{Online assortment optimization:} The action set $\ac_t$ consists of the assortments that can be shown to arrival $t\in [T]$. Similar to the previous case, each action $a\in \ac_t$ has a vector outcome $e_{a,t}=(e_{a,i,t})_{i\in I}$, where $e_{a,i,t}=1$ if arrival $t$ chooses resource $i$ and $e_{a,i,t}=0$ otherwise. If $e_{a,i,t}=1$, then we include outcome $e_{a,t}$ in the set $N_i$. The reward functions $f_i$ and objective function $f$ are the same as in the previous two cases.

\emph{Two-sided assortment optimization:} The set of actions and outcomes is the same as online assortment optimization. However, the reward function $f_i(X)$ now represents the probability that resource $i$ chooses at least one element from the set $X$. As noted by \cite{aouad2020online}, this is a monotone submodular function for a large variety of discrete choice models.

\begin{remark}[Scalar outcomes do not suffice]\label{scalar}
	Recall that one way to represent each outcome in assortment optimization is using a binary vector $e=(e_i)_{i\in I}$ where $e_i=1$ if $i$ is selected by the customer and $e_i=0$ otherwise. In \cite{asadpour2016maximizing}, every outcome is a scalar and we sketch an example to show that a monotone function $f$ on scalar outcomes cannot capture the assortment problem. Consider two items $i$ and $j$, with unit capacity and per unit rewards $r_i=1$ and $r_j=2$. Consider two arrivals with feasible assortments (actions) $U_1=\{j\}$ at arrival 1 and $U_2=\{i,j\}$ at arrival 2. The assortment (action) $U_2$ has three possible outcomes and we can capture them using a random variable $u_2\in[0,1]$ such that, $u_2=\delta_{2,i}$ when item $i$ is chosen, $u_2=\delta_{2,j}$ when item $j$ is chosen, and $u_2=0$ when the outside option is selected. Similarly, let $u_1=\delta_{1,j}$ when item $j$ is selected and $u_1=0$ otherwise. Consider a monotone and DR submodular objective function $f$ on $[0,1]^2$. We have that $f(\{0,\delta_{2,i}\})=r_i$ and $f(\{0,\delta_{2,j}\})=r_j$. Since $r_j>r_i$, from the monotonicity of $f$ we have that $\delta_{2,j}>\delta_{2,i}$. Further, $f(\{\delta_{1,j},\delta_{2,j}\})=r_j$ (item $j$ has unit capacity) and $f(\{\delta_{1,j},\delta_{2,i}\})=r_j+r_i$. But this violates monotonicity of $f$ because $f(\{\delta_{1,j},\delta_{2,i}\})>f(\{\delta_{1,j},\delta_{2,j}\})$ even though $\delta_{2,j}>\delta_{2,i}$.
	
	\end{remark}

\subsection{Online Matching with Reusable Resources}\label{appx:reuse}
As mentioned in Section \ref{sec:prelim}, a reusable resource is used/rented by an arrival for some time and then returned back to the system. A reward is generated each time the resource is rented to a new arrival. Specifically, suppose that arrival $t\in[T]$ arrives at \emph{time} $a(t)\in[0,1]$. At each moment in $[0,1]$, every resource is either \emph{available} or \emph{unavailable}, with all resources initially available at time 0. If resource $i$ is available at time $a(t)$, matching it to arrival $t$ generates a reward $r_i$ and $i$ becomes unavailable for a fixed duration $d_i$. Thus, $i$ is unavailable during the interval $(a(t),a(t)+d_i)$. Matching an arrival to an unavailable resource has no effect and generates no reward. The objective is to maximize the total reward. Notably, the classic online bipartite matching is a special case of this setting where $d_i\to +\infty\,\, \forall i\in I$.
	
  Perhaps surprisingly, we find that OSW cannot capture this problem, even for instances with a \emph{single} reusable resource.
	
\begin{example}[OSW Fails to Capture Reusability] 
	Consider an instance with a single reusable resource that is used for a duration of $1.5$ time units after each match and during this time the resource is unavailable. 
	Arrivals occur at times $1, 2,$ and  $3$. If arrival 1 is matched, the resource is used from time 1 to 2.5 and returns prior to arrival 3's arrival. Arrivals 1 and 3 can both be matched to the resource and arrival 2 can be matched only if the other two arrivals are unmatched. Now, consider a function $f$ such that $f(\{1,2\})=f(\{2,3\})=f(1)=f(2)=f(3)=1$ and $f(\{1,3\})=f(\{1,2,3\})=2$. This function captures the total number of matches in every possible allocation of arrivals to the resource. 
	Clearly, $f$ is not a submodular function ($f(1\mid\{2,3\})>f(1\mid\{2\})$).
\end{example}

We show that OSOW captures OBM with reusable resources (OBMR). Recall that OBMR, a unit of resource $i$ that is matched to arrival $t$ at time $a(t)$ is rented/used for a fixed duration $d_i$ and returned at time $a(t)+d_i$. Given a set of arrivals $S$ that have an edge to $i$, consider the process where we start with the first arrival in $S$ and match $i$ to every arrival in $S$ where $i$ is available. We refer to this as the \emph{matching process} on $S$ \footnote{This is the deterministic counterpart of the $(F,\sigma)$-random process introduced in \cite{reuse}.}. Let $n_i(S)$ denote the total number of arrivals matched to $i$ in this process. Observe that $\sum_{i\in I} n_i(S_i)$ is the total reward of an allocation $\{S_i\}_{i\in I}$ of arrivals to resources in OBMR. It is not hard to see that $n_i$ is a monotone function. In fact, this is a direct corollary of Lemma 5 in \cite{reuse}. We show that the arrival order is a submodular order for $n_i$. 

\begin{lemma}\label{reuse}
	In OBMR, for every resource $i\in I$, the arrival order is a submodular order for $n_i$.
\end{lemma}
\begin{proof}{Proof.}
	Let $\pi$ denote the arrival order over the ground set of arrivals. Consider $\pi$-nested sets $B\subseteq A\subseteq [T]$ and a set $C$ that succeeds $A$ in the arrival order. For $S\in \{A,B\}$, let $t_S$ be the last arrival in $S$ that is matched to $i$ in the matching process on $S\cup C$. Let $t_{C\mid S}$ denote the first arrival in $C$ that is matched in the process on $S\cup C$ and let $t_{C\mid S}=T+1$ if no such arrival exists. Observe that $n_i(C\mid S)=n_i(S\cup C)-n_i(S)$ is the number of arrivals in $C$ that are matched to $i$ in the matching process on $S\cup C$. Due to the time nesting of $A$ and $B$, we have, $t_A\geq t_B$. Thus, 
	$t_{C\mid A}\geq t_{C\mid B}.$  Now, $n_i(C\mid B)\geq n_i(C\mid A)$ follows from the monotonicity of $n_i$. 
	\hfill\Halmos\end{proof}
\subsubsection*{Beyond Deterministic Reusability.} Prior work also considers a more general version of reusability where the usage durations are stochastic. Specifically, when a resource $i$ is rented to any arrival $t$, it is used for a random duration $d_{t}$ sampled independently from a distribution $D_i$. \cite{ms} showed that an adaptive Greedy algorithm is  $0.5$-competitive in this setting. 

It is possible to extend Lemma \ref{reuse} to show that the arrival order is a submodular order for the expected reward function $E_{D_i}[n_i(\cdot)],$ where the expectation is taken over the random usage durations. Then, using our result for OSOW-SO, we can show that Greedy is  $0.5$-competitive for settings with stochastic usage durations against an offline benchmark that is \emph{oblivious} to the realizations of usage durations. Extending our techniques to compare against the duration adaptive offline benchmark appears challenging because the arrival order is not a submodular order for every sample path of usage durations.

Nonetheless, OSOW can capture generalizations of OBMR. For example, given a monotone submodular function $h_i$, the function $f_i=n_i+h_i$ is monotone and the arrival order is a submodular order. The literature on online allocation of reusable resources also includes several variations of the model that we consider here \citep{ dickerson, RST18, owen, levirad, feng2, rodilitz, baek, baek2, orestis, zeyu, wangchi, radnew, sumida2024dynamic}. It would be interesting to see if OSOW captures these settings.

\subsection{Stronger Benchmark for Stochastic Rewards with Patience}\label{appx:patience}
In the setting of stochastic rewards with patience, the offline benchmark described in Section~\ref{sec:offline} may visit arrivals in any (adaptive) order, but it is required to complete all match attempts for a given arrival before moving to the next one. \cite{borodin2022prophet} consider a stronger benchmark for this problem, which can probe edges in any adaptive order. For example, their benchmark may attempt to match one arrival, move to attempt matching another arrival, and then return to the first arrival to try a different resource. We show that the benchmark 
$\optc$ also serves as an upper bound for this stronger benchmark. Consequently, all our competitive ratio results for the stochastic rewards with patience setting remain valid against the strongest benchmark in the existing literature. Let $\opt^+(G)$ denote the value of the benchmark of \cite{borodin2022prophet} on instance $G$.

	\begin{lemma}
	For every instance $G$ of stochastic rewards with patience we have $\optc(G)\geq \opt^+(G).$
\end{lemma}
\begin{proof}{Proof.} 	
Consider an arbitrary arrival \(t\), and let \(\omega_{-t}\) denote the random realizations of all edges not incident to \(t\). Conditioned on \(\omega_{-t}\), let \(a(\omega_{-t})\) denote the action selected by \(\opt^{+}\) at arrival \(t\) in the event that all edges incident to \(t\) fail. This action is \emph{maximal} in the sense that it fully specifies the sequence of attempts made to match arrival \(t\). More precisely, since \(\opt^{+}\) is non-anticipatory, conditioned on \(\omega_{-t}\) it attempts to match arrival \(t\)—possibly interleaving these attempts with attempts for other arrivals—according to the sequence of edges encoded by \(a(\omega_{-t})\), until either a successful match occurs or arrival \(t\) exhausts its patience.

Let \(\sigma_{\omega_{-t}}\) denote the induced ordering of resources attempted for arrival \(t\) under action \(a(\omega_{-t})\). Conditioned on \(\omega_{-t}\) and on arrival \(t\) having sufficient patience, the probability that the realized outcome of \(a(\omega_{-t})\) is a successful match to resource \(i\) (i.e., \(e_{a(\omega_{-t}),i,t}=1\)) is
\begin{equation}\label{prob}
p_{i,t}\prod_{j \prec_{\sigma_{\omega_{-t}}} i} (1-p_{j,t}).
\end{equation}

An analogous characterization applies to any online algorithm. Consequently, without loss of generality, we interpret each action in \(\mathcal{A}_t\) as a \emph{maximal} action at arrival \(t\). Under this interpretation, if \(\opt^{+}\) selects an action \(a \in \mathcal{A}_t\) with probability \(y_a\), then the probability that a particular outcome \(e \in N(a)\) is realized is \(y_a p_e\). Here, \(p_e\) denotes the conditional probability that action \(a\) results in outcome \(e\), given by the product of the probability that arrival \(t\) has sufficient patience and the corresponding expression in~\eqref{prob}. Importantly, the coefficient \(y_a\) captures the randomness induced by \(\omega_{-t}\), i.e., the realizations of all edges not incident to \(t\).

The remainder of the proof follows the same structure as the proof of Lemma~\ref{concdom}, which we reproduce here for completeness. It suffices to show that there exist a vector \(Y = (y_a)_{a \in \mathcal{A}} \in \Delta(\mathcal{A})\) and a distribution \(\alpha\) over subsets of \(\mathcal{X}\) such that
\[
\sum_{X \in \mathcal{X}} \alpha(X) f(X) = \opt^{+}(G), \qquad 
\sum_{X \in \mathcal{X}} \alpha(X) = 1, \qquad
\sum_{X \in \mathcal{X} : e \in X} \alpha(X) = p_e y_a 
\quad \forall\, e \in N(a),\, a \in \mathcal{A}.
\]

Let \(S\) denote the (random) set of actions chosen by \(\opt^{+}\), and let \(Z \in \mathcal{X}\) denote the corresponding set of realized outcomes; note that \(Z \subseteq N(S)\). Define \(y_a\) as the probability that action \(a\) is selected, i.e., \(y_a = \Pr[a \in S]\). As observed above, for each outcome \(e \in N(a)\), the probability that \(e \in Z\) is \(y_a p_e\).

Finally, define \(\alpha(X)\) as the probability that \(Z = X\). By construction, we have
\[
\sum_{X \in \mathcal{X}} \alpha(X) = 1 
\quad \text{and} \quad 
\sum_{X \in \mathcal{X} : e \in X} \alpha(X) = p_e y_a 
\quad \forall\, e \in N(a),\, a \in \mathcal{A},
\]
which satisfies the desired conditions. Since \(\opt^{+}\) receives total reward \(f(X)\) on sample paths where \(Z = X\), we conclude that
\[
\opt^{+}(G) = \sum_{X \in \mathcal{X}} \alpha(X)\, f(X).
\]

	\hfill\Halmos\end{proof}
{\color{black}
\section{Missing Proofs from Section \ref{sec:osow}}\label{appx:osow}

\begin{repeatlemma}[Lemma \ref{SOineq}.]
		Consider a collection of monotone functions \(\{F_1,\cdots,F_u\}\) defined on the ground set \(\mathcal{A}\). Suppose that each function admits an arrival-consistent submodular order (not necessarily the same order across functions). 	 Then, the function $F \coloneqq \sum_{j\in[u]} F_j$ satisfies the following inequality:
	\[
	F(\opt_{T+1} \cup \alg_{T+1})
	\;\le\;
	F\!\left(\alg_{T+1}\right)
	\;+\;
	\sum_{t\in[T]}
	F\!\left(\{o_t\}\setminus\{\salg_t\}
	\;\middle|\;
	\cup_{\tau\in[t-1]} \{\salg_\tau\}\right).
	\]
\end{repeatlemma}
\begin{proof} {Proof.} The main ingredient in this proof is Lemma \ref{singfunc} (shown later), which essentially follows from Corollary 2 in \cite{SOF}. For completeness, the lemma is stated and proved after the conclusion of this proof.
	Observe that for all $A,B\subseteq \curac,$
	\[F(A\mid B)=\sum_{j\in[u]} F_j(A\mid B).\]
	Applying Lemma \ref{singfunc} to each function $F_j$, we obtain,
	\begin{eqnarray*}
		F(\opt_{T+1} \cup \alg_{T+1})&=&\sum_{j\in[u]} F_j(\opt_{T+1} \cup \alg_{T+1}),\\
		&\leq & \sum_{j\in[u]} \left[F_j\!\left(\alg_{T+1}\right)
		\;+\;
		\sum_{t\in[T]}
		F_j\!\left(\{o_t\}\setminus\{\salg_t\}
		\;\middle|\;
		\cup_{\tau\in[t-1]} \{\salg_\tau\}\right)\right],\\
		&=& F\!\left(\alg_{T+1}\right)
		\;+\;
		\sum_{t\in[T]}
		F\!\left(\{o_t\}\setminus\{\salg_t\}
		\;\middle|\;
		\cup_{\tau\in[t-1]} \{\salg_\tau\}\right).
	\end{eqnarray*}
	\hfill\Halmos	\end{proof}
\begin{lemma}\label{singfunc}
	Given a monotone submodular function $F$ with an arrival-consistent submodular $\pi_{\curac}$ on ground set $\curac$, and subsets $\opt_{T+1} = \{o_1,\ldots,o_T\}$ and  
	$\alg_{T+1} = \{\salg_1,\ldots,\salg_T\}$ of $\curac$ such that $\{o_t,\salg_t\}\subseteq \ac_t\,\, \forall t\in [T]$,  we have,
	\[
	F(\opt_{T+1} \cup \alg_{T+1})
	\;\le\;
	F\!\left(\alg_{T+1}\right)
	\;+\;
	\sum_{t\in[T]}
	F\!\left(\{o_t\}\setminus\{\salg_t\}
	\;\middle|\;
	\cup_{\tau\in[t-1]} \{\salg_\tau\}\right).
	\]
\end{lemma}
\begin{proof} {Proof.}
	Let $\opt_{t+1}=\cup_{\tau\in[t]} \{o_\tau\}$ and $\alg_{t+1}=\cup_{\tau\in[t]} \{\salg_\tau\}$.  To prove the lemma we first establish the following key inequality for all $t\in[T]$.
	\begin{eqnarray}
		&&F\!\left(\{o_{t},\salg_{t},\cdots,o_{T},\salg_{T}\}	\;\middle|\;
		\alg_{t} \right)\nonumber\\
		&&\leq F(\{\salg_{t}\}\mid \alg_{t})\;+\; F(\{o_{t}\}\backslash \{\salg_{t}\}\mid \alg_{t})\;+\; F\!\left(\{o_{t+1},\salg_{t+1},\cdots,o_{T},\salg_{T}\}	\;\middle|\;
		\alg_{t+1} \right).\label{genud2}
	\end{eqnarray} 
	Then, using induction we establish that,
	\begin{eqnarray}
		&&F(\opt_{T+1} \cup \alg_{T+1})\nonumber\\
		&&\;\le\;
		F\!\left(\alg_{t+1}\right)
		\;+\; F\!\left(\{o_{t+1},\salg_{t+1},\cdots,o_{T},\salg_{T}\}	\;\middle|\;
		\alg_{t+1} \right) \;+\;
		\sum_{\tau\in[t]}
		F\!\left(\{o_\tau\}\setminus\{\salg_\tau\}
		\;\middle|\;
		\alg_{\tau}\right)
		\quad \forall t\in[T].\label{geud}
	\end{eqnarray}
	The main claim follows from \eqref{geud} when $t=T$. 
	
	Before proceeding with the proof we observe that the action $o_t$ could precede, succeed, or equal $\salg_t$. To avoid a case-wise approach we define sets $\{o^-_t,\salg_t,o^+_t\}$ where 
	\[o^-_t=\begin{cases}
		&o_t \quad \text{if } o_t\prec_{\pi_{\curac}} \salg_t,\\
		&\emptyset \quad \text{otherwise},
	\end{cases}\quad \text{and }\quad o^+_t=\begin{cases}
		&o_t \quad \text{if } \salg_t\prec_{\pi_{\curac}} o_t,\\
		&\emptyset \quad \text{otherwise}.
	\end{cases}\]
	Observe that when $o^-_t=o_t$, we have $o^-_t\prec_{\pi_{\curac}} \salg_t \prec_{\pi_{\curac}} \{o_{t+1},\salg_{t+1},\cdots,o_{T},\salg_{T}\}$. Similarly, when $o^+_t=o_t$, we have $ \salg_t\prec_{\pi_{\curac}}o^+_t\prec_{\pi_{\curac}} \{o_{t+1},\salg_{t+1},\cdots,o_{T},\salg_{T}\}$. Note that $o^-_t=o^+_t=\emptyset$ when $o_t=\salg_t$. Recall that, $F(\{\emptyset\}\mid S)=0$ for all $S\subseteq \curac$.

	We now show inequality \eqref{genud2} for all $t\in[T]$.
	\begin{eqnarray*}
		F\!\left(\{o_{t},\salg_{t},\cdots,o_{T},\salg_{T}\}	\;\middle|\;
		\alg_{t} \right) &=& F\!\left(\{o^-_{t}\}\mid \alg_{t}\right)+F\left(\{\salg_{t},o^+_{t}\}\cup\{o_{t+1},\salg_{t+1},\cdots,o_{T},\salg_{T}\}	\;\middle|\;
		\alg_{t}\cup\{o^-_{t}\} \right),\\
		&\leq&F\!\left(\{o^-_{t}\}\mid \alg_{t}\right)+F\left(\{\salg_{t},o^+_{t}\}\cup\{o_{t+1},\salg_{t+1},\cdots,o_{T},\salg_{T}\}	\;\middle|\;
		\alg_{t} \right),\\
		&= &F\!\left(\{o^-_{t}\}\mid \alg_{t}\right)+F\left(\{\salg_{t}\}	\;\middle|\;
		\alg_{t} \right)+F\left(\{o^+_{t}\}	\;\middle|\;
		\alg_{t}\cup\{\salg_{t}\} \right)\\
		&&+F\left(\{o_{t+1},\salg_{t+1},\cdots,o_{T},\salg_{T}\}	\;\middle|\;
		\alg_{t}\cup \{\salg_{t},o^+_{t}\} \right),\\ 
		&\leq &F\!\left(\{o^-_{t}\}\mid \alg_{t}\right)+F\left(\{\salg_{t}\}	\;\middle|\;
		\alg_{t} \right)+F\left(\{o^+_{t}\}	\;\middle|\;
		\alg_{t} \right)\\
		&&+F\left(\{o_{t+1},\salg_{t+1},\cdots,o_{T},\salg_{T}\}	\;\middle|\;
		\alg_{t+1} \right),\\ 
		&= &F\left(\{\salg_{t}\}	\;\middle|\;
		\alg_{t} \right)+F\left(\{o_{t+1},\salg_{t+1},\cdots,o_{T},\salg_{T}\}	\;\middle|\;
		\alg_{t+1} \right)+F\left(\{o_{t}\}\backslash\{\salg_{t}\}	\;\middle|\;
		\alg_{t} \right),
	\end{eqnarray*}
	here the two inequalities follow from the submodular order property of $F$ applied to the $\pi_{\curac}$-nested pairs $(\alg_{t}, \alg_{t}\cup\{o^-_{t}\})$, $(\alg_{t}, \alg_{t}\cup\{\salg_{t}\})$, and $(\alg_{t+1},\alg_{t}\cup \{\salg_{t},o^+_{t}\})$. The final identity follows from the observation that $F(\{o_{t}\}\backslash \{\salg_{t}\}\mid \alg_{t})=F(\{o^-_{t}\}\mid \alg_{t})+F(\{o^+_{t}\}\mid \alg_{t})$. 
	
	Now, using \eqref{genud2} for $t=1$ we have,
	\[F\!\left(\opt_{T+1}\cup \alg_{T+1} \right)\leq F(\{\salg_{1}\})\;+\; F(\{o_{1}\}\backslash \{\salg_{1}\})\;+\; F\!\left(\{o_{2},\salg_{2},\cdots,o_{T},\salg_{T}\}	\;\middle|\;
	\{\salg_1\} \right),\]
	which is exactly inequality \eqref{geud} for $t=1$. Suppose that \eqref{geud} is true for all $t\leq t_0-1$. Then, by induction,
	\begin{eqnarray*}
		F(\opt_{T+1} \cup \alg_{T+1})&\le &
		F\!\left(\alg_{t_0}\right)
		\;+\; F\!\left(\{o_{t_0},\salg_{t_0},\cdots,o_{T},\salg_{T}\}	\;\middle|\;
		\alg_{t_0} \right) \;+\;
		\sum_{\tau\in[t_0-1]}
		F\!\left(\{o_\tau\}\setminus\{\salg_\tau\}
		\;\middle|\;
		\alg_{\tau}\right),\\
		&\leq & 	F\!\left(\alg_{t_0+1}\right)
		\;+\; F\!\left(\{o_{t_0+1},\salg_{t_0+1},\cdots,o_{T},\salg_{T}\}	\;\middle|\;
		\alg_{t_0+1} \right) \;+\;
		\sum_{\tau\in[t_0]}
		F\!\left(\{o_\tau\}\setminus\{\salg_\tau\}
		\;\middle|\;
		\alg_{\tau}\right),
	\end{eqnarray*}
	here the first inequality corresponds to inequality \eqref{geud} for $t=t_0-1$ and the second inequality follows from inequality \eqref{genud2} for $t=t_0$.
	
	\hfill\Halmos\end{proof}
}

\section{Proof of Lemma \ref{redprop}}\label{appx:pas}
Lemma \ref{redprop} states that: $(i)$ $\gc$ is an instance of OSOW, $(ii)$ The dominance property holds, and $(iii)$ The invariance property holds. The dominance property follows directly from the definition (see \eqref{optineq}), so we focus on proving $(i)$ and $(iii)$. 
\subsection{Proof of Claim $(i)$}\label{appx:gcprop1}
\begin{repeatlemma}[Lemma \ref{gcprop1}.]
	If $G$ is an instance of OSOW-SO then $\gc$ is an instance of OSOW. Similarly, if $G$ is an instance of OSW-SO then $\gc$ is an instance of OSW. 
\end{repeatlemma} 

\begin{proof}{Proof.} The lemma above strengthens claim $(i)$. 
		Recall that $F(S_1\cup S_2)$, as defined in \eqref{def1}, is a non-negative linear combination of functions $f((N(S_1)\cap P_1) \cup (N(S_2)\cap P_2))$. To establish that $\gc$ is an instance of OSOW, it suffices to show that $f((N(S_1)\cap P_1) \cup (N(S_2)\cap P_2))$ is monotone with an arrival-consistent submodular order on $\curac$. By Remark 1 in \cite{SOF}, multiplying a monotone function with a submodular-order by a non-negative scalar preserves both monotonicity and the submodular-order property.
		
		When $G$ is an instance of OSW-SO, i.e., $f$ is submodular, it suffices to show that $f((N(S_1)\cap P_1) \cup (N(S_2)\cap P_2))$ is also submodular. This implies that $\gc$ is an instance of OSW using the fact that non-negative linear combination of monotone submodular functions is also monotone submodular. 

		Fix arbitrary sets $P_1\in \curn$ and $P_2\in \mathcal{X}$ and define $\psi:2^{\no}\to 2^{N}$ as follows,
	\begin{equation}\label{psidef}
		\psi(S)= (N(S\cap \curac)\cap P_1) \cup (N(S\cap \h{\curac})\cap P_2).
	\end{equation}
{\color{black}	By definition of $\curn$ and $\mathcal{X}$, $\psi(\{a\})=N(a)\cap P_1$ (for $a\in \curac$) and $\psi(\{\h{a}_t\})=N_t\cap P_2$ are both singleton sets that represent the realized outcomes of $a$ and $\hat{a}$ respectively.}
	We will show that $f(\psi(\cdot)):2^{\no}\to \mathbb{R}_{+}$ is a monotone submodular (order) function when $f$ is a monotone submodular (order) function. We first show that $f(\psi(\cdot))$ is monotone. Then, we show that $f(\psi(\cdot))$ is submodular when $f$ is submodular. Finally, we show that if $f$ has an arrival-consistent submodular order on $N$ then $f(\psi(\cdot))$ has an arrival-consistent submodular order on $\curac$. 

	
	\paragraph{Monotonicity:} Suppose that $f$ is monotone. Observe that,
	\begin{equation}\label{psialt}
		\psi(S)=\cup_{a\in S} \psi(a)\quad \forall S\subseteq \no.
	\end{equation}
	Thus, $f(\psi(\cdot))$ is monotone because $f$ is monotone and $\psi(A)\supseteq \psi(B)$ for all $B\subseteq A$. 
	\smallskip
	
	\paragraph{Submodularity:} Suppose that $f$ is a monotone submodular function. Consider subsets $B\subseteq A\subseteq \no$, and set $C\subseteq \no\backslash A$. We have,
	\begin{eqnarray*}
		f(\psi(A\cup C))-f(\psi(A))&=&f(\psi(A\cup C)\backslash \psi(A) \mid \psi(A)),\\
		&\leq& f(\psi(A\cup C)\backslash \psi(A) \mid \psi(B)),\\
		&=&	f(\psi(C)\backslash \psi(A) \mid \psi(B)),\\
		&\leq&	f(\psi(C)\backslash \psi(B) \mid \psi(B)),\\
		&= & f(\psi(B\cup C)\backslash \psi(B) \mid \psi(B)),\\
		&=& 	f(\psi(B\cup C))-f(\psi(B)).
	\end{eqnarray*}
	The first inequality follows from submodularity of $f$. The second inequality follows from the monotonicity of $f$. The second and third equalities follow from \eqref{psialt}. Overall, this proves that $f(\psi(\cdot))$ is submodular. 
	
	Note that we use the submodularity of $f$ only in the first inequality. In the next part, we replace submodularity with the weaker submodular order property.
	\smallskip
	\paragraph{Submodular Order:}  Suppose that $f$ is a monotone function with an arrival-consistent submodular order $\pi$ over $N$. {\color{black} To show that $f(\psi(\cdot))$ is also a submodular order function we first define a candidate order $\pi_{\no}$ over $\no$. 
	
	Recall that, $\psi(\{a\})$ is a singleton that represents the realized outcome of action $a\in \no$. At a high level, we define $\pi_{\no}$ so that it is consistent with the order induced by $\pi$ over the set $(N\cap P_1) \cup (N\cap P_2)$ of realized outcomes. Specifically, distinct actions $a_1,a_2\in \no$ are ordered as follows:
	\begin{enumerate}[(i)]
		\item If $a_1,a_2\in \curac$, then $a_1\succ_{\pi_{\no}} a_2$ if and only if $N(a_1)\cap P_1 \succ_{\pi} N(a_2)\cap P_1$.
		\item If $\hat{a}_t,\hat{a}_\tau\in \hat{\curac}$, then $\hat{a}_t\succ_{\pi_{\no}} \hat{a}_\tau$ if and only if $t>\tau$ (which coincides with $N_t\cap P_2 \succ_{\pi} N_\tau\cap P_2$).
		\item If $a\in \curac$ and $\hat{a}_t\in \hat{\curac}$, then $\hat{a}_t \succ_{\pi_{\no}}  a$ if either $  N_t\cap P_2 \succ_{\pi} N(a)\cap P_1 $ or $N_t\cap P_2=N(a)\cap P_1$. 
	\end{enumerate}
By definition, order $\pi_{\no}$ is arrival-consistent. Note that the realized outcomes of $a$ and $\hat{a}_t$ coincide when $N_t\cap P_2=N(a)\cap P_1,$ and in this case we can set any order for $a$ and $\hat{a}_t$.
}

	In the earlier proof of submodularity of $f(\psi(\cdot))$ (when $f$ is submodular), we use the submodularity of $f$ only once to obtain the following inequality,
	\begin{equation}\label{sosuffice}
	f(\psi(C)\backslash \psi(A)\mid \psi(A))\,\leq\, f(\psi(C)\backslash \psi(A)\mid \psi(B)),
\end{equation}
	for all $B\subseteq A$ and $C\subseteq \no\backslash A$. To show that $\pi_{\no}$ is a submodular order for $f(\psi(\cdot))$, we only need to establish inequality \eqref{sosuffice} for $\pi_{\no}$-nested sets $B\subseteq A$ and $C\succ_{\pi_{\no}}A$. This will be the focus of the rest of the proof. Note that Lemma \ref{xzhelp} is the key ingredient in the following part of the proof and is shown separately later.

	From Lemma \ref{xzhelp}, we have that either $\psi(C)\backslash \psi(A)=\emptyset$ or $\psi(C)\backslash \psi(A)\succ_{\pi} \psi(A)$. If $\psi(C)\backslash \psi(A)=\emptyset$, the left hand side in \eqref{sosuffice} equals 0 and the inequality holds trivially by the monotonicity of $f$. Now, assume $\psi(C)\backslash \psi(A)\neq \emptyset$, so we have 
	\[\psi(C)\backslash \psi(A)\succ_{\pi}\psi(A).\] 
	Then, inequality \eqref{sosuffice} follows directly from the submodular order property of $f$ provided that,
				\[\psi(B) \text{ and } \psi(A)\text{ are $\pi$-nested}.\]
	Using Lemma \ref{xzhelp} again, we have that either $\psi(A\backslash B)\backslash \psi(B)=\emptyset$ or $\psi(A\backslash B)\backslash \psi(B)\succ_{\pi} \psi(B)$. Rewriting this using \eqref{psialt}, 
	we have that either $\psi(A)= \psi(B)$ or $\psi(A)\backslash \psi(B)\succ_{\pi} \psi(B)$. If $\psi(A)=\psi(B)$ then the two sides in \eqref{sosuffice} are equal. Otherwise, assume $\psi(A)\backslash \psi(B)\succ_{\pi} \psi(B)$. We conclude the proof by  observing that, $\psi(A)\backslash \psi(B)\succ_{\pi} \psi(B)$ implies that $\psi(B)$ and $\psi(A)$ are $\pi$-nested sets. 

	\hfill\Halmos 	 	\end{proof}
\begin{lemma}\label{xzhelp}
	Given sets $X,Z\subseteq \no$ such that $Z\succ_{\pi_{\no}} X$, and the function $\psi$ defined in \eqref{psidef}, it holds that either $\psi(Z)\backslash \psi(X)=\emptyset$ or $\psi(Z)\backslash \psi(X)\succ_{\pi} \psi(X)$.
\end{lemma}
\begin{proof}{Proof.} {\color{black} We assume that $\psi(Z)\backslash \psi(X)\neq\emptyset$ and show that this implies $\psi(Z)\backslash \psi(X)\succ_{\pi} \psi(X)$. We give a proof by contradiction. 

	Let
	\[Z_1=\{a\in Z \mid \exists a'\in X, \psi(a')=\psi(a)\},\]
	be the set of all actions in $Z$ that have the same realized outcome as some action in $X$. From \eqref{psialt}, we have that
		\[\psi(Z)\backslash \psi(X)=\psi(Z\backslash Z_1),\]
		which is the set of all realized outcomes that are not common between $Z$ and $X$. Now, suppose that $\psi(Z\backslash Z_1)$ does not succeed $\psi(X)$ in the order $\pi$. Then there exists distinct outcomes $e\in \psi(Z\backslash Z_1)$ and $e'\in \psi(X)$ such that $e'\succ_{\pi} e$. Since $\psi$ maps every action to a single outcome, there exist distinct actions $a\in Z\backslash Z_1$ and $a'\in X$ that correspond to the outcomes $e$ and $e'$ respectively. By definition of the order $\pi_{\no}$, we have that, $a'\succ_{\pi_{\no}} a$. This contradicts the fact that $Z\succ_{\pi_{\no}}X$. Therefore, $\psi(Z\backslash Z_1) \succ_{\pi}\psi(X)$.
}
	\hfill\Halmos\end{proof}
\subsection{Proof of Claim $(iii)$: The Invariance Property}\label{appx:pas:invar}
To prove the invariance property, we need to show that Greedy does not choose any action from the set $\h{\curac}$. Using Lemma \ref{key} (restated below), we show that if Greedy does not pick any of the new actions $\{\h{a}_1,\cdots,\h{a}_{t-1}\}$ prior to $t$ then it will not pick action $\h{a}_t$ at $t$.
Let $\h{\alg}_t=\{\h{\salg}_1,\cdots,\h{\salg}_{t-1}\}$ denote the set of actions chosen prior to arrival $t$ by \grd\ on instance $\gc$. Suppose that $\h{\alg}_t\subseteq \curac$ and note that this is true for $t=1$ because $\h{\alg}_1=\{\emptyset\}$. 
Using Lemma \ref{key} with $S=\h{\alg}_t$, we have,
\[ F(\h{a}_t\mid \h{\alg}_t)\leq \max_{a\in \ac_t} {F} (a\mid \h{\alg}_t).\] 
Thus, Greedy chooses an action from the original set of actions ($\ac_t$) at arrival, i.e., $\h{\salg}_t\in \ac_t$ and $\h{\alg}_{t+1}\subseteq \curac$. Now, our claim follows by induction over $t$.

\begin{repeatlemma}[Lemma \ref{key}.]
	For any set function $f$ (not necessarily monotone or submodular order), the function $F$ defined in \eqref{def1} satisfies, 
	\[F(\h{a}_t\mid S)\leq \max_{a\in \ac_t} F(a\mid S)
	\qquad \forall t\in [T], S\subseteq \curac\backslash \ac_t.\] 
\end{repeatlemma}
\begin{proof}{Proof.}
 	We being by restating the two main lemmas that we use to prove Lemma \ref{key}. For all $a\in \curac,\, S\subseteq \curac\backslash \{a\},$ and $e\in N(a)$, let
\[w_{e,S}=\sum_{P\in \curn(S)} \gamma(P)\,f(e\mid P),\]
denote the marginal reward of outcome $e\in N(a)$ when action set $S$ has been selected.
\begin{repeatlemma}[Lemma \ref{useF}.]
	For any set function $f$ (not necessarily monotone or submodular), every arrival $t\in [T]$, action $a\in \ac_t$, and set $S\subseteq \curac\backslash \{a\}$, the function $F$ defined in \eqref{def1} satisfies, 
	\[
	F(a\mid S)=\sum_{e\in N(a)} p_e\, w_{e,S}.
	\]
\end{repeatlemma}

\begin{repeatlemma}[Lemma \ref{key2}.]
	For any set function $f$ (not necessarily monotone or submodular), every arrival $t\in [T]$ and set $S\subseteq \curac\backslash \ac_t$, the function $F$ defined in \eqref{def1} satisfies, 
	\[
	F(\h{a}_t\mid S)=\sum_{a\in \ac_t} y^c_a \left(\sum_{e\in N(a)} p_e\, w_{e,S}\right).
	\]
\end{repeatlemma}	
	We are now ready to prove Lemma \ref{key}. First, we have from Lemma \ref{key2},
		\begin{eqnarray*}
			F(\h{a}_t\mid S)		&=& \sum_{a\in \ac_t}y^c_a \sum_{e\in N(a)} p_e\, w_{e,S},\\
			&=& \sum_{a\in \ac_t} y^c_a\, F(a\mid S),\\
			&\leq & \max_{a\in \ac_t}\, F(a\mid S).
		\end{eqnarray*}
		The second equality follows from Lemma \ref{key1}. The final equality follows from the fact that $\sum_{a\in \ac_t}y^c_a=1$.

	\hfill\Halmos	\end{proof}

\begin{proof}{Proof of Lemma \ref{key1}.} Since $S\cup a \subseteq \curac$, we use the original definition of $F$ from equation \eqref{defF} to express $F(a\mid S)$ as follows,
	\begin{eqnarray*}
		F(a\mid S)&=&\sum_{P\in \curn}\gamma(P)\left(f(N(S\cup a)\cap P)-f(N(S)\cap P)\right),\\
		&=&\sum_{P'\in \curn(S)}\,\sum_{e\in N(a)}\gamma(P')\,p_e\,\left(f(P'\cup e)-f(P')\right),\\
		&=&\sum_{e\in N(a)}\sum_{P'\in \curn(S)}\,\gamma(P')\,p_e\,f(e\mid P'),\\
		&=&\sum_{e\in N(a)}p_e\left(\sum_{P'\in \curn(S)}\,\gamma(P')\,f(e\mid P')\right).
	\end{eqnarray*}
The third and fourth identities follow from rearranging the terms. In the second identity, we use the fact for evaluating $F(a\mid S)$, it suffices to focus only on the set of partial mappings $\curn(S)$. 
Formally,
	\begin{eqnarray*}
		&&\sum_{P\in \curn}\gamma(P)\left(f(N(S\cup a)\cap P)-f(N(S)\cap P)\right)\\
		&&=\sum_{P'\in \curn(S)}\,\sum_{e\in N(a)}\sum_{P''\in \curn\backslash \curn(S\cup a)}\gamma(P')\,p_e\,\gamma(P'')\,\left(f(P'\cup e)-f(P')\right),\\
		&&=\left(\sum_{P''\in \curn\backslash \curn(S\cup a)}\gamma(P'')\right)\sum_{P'\in \curn(S)}\,\sum_{e\in N(a)}\gamma(P')\,p_e\, f(e\mid P'\cup e),\\
	\end{eqnarray*}
	here, we split $P\in \curn$ into $P'=P\cap N(S)$ (the realized outcomes of $S$), $e$ (the realized outcome of $a$), and $P''=P\backslash N(S\cup a)$ (all other realized outcomes). The probability that $P'\cup e$ is the set of realized outcomes of actions in $S\cup a$ is given by $\gamma(P')\, p_e$. Observe that the marginal value $f(e\mid P')$ is independent of $P''$ and 
	\[\sum_{P''\in \curn\backslash \curn(S\cup a)}\gamma(P'')=\sum_{P''\in \curn(\curac\backslash (S\cup a))}\gamma(P'')=1.\] %
	\hfill\Halmos	\end{proof}

\begin{proof}{Proof of Lemma \ref{key2}.} Since $S\cap \h{\curac}=\emptyset$, we use \eqref{def2} to obtain,
	\begin{eqnarray*}
		F(\h{a}_t\mid S)&=&  \sum_{P_1\in \curn }\sum_{P_2\in \mathcal{X}}\gamma(P_1)\, \alpha^c(P_2)\, f(N_t\cap P_2 \mid N(S)\cap P_1),\\
		&\overset{(a)}{=}&\sum_{P_2\in \cal{X}} \alpha^c(P_2)	\,\left(\sum_{P'_1\in \curn(S)}\gamma(P'_1)\, f(N_t\cap P_2\mid P'_1)\right),\\
		&\overset{(b)}{=}& \sum_{e\in N_t}\left(\sum_{P_2\in\cal{X}\mid P_2\ni e} \alpha^c(P_2)\right)\,w_{e,S},\\
		&\overset{(c)}{=}& \sum_{a\in \ac_t} \sum_{e\in N(a)} y^c_a\, p_e\, w_{e,S}.\\
		&=& \sum_{a\in \ac_t} y^c_a \sum_{e\in N(a)}\, p_e\, w_{e,S}.
	\end{eqnarray*}
	Equality $(a)$ follows from a rearrangement of terms and a change of variables (similar to the proof of Lemma \ref{key1}) that is captured in the following identity,
	\[ \sum_{P_1\in \curn }\gamma(P_1)f(N_t\cap P_2 \mid N(S)\cap P_1)=	\sum_{P'_1\in \curn(S)}\gamma(P'_1)\, f(N_t\cap P_2\mid P'_1).\]
	To get equality $(b)$  we use the fact that $N_t\cap P_2$ is a singleton (since $P_2\in \cal{X}$) and also use the definition of $w_{e,S}$. Equality $(c)$ follows from the fact that $\sum_{P_2\in\cal{X}\mid P_2\ni e} \alpha^c(P_2)=y^c_a\,p_e$.
	\hfill\Halmos	\end{proof}

\section{Proofs for Results in the RO Model}\label{appx:rom}
 \begin{repeattheorem}[Theorem \Ref{rom}.]
	Greedy is $\alpha$-competitive ($\alpha\geq 0.5096$) for OSW-SO in the RO model. 
\end{repeattheorem}

We will prove this theorem using the reduction technique. To use the technique in the RO model, we first show that permuting the arrivals of $G$ and then constructing the corresponding instance without stochastic outcomes is the same as starting with the instance $\gc$ and then permuting its arrivals.
\begin{lemma}\label{romhelp}
	Given an instance $G$ of OSW-SO, let $H=G_\sigma$ for some permutation $\sigma$ 
	and let $\gc_\sigma$ denote the permuted version of $\gc$.  Let $\h{H}$ denote the instance of OSW corresponding to $H$. 
	The instances $\h{H}$ and $\gc_{\sigma}$ are equivalent. 
\end{lemma}
\begin{proof}{Proof.}
	Recall that $\gc$ is constructed by adding an extra action $\h{a}_t$ for every $t\in [T]$, and defining the objective function $F$ using the optimal solution $\alpha^c(\cdot)$ of $\optc(G)$. $\gc_{\sigma}$ and $\gc$ have the same objective function but $\gc_\sigma$ has a reordered sequence of arrivals. Specifically, action set at the $t$-th arrival in $\gc_\sigma$ is given by $\ac_{\sigma(t)}\cup\h{a}_{\sigma(t)}$. 
	
	The instance $\h{H}$ is constructed by augmenting the action sets $\{\ac_{\sigma(t)}\}_{t\in [T]}$ of $H$ and defining the objective function $F_{\h{H}}$ using the optimal solution $\alpha^c(\cdot)$ of $\optc(H)$. Clearly,
	$\h{H}$ and $\gc_\sigma$ have an identical sequence of actions sets. Furthermore, since $\optc(H)=\optc(G)$, it follows that $\h{H}$ and $\gc_\sigma$ also have the same objective function, i.e., $F_{\h{H}}=F$. 
	\hfill\Halmos	\end{proof}

\begin{proof}{Proof of Theorem \ref{rom}.}
	Consider an arbitrary instance $G$ of OSW-SO with a fixed arrival sequence and the corresponding instance $\gc$ of OSW. 
	Let $\sigma$ denote a random permutation of the arrival sequence and let $H=G_\sigma$ denote a randomly permuted version of $G$. From Lemma \ref{romhelp}, we have that $\gc_{\sigma}$, the randomly permuted version of $\gc$, is identical to $\h{H}$. 
	
	From Lemma \ref{redprop}, we have
	\[\aalg(\h{H})=\aalg(H)\quad \text{and}\quad \opt(\h{H})\geq \optc(H).\] 
	Using the equivalence between $\h{H}$ and $\gc_\sigma$ and the permutation invariance of \opt\ and $\optc$, we get,
	\[\alg(\gc_\sigma)=\alg(G_\sigma)\quad \text{and}\quad \opt(\gc)= \opt(\gc_\sigma)\geq \optc(G_\sigma) =\optc(G).\]
	Now, taking expectation over $\sigma$, 
	\begin{eqnarray*}
		E_\sigma[\alg(G_\sigma)]&=&E_\sigma[\alg(\gc_\sigma)],\\
		&\geq &\alpha\, \opt(\gc_\sigma),\\
		&\geq & \alpha \,\optc(G),\\
		&\geq & \alpha\, \opt(G).
	\end{eqnarray*}
	The first inequality follows from the fact that $\gc_\sigma$ is an instance of OSW and Greedy is $\alpha$-competitive for OSW in the RO model. The final inequality follows from Lemma \ref{concdom}.
	
	\hfill\Halmos	\end{proof}

\section{Proofs for Results in the UIID Model}\label{appx:iid}
\subsection{From RO to UIID}\label{appx:ro2iid}
The competitive ratio of an online algorithm in the RO model is a lower bound on its competitive ratio in the UIID model. This is because the RO model ``subsumes" the UIID model in the following sense: Suppose that the adversary in the RO model generates a random instance using a UIID model. Permuting the arrival sequence uniformly randomly does not change the final distribution over instances. Thus, one can generate instances from a UIID model in the RO model. 
	
	Formally, let $\cal{H}$ denote the set of realizable instance in a UIID model and let $\bm{H}$ denote a random instance from $\cal{H}$. Let $H_\sigma$ represent the instance generated by randomly reordering the arrival sequence of instance $H\in \cal{H}$. Let $\beta_{RO}$ denote the competitive ratio of an online algorithm \aalg\ in the RO model. In the RO model, the adversary can pick the worst instance from $\cal{H}$. Therefore,
	\begin{eqnarray*}
		E_\sigma[\aalg(H_\sigma)]&\geq &\beta_{RO}\,\opt(H)\quad \forall H\in \cal{H},\\
		E_D[E_\sigma[\aalg(\bm{H}_\sigma)]] &\geq & \beta_{RO}\, E_D[\opt(\bm{H})],\\
		E_D[\aalg(\bm{H})] &\geq & \beta_{RO}\, E_D[\opt(\bm{H})].
	\end{eqnarray*}
	Since the last inequality holds for all sets $\cal{H}$ and distributions $D$, \aalg\ is at least $\beta_{RO}$-competitive in the UIID model. 

\subsection{Deterministic Outcomes}\label{appx:iiddeterm}
Recall that OSW with UIID arrivals is defined by a set of $T$ types with actions sets $\{\ac_t\}_{t\in [T]}$, a monotone DR submodular objective function $F:\mathbb{Z}^{|\curac|}_{+}\to \mathbb{R}_{+}$, and a distribution $D$ over $[T]$. In OSW-SO, in addition to the action set $\curac$ we have a set $N$ of possible outcomes, a monotone DR objective function $f: \mathbb{Z}^{|N|}_{+}\to \mathbb{R}_+$, and probabilities $\{p_e\}_{e\in N}$. A random instance $\bm{H}$ in the UIID model is generated by sampling a sequence of $T$ IID types from $D$. We use $H$ to denote a realization of $\bm{H}$ and $\mh$ to denote the set of all possible realizations of $\bm{H}$.

 We perform a change of variables to switch from the monotone DR submodular function $F$ to a monotone submodular function $F_T$ over an expanded ground set $\curac_T$ with $T\,|\curac|$ elements. As we discuss later, this will be helpful for proving Theorem \ref{uiid}. To familiarize the reader we introduce this change here and use the new notation in the proof of Theorem \ref{olduiid}. In the following, we use $t$ to index arrival types and $\tau$ to index the arrivals in $\bm{H}$.
Let 
\[\curac_T=\{a_{\tau}\mid a\in \curac,\, \tau\in [T]\}\quad \text{and}\quad \ac_{\tau,t}=\{a_\tau \mid a\in \ac_t\}\,\, \forall \tau,t\in [T],\] 
where $a_\tau$ is a distinct ``copy'' of action $a$. Note that $\ac_{\tau,t}$ is the set of feasible actions at arrival $\tau$ of type $t$. Before proceeding, we illustrate this change with an example.

\begin{example} Consider an instance with $T=2$ and two types of arrivals $A_1=\{a,b\}$ and $A_2=\{c,d\}$. The new ground set $\curac_T=\{a_1,a_2,b_1,b_2,c_1,c_2,d_1,d_2\}$. For $\tau\in[2]$, the new sets of feasible actions are as follows: $A_{\tau,1}=\{a_\tau,b_\tau\}$ and $A_{\tau,2}=\{c_{\tau},d_\tau\}$.	
\end{example}

 For $S\subseteq \curac_T$, let $x_{a,S}=|\{\tau\mid a_{\tau}\in S\}|$ denote the number of copies of $a$ in $S$. Let $\bm{x}_S=(x_{a,S})_{a\in \curac}\in \mathbb{Z}^{|\curac|}_{+}$. Define 
\[F_T(S)=F(\bm{x}_S)\quad \forall S\subseteq \curac_T.\]
It is not hard to see that $F_T$ is monotone submodular. We illustrate the change of variables with a simple example.

\begin{theorem}\label{olduiid} 
	Greedy is $(1-1/e)$-competitive for OSW in the UIID model.	
\end{theorem}
\begin{proof}{Proof.} 
	
	 Suppose that we add an extra layer of randomness to the UIID instance generation process: Given a random sequence $\bm{H}$, we permute the arrival sequence uniformly randomly. Let $\bm{H}_{\sigma}$ denote the final (permuted) sequence, where $\sigma$ is a random permutation of $[T]$. Observe that $\bm{H}_{\sigma}$ has the same distribution as $\bm{H}$. Every sequence $\bm{H}$ describes a unique instance with $T$ arrivals. With a slight abuse of notation, we use $\bm{H}$ to denote the sequence as well as corresponding instance. Let $O(\bm{H})$ denote the optimal offline solution (set of actions) on the instance $\bm{H}$. The optimal offline solution is permutation invariant, i.e., $O(\bm{H})=O(\bm{H}_\sigma)$. Let $\alg_\tau(\bm{H})$ denote the action set chosen by \grd\ prior to arrival $\tau$. Let $\aalg(\bm{H})$ and $\opt(\bm{H})$ denote the total reward of the solution output by \grd\ and \opt\ on $\bm{H}$.  Let $E_{\sigma}[\cdot]$ denote expectation w.r.t.\ the randomness in $\sigma$ and let $E_D[\cdot]$ denote expectation w.r.t.\ the randomness in $\bm{H}$.
	
	Now, consider two independently drawn UIID sequences, $H^1$ and $H^2$. Let $H^2_\sigma$ denote a randomly permuted version of $H^2$. 
	Fix an arbitrary index $\tau\in[T-1]$ and let $H^3$ denote the hybrid instance with $T+\tau-1$ arrivals, constructed by taking the first $\tau-1$ types of sequence $H^1$ followed by the $T$ types of $H^2_\sigma$. For $\tau'\geq \tau$, let $H(\tau')$ denote the type of arrival $\tau'$ 
	in $H^3$, which is of the same type as the $\tau'-\tau+1$-th arrival in $H^2_\sigma$. Let $\ac_{\tau,H(\tau')}$ denote the set of feasible actions for arrival $\tau'$ (arrival $\tau'-\tau+1$) in $H^3$ (in $H^2_\sigma$).
	
	Let $\alg_{\tau}(H^1)$ denote the set of actions chosen by Greedy prior to arrival $\tau$ on sequence $H^1$. Greedy also chooses the action set $\alg_{\tau}(H^1)$ prior to $\tau$ on sequence $H^3$. At arrival $\tau$ in $H^3,$ Greedy chooses,
	\[\salg_\tau(H^3)=\argmax_{a_\tau\in A_{\tau,H(\tau)}}\, F_T(a_\tau\mid \alg_\tau).\]
	Taking expectation over the randomness in $\sigma$, we have,
	\begin{eqnarray*}
		T\, E_{\sigma}[F_T(\salg_\tau(H^3)\mid \alg_\tau(H^1))] 
		&=&	T\left[\sum_{\tau'\geq \tau}^{T+\tau-1}\frac{1}{T}\max_{a_{\tau'}\in A_{\tau',H(\tau')}}\,F_T(a_{\tau'}\mid \alg_{\tau}(H^1))\right],\\
		&\geq & \sum_{a_{\tau'}\in O(H^2)} F_T(a_{\tau'}\mid \alg_{\tau}(H^1)),\\
		&\geq & F_T(O(H^2)\mid \alg_\tau(H^1)),\\
		&= &F_T(O(H^2)\cup \alg_\tau)-F_T(\alg_\tau(H^1)),\\
		&\geq &  F_T(O(H^2))-F_T(\alg_\tau(H^1)),
	\end{eqnarray*}
	here the first inequality follows from the fact that the $\tau'$-th arrival in $H^3$ has the same type as arrival $\tau'-\tau+1$ in $H^2_\sigma$ and 
	\[\max_{a_{\tau'}\in A_{\tau',H(\tau')}}\,F_T(a_{\tau'}\mid \alg_{\tau}(H^1))\geq \, F_T\left(O(H^2)\cap \ac_{\tau'-\tau+1,\,H(\tau')}\mid \alg_\tau(H^1)\right).\] 
	The second inequality follows from the submodularity of $F_T$ and the last inequality by monotonicity of $F$. 
	
	Taking expectation over the randomness in $H^1$ and $H^2$ ($H^3$) and using the fact that the $\tau$-th types in $H^3$ and $H^1$ have the same distribution, we get,
	\begin{eqnarray*}
		E_{D}[F_T(\salg_\tau(\bm{H})\mid \alg_\tau(\bm{H}))]&\geq &\frac{1}{T} \left[E_{D}[F_T(O(\bm{H}))]-E_{D}[F_T(\alg_\tau(\bm{H}))]\right]\quad \forall \tau\in[T],\\
		&=&\frac{1}{T} \left[E_{D}[\opt(\bm{H})]-E_{D}[F_T(\alg_\tau(\bm{H}))]\right]\quad \forall \tau\in[T].
	\end{eqnarray*}
	Note that, $E_{D}[F_T(\salg_\tau(\bm{H})\mid \alg_\tau(\bm{H}))]$ equals the expected marginal value of the Greedy solution at arrival $\tau$ in $\bm{H}$. Unrolling the recursion, we get
	\[E_{D}[\aalg(\bm{H})]\geq \left(1-\left(1-\frac{1}{T}\right)\right)^T E_{D}[\opt(\bm{H})].\]
	\hfill\Halmos	\end{proof}

\subsection{Proof of Theorem \ref{uiid}}\label{appx:iidso}
Recall that an instance of OSW-SO with UIID arrivals is defined by a set of $T$ types with actions sets $\{\ac_t\}_{t\in [T]}$, outcome sets $\{N_t\}_{t\in [T]}$, probabilities $\{p_e\}_{e\in N}$, monotone DR submodular objective functions $f:\mathbb{Z}^{|N|}_{+}\to \mathbb{R}_{+}$ and $F:\mathbb{Z}^{|\curac|}_{+}\to \mathbb{R}_{+}$, and a distribution $D$ over $[T]$. A random sequence of types $\bm{H}$ in the UIID model is generated by sampling a sequence of $T$ IID types from $D$. Let $H$ to denote a realization of $\bm{H}$ and $\mh$ to denote the set of all ($T^T$) possible realizations of $\bm{H}$.

\paragraph{Expanding the Ground Sets:} When two arrivals in realization $H$ have the same type, their action (outcome) sets are not disjoint. However, we use the disjointedness of action (outcome) sets in several lemmas (such as Lemma \ref{concdom} and Lemma \ref{key}) that are used in the proof of Theorem \ref{concave1} (for adversarial arrivals). To address this issue, we create $T$ distinct copies of every action and outcome and redefine $f$ (and $F$) as set functions  $f_T$ (and $F_T$) on expanded ground set $N_T$ (and $\curac_T$). This allows us to use the lemmas that we showed previously, although it increases the complexity of our notation. 

 Recall the definition of the expanded set of actions (introduced in Appendix \ref{appx:iiddeterm}),
\[\curac_T=\{a_{\tau}\mid a\in \curac,\, \tau\in [T]\}\quad \text{and} \quad \ac_{\tau,t}=\{a_\tau \mid a\in \ac_t\}\,\, \forall t,\tau\in [T],\] 
where we use $\tau$ to index arrivals and $t$ to index arrivals types. Recall that $\ac_{\tau,t}$ is the set of feasible actions at arrival $\tau$ of type $t$. Similarly, we expand the set of outcomes by creating $T$ copies of every outcome. Specifically, let
 \begin{equation*}\label{expandout1}
 	N(a_\tau)=\{e_\tau\mid e\in N(a)\},
 \end{equation*}
 be the outcome set of action $a_\tau\in \curac_T$ with probability
 \begin{equation*}\label{expandout2}
 	p_{e_{\tau}}=p_e\quad \forall\,e_{\tau}\in N(a_{\tau}).
 \end{equation*}
 The outcome $e_\tau$ is a distinct ``copy'' of outcome $e$ that can only be realized at arrival $\tau$. Let 
  \begin{equation}\label{expandout3}
 N_T=\cup_{a_{\tau}\in \curac_T} N(a_\tau),\quad N(S)=\cup_{a_\tau\in S} N(a_\tau)\,\, \forall S\subseteq \curac_T,\quad \text{and}\quad  N_{\tau,t}=N(\ac_{\tau,t}).
\end{equation}
Finally, let
\[\curn_T=\{P\subseteq N_T\mid |P\cap N(a_\tau)|=1\,\, \forall a_\tau\in \curac_T\},\] 
denote the set of all possible action to outcome mappings in the expanded ground sets. The outcome sets of distinct actions are now disjoint. Next, we define the objective functions on the expanded ground sets.

  Given a set $P\subseteq N_T$, let, $y_{e,P}=|\{\tau\mid e_{\tau}\in P\}|$ denote the number of copies of $e$ in $P$. Let $\bm{y}_P=(y_{e,P})_{e\in N}\in \mathbb{Z}^{|N|}_{+}$. The new objective functions are
\[f_T(P)=f(\bm{y}_P)\quad \forall P\subseteq N_T\]
\begin{equation}\label{uiidF1}
\text{and}\quad 	F_T(S)=\sum_{P\in \curn_T} \gamma(P)\, f_T(N(S)\cap P)\quad \forall S\subseteq \curac_T.
\end{equation}
It is not hard to see that $f_T$ and $F_T$ are monotone submodular.

\begin{proof}{Proof of Theorem \ref{uiid}.} 
Let $I=(G,D)$ denote an occurrence of OSW-SO with UIID arrivals. We claim that there exists an occurrence $\h{I}=(\h{G},D)$ of OSW (without stochastic outcomes) with UIID arrivals such that,
	\[(i)\,\, E_D[\opt(\h{\bm{H}})]\geq E_D[\opt(\bm{H})] \quad \text{ and } \quad (ii)\,\, E_D[\aalg(\bm{H})]=E_D[\aalg(\h{\bm{H}})],\]	
	here $\bm{H}$ denotes a random sequence of $T$ types that induces an instance of $I$ as well as $\h{I}$. With a slight abuse of notation, we use $\bm{H}$ to denote the instance of $I$ and $\h{\bm{H}}$ to denote the instance of $\h{I}$.
	
Assuming the truthfulness of $(i)$ and $(ii)$, we have,
	\[E_D[\aalg(\bm{H})]=E_D[\aalg(\h{\bm{H}})]\geq (1-1/e)\, E_D[\opt(\h{\bm{H}})]\geq \, E_D[\opt(\bm{H})],\]
	where the first inequality follows from Theorem \ref{olduiid}. In the rest of the proof, we construct the occurrence $\h{I}$ and then prove claims $(i)$ and $(ii)$. 
\smallskip

We begin by introducing some notation. Consider an arbitrary realization $H\in \mh$ of $\bm{H}$. For each arrival $\tau\in[T]$, let $H(\tau)$ denote its type. Since only actions corresponding to the realized arrival types are available, we define
\[\ac_H=\cup_{\tau\in[T]} A_{\tau,H(\tau)},\]
to be the set of actions that are feasible in realization  $H$. Let
\[ \cal{X}_{H}=\{X\mid X\subseteq N(\ac_H),\,|X\cap N_{\tau,H(\tau)}|=1\,\, \forall \tau\in [T]\}\]
be the collection of all realizable outcome sets that can arise from feasible selections of actions in $H$ (exactly one action from each arrival). We compare Greedy with the upper bound $\optc(H)=\max_{Y\in\Delta(\ac_H)} F^c_T(Y)$ on the offline benchmark. From Lemma \ref{concdom}, we have,
\begin{equation}\label{lem4uiid}
	\optc(H)\geq \opt(H).
\end{equation}
Thus, $E_D[\optc(\bm{H})]\geq E_D[\opt(\bm{H})]$. Let $Y^c_H=(y^c_{a_\tau,H})_{a_\tau\in \ac_H}$ denote the optimizer of $\optc(H)$ and 
let $\alpha^c$ denote the optimal probability distribution over $\cal{X}_H$ such that, 
\begin{equation}\label{alphauiid}
	\sum_{X\in \cal{X}_H} \alpha^c_H(X)= 1,\qquad \sum_{X\in\mathcal{X}_H\mid X\ni e_{\tau}} \alpha^c_H({X})\,  = \, p_{e_{\tau}}\,y^c_{a_{\tau,H}}\quad \forall e_{\tau}\in N(a_{\tau}),\,a_\tau\in \ac_H, 
\end{equation}
and
\begin{equation}\label{optcuiid}
	\quad	\optc(H)=F^c(Y^c_H)=\sum_{X\in \cal{X}_H} \alpha^c_H (X)\, f_T(X).
\end{equation}
{\color{black}
\noindent	\textbf{Construction of $\h{I}$:} The overall construction is quite similar to the adversarial case, in that we augment the set of actions in $G$ and then extend the function $F_T$ to the new ground set. The main difference is that, in the enlarged action set, we include \emph{multiple} new actions for every arrival type so that we can capture the optimal offline solution for every realization $H$. For the sake of intuition, as we describe each component of the construction, we first describe it in context of the original (non-expanded) ground set of actions before switching to the expanded ground set where we simply have $T$ copies of all types and actions. Note that the distribution $D$ over types remains unchanged.

\paragraph{New Actions:} Consider arrival number $k\in[T]$ in a realization $H$. Similar to the adversarial case, we introduce a new action $\h{a}_{k,H}$ to capture the optimal offline decision at arrival $k$ in $H$. Doing this for every $k$ and every $H$ gives a set of new actions $\{\h{a}_{k,H}\}_{k\in[T],H\in\mh}$ that together represent the optimal offline decision at every arrival of every possible instance in $\mh$. 

In the expanded ground set framework,   every action has $T$ copies, resulting in the following set of new actions:
  \[\h{\curac}_T=\{\h{a}_{\tau,k,H}\mid\, \tau\in[T],\, k\in[T], H\in \mh \}.\] 
Here, the set $\{\h{a}_{\tau,k,H}\}_{\tau\in[T]}$ consists of $T$ copies of $\h{a}_{k,H}$. Next, we distribute the new actions among the various types. Let $\no_T=\curac_T\cup \h{\curac}_T$ denote the new ground set of actions.

\paragraph{Enlarged Action Set for Each Arrival Type:} In the original (non-expanded) ground set, we augment the action set of type $t$ with the new set
\[\h{\ac}_{t}=\{\h{a}_{k,H}\mid H(k)=t,\, k\in [T],\, H\in\mh\},\] 
that captures the offline decision at every type-$t$ arrival across all realizations in $\mh$. In the expanded ground set, this augmentation is applied at the level of individual arrivals. Specifically, for each arrival $\tau$ of type $t$, we define
 \[\h{\ac}_{\tau,t}=\{\h{a}_{\tau,k,H}\mid H\in\mh,\,H(k)=t\}.\] 
Consequently, arrival $\tau$ of type $t$ in instance $\h{H}$ of $\h{I}$ has the set of feasible actions $\ac_{\tau,t}\cup \h{\ac}_{\tau,t}$. A key feature of the construction is that feasibility depends only on the arrival type. In particular, the action $\h{a}_{\tau,k,H}$ is feasible at arrival $\tau$ in every realization $H'\in \mh$ that satisfies $H'(\tau)=H(k)$, independent of the rest of the realization. 

\paragraph{A Canonical Feasible Solution:} Before extending $F_T$ to the enlarged ground set $\no_T$, observe that the set
\[\h{\ac}_H=\{\h{a}_{\tau,\tau,H}\mid \tau\in[T]\}\] 
is a feasible solution to instance $\h{H}$ of $\h{I}$, since it includes exactly one (feasible) action from $\h{\ac}_{\tau,H(\tau)}$ for each arrival $\tau\in[T]$. In the next step of the construction, we will define the extended objective function so that this canonical solution satisfies $F_T(\h{\ac}_H)=\optc(H)$.}
  
\paragraph{Objective:} We now extend $F_T$ over the ground set $\no_T$. Although $\gc$ does not have stochastic outcomes, for the purpose of this definition we use the expanded set of outcomes $N_T$ defined in \eqref{expandout3} and let 
\begin{equation}\label{nassignuiid}
	N(\h{a}_{\tau,k,H})=N_{\tau,H(k)}\quad \forall\,\h{a}_{\tau,k,H}\in \no_T\quad \text{and }\quad N(\h{\ac}_H)=N(\ac_H). 
\end{equation}

For every $S\subseteq \no_T$ such that $S_1=S\cap \curac_T$ and $S_2=S\cap \h{\curac}_{T}$, let
\begin{eqnarray}
	F_T(S)&\coloneqq\, &\sum_{P_1\in \curn_H}\sum_{P_H\in \cal{X}_{H} \,\, \forall H\in \mh} \gamma(P_1) \, \prod_{H\in \mh}\alpha^c_H(P_H)\, f_T\left(\cup_{H\in \mh} (N(S_2)\cap P_H)\cup (N(S_1)\cap P_1)\right).\label{defFT}%
\end{eqnarray}	
We have defined $F_T$ as the expected value of the realized outcomes when we independently sample the mappings $P_1$ and $\{P_H\}_{H\in \mh}$ w.p.\ $\gamma(P_1)$ and $\{\alpha^c_H(P_H)\}_{H\in \mh}$ respectively. 
When $S_2\subseteq \h{\ac}_H$ for some $H$, this definition is consistent with the function defined in \eqref{def1}. In particular, 
\begin{eqnarray}
	F_T(\h{\ac}_H)&=&\sum_{P\in \cal{X}_H} \alpha^c_H(P)\,\,f_T\left( N(\ac_H)\cap P)\right)=\sum_{P\in \cal{X}_H} \alpha^c_H(P)\,\,f_T(P)=\optc(H)\label{optinequiid}, 
\end{eqnarray}
here we used \eqref{nassignuiid} and the fact that 
$\sum_{P_1\in \curn_T}\gamma(P_1)=1$ and $\sum_{P_{H'}\in \cal{X}_{H'}}\alpha^c_{H'}=1$ for every $H'\in \mh$. 
This completes the definition of instance $\gc$. 
Similar to the function $F$ that we defined in \eqref{def1}, we have that $F_T$ is monotone submodular. For completeness, we include a proof of this after the present proof concludes (see Lemma \ref{ftsubmod}). Overall, $\gc$, which is given by feasible action sets $\ac_{\tau,t}\cup \h{\ac}_{\tau,t}$ and $F_T$, and the distribution $D$ together define an occurrence $\h{I}$ of OSW with UIID arrivals. Next, we show properties $(i)$ and $(ii)$.

To see $(i)$, let $\bm{H}$ denote a random sequence of $T$ types (and induced instance of $I$) and let $\h{\bm{H}}$ denote the induced instance of $\h{I}$. Observe that
\begin{eqnarray*}
	E_D[\opt(\h{\bm{H}})]\, \geq\, E_D[F_T(\h{\ac}_{\bm{H}})]\, =\, E_D[\optc(\bm{H})]\,\geq\, E_D[\opt(\bm{H})].
\end{eqnarray*}
The first inequality follows from the fact that $\h{\ac}_{\bm{H}}$ is a feasible solution of instance $\h{\bm{H}}$. The equality follows from \eqref{optinequiid} and the final inequality follows from \eqref{lem4uiid}.

 To establish $(ii)$, it is sufficient to show that $\aalg(\h{H})=\aalg(H)$. Similar to the adversarial case, we establish this by proving that \grd\ does not choose any of the new actions of instance $\h{H}$. Specifically, we show that on $\h{H}$, \grd\ chooses from the original action set $\ac_{\tau,H(\tau)}$ at every arrival $\tau\in[T]$. The key ingredient is as follows: 
For every arrival $\tau\in[T]$ of $\h{H}$ with feasible action set $\ac_{\tau,H(\tau)}\cup \h{\ac}_{\tau,H(\tau)}$, and every set $S\subseteq \curac_T\backslash \ac_{\tau,H(\tau)}$, we have
 \[F_T(\h{a}_{\tau,k,H'}\mid S)\leq \max_{a_{\tau}\in\ac_{\tau,H(\tau)}} F_T(a_\tau\mid S)\quad \forall\, \h{a}_{\tau,k,H'}\in\h{\ac}_{\tau,H(\tau)}.\]
 Therefore, in terms of marginal value, every action in $\h{\ac}_{\tau,H(\tau)}$ is dominated by the best action in $\ac_{\tau,H(\tau)}$ provided that $S$ does not contain any action from $\h{\curac}_T\cup \ac_{\tau,H(\tau)}$. This is a generalization of Lemma \ref{key}, which is formally stated and proved after the proof of this theorem (see Lemma \ref{keyuiid}). Using this result for $\tau=1$ and $S=\emptyset$, we have that \grd\ does not choose a new action at arrival 1. If \grd\ does not choose any new action prior to arrival $\tau$, then using the result above we have that \grd\ chooses an original action at $\tau$, completing the proof.
	\hfill\Halmos\end{proof}

\begin{lemma}\label{ftsubmod}
The function $F_T$ defined in \eqref{defFT} is monotone submodular.	
	\end{lemma}
\begin{proof}{Proof.}
	Fix arbitrary sets $P_1\in \curn_T$ and $P_H\in \cal{X}_H\,\, \forall\, H\in \mh$. Let $P_2=\cup_{H\in \mh} P_H$ and define $\psi:2^{\no_T}\to 2^{N_T}$ as follows,
	\begin{equation*}\label{psidefuiid}
		\psi(S)= (N(S\cap \curac_T)\cap P_1) \cup (N(S\cap \h{\curac}_T)\cap P_2).
	\end{equation*}
	Observe that $F_T$ is a linear combination of functions of the form $f_T(\psi())$. Since the family of monotone submodular functions is closed under addition, it suffices to show that $f_T(\psi(\cdot)):2^{\no_T}\to \mathbb{R}_{+}$ is a monotone submodular function when $f_T$ is monotone submodular. The rest of the proof mimics the proof of Lemma \ref{gcprop1} and we repeat it for completeness.
	
		\paragraph{Monotonicity:} Suppose that $f_T$ is monotone. Observe that,
	\begin{equation}\label{psialtuiid}
		\psi(S)=\cup_{x\in S} \psi(x	)\quad \forall S\subseteq \no_T.
	\end{equation}
	Thus, $f(\psi(\cdot))$ is monotone because $f$ is monotone and $\psi(A)\supseteq \psi(B)$ for all $B\subseteq A$. 
	\smallskip
	
	\paragraph{Submodularity:} Suppose that $f_T$ is a monotone submodular function. Consider an sets $B\subseteq A\subseteq \no_T$, and set $C\subseteq \no_T\backslash A$, we have,
	\begin{eqnarray*}
		f_T(\psi(A\cup C))-f_T(\psi(A))&=&f_T(\psi(A\cup C)\backslash \psi(A) \mid \psi(A)),\\
		&\leq& f_T(\psi(A\cup C)\backslash \psi(A) \mid \psi(B)),\\
		&=&	f_T(\psi(C)\backslash \psi(A) \mid \psi(B)),\\
		&\leq&	f_T(\psi(C)\backslash \psi(B) \mid \psi(B)),\\
		&= & f_T(\psi(B\cup C)\backslash \psi(B) \mid \psi(B)),\\
		&=& 	f_T(\psi(B\cup C))-f_T(\psi(B)).
	\end{eqnarray*}
The first inequality follows from submodularity of $f_T$. The second inequality follows from the monotonicity of $f_T$. The second and third equalities follow from \eqref{psialtuiid}. 	This proves that $f_T(\psi(\cdot))$ is submodular. 
	
	\hfill\Halmos\end{proof}

\begin{lemma}\label{keyuiid}
	For every arrival $\tau\in [T]$ of sequence $H\in \mh$ and every set $S\subseteq \curac_T\backslash \ac_{\tau,H(\tau)}$, we have
	\[F_T(\h{a}_{\tau,k,H'}\mid S)\leq \max_{a_{\tau}\in \ac_{\tau,t}} F_T(a_{\tau}\mid S)\quad \forall\, \h{a}_{\tau,k,H'}\in \h{\ac}_{\tau,H(\tau)}.\] 
\end{lemma}
\begin{proof}{Proof.} 
	Fix an arbitrary action $\h{a}_{\tau,k,H'}\in \h{\ac}_{\tau,t}$. For brevity, let $t=H(\tau)$ and note that $H'(k)=t$ because $\h{a}_{\tau,k,H'}\in \h{\ac}_{\tau,t}$. We have,
	\begin{eqnarray}
		F_T(\h{a}_{\tau,k,H'}\mid S)	& = &  \sum_{P\in \curn_T}\sum_{P_{H'}\in \cal{X}_{H'}}\gamma(P)\, \alpha^c_{H'}(P_{H'})\, f_T(N_{\tau,t}\cap P_{H'} \mid N(S)\cap P),\label{eq0u}\\		
		& = &  \sum_{P\in \curn_T}\gamma(P)\,\left( \sum_{P_{H'}\in \cal{X}_{H'}}\sum_{e_{\tau}\in N_{\tau,t}\cap P_{H'}}\alpha^c_{H'}(P_{H'})\, f_T(e_{\tau} \mid N(S)\cap P)\right),\label{eq1u}\\	
		& =&   \sum_{P\in \curn_T}\gamma(P)\,\left( \sum_{e_{\tau}\in  N_{\tau,t}}f_T(e_{\tau} \mid N(S)\cap P)\left(\sum_{P_{H'}\in \cal{X}_{H'} \mid P_{H'}\ni e_{\tau}}\alpha^c_{H'}(P_{H'})\right)\right),\nonumber\\
		&= &  \sum_{P\in \curn_T}\gamma(P)\left(\sum_{a\in \ac_{\tau,t}}\sum_{e_{\tau}\in  N(a_{\tau})} y^c_{a_{\tau}, H'}\,\, p_{e_{\tau}} \,f_T(e_{\tau} \mid N(S)\cap P)\right),\label{eq2u}\\
		& =&  \sum_{a_{\tau}\in  \ac_{\tau,t}}  y^c_{a_{\tau}, H'} \left(\sum_{P\in \curn_T}\gamma(P)\,\left(\sum_{e_{\tau}\in N(a_{\tau})}p_{e_{\tau}}\,f_T(e_{\tau} \mid N(S)\cap P)\right)\right),\nonumber\\
		& =&  \sum_{a_{\tau}\in  \ac_{\tau,t}}  y^c_{a_{\tau}, H'}\, F(a_{\tau}\mid S),\label{eq4u}\\
		&\leq &  \max_{a_{\tau}\in \ac_{\tau,t}} F(a_{\tau}\mid S)\label{eq7u}.
	\end{eqnarray}
We obtain equality \eqref{eq0u} by ignoring the realized mapping $P_{H''}$ for every instance $H''$ ($\neq H'$) because the sets $S\cup\{\h{a}_{\tau,k,H'}\}$ and $\h{\ac}_{H''}$ are disjoint for all $H''\in \mh\backslash \{H'\}$. We get equality \eqref{eq1u} by using the fact that the set $N_{\tau,t}\cap P_{H'}$ is a singleton.
Equality \eqref{eq2u} follows from \eqref{alphauiid}. 
Inequality \eqref{eq7u} follows from the fact that $Y^c_{H'}\in \Delta(\ac_{H'})$, which implies that $\sum_{a_{\tau}\in \ac_{\tau,t}} y^c_{a_\tau,{H'}}= 1$.  
It remains to show equality \eqref{eq4u}. For every $a_{\tau}\in \ac_{\tau,t}$ and $X\subseteq \curac_T\backslash \ac_{\tau,t}$, we have
\begin{eqnarray*}
	F_T(a_{\tau}\mid X)&=&\sum_{P\in \curn_T}\gamma(P)\left[f_T(N(X\cup a_{\tau})\cap P)-f_T(N(X)\cap P)\right],\\
	&=&\sum_{P\in \curn_T}\gamma(P)\,f_T(N(a_{\tau})\cap P\mid N(X)\cap P),\\ 	
	&=&\sum_{P\in \curn_T}\sum_{e_{\tau}\in P\cap N(a_{\tau})}\gamma(P\backslash N(a_{\tau}))\,p_{e_{\tau}}\,f_T(e_{\tau}\mid N(X)\cap P),\\ 
	&\overset{(*)}{=}&\sum_{P\in \curn_T}\gamma(P\backslash N(a_{\tau}))\left[\sum_{e'_\tau\in N(a_\tau)} p_{e'_\tau}\left(\sum_{e_\tau\in N(a_\tau)}p_{e_\tau}\,f_T(e_\tau\mid N(X)\cap (P\cup e'_\tau))\right)\right],\\ 
	&=&\sum_{Q\in \curn_T}\gamma(Q)\left(\sum_{e_\tau\in N(a_\tau)}p_{e_\tau}\,f_T(e_\tau\mid N(X)\cap Q)\right), 		
\end{eqnarray*}
To obtain equation $(*)$, we use the fact that for $X\subseteq \curac_T\backslash \ac_{\tau,t}$, we have,
\[N(X)\cap (P\cup e'_\tau)=N(X)\cap P,\] 
and the following identities,
\[\sum_{e'_\tau\in N(a_\tau)}p_{e'_{\tau}}=\sum_{e\in N(a)}p_e=1.\]	

	\hfill\Halmos	\end{proof}

\section{The Single-Arrival Problem and Approximate-Greedy}\label{appx:sap}
As noted in Remark \ref{sapremark}, the \grd\ algorithm solves an instance of the \sap\ (restated below) at each arrival $t\in [T]$, with weights (marginal reward values) $w_{e,\alg_t}=\sum_{P\in \curn(\alg_t)} \gamma(P)\,f(e\mid P)$. 
\[\text{\sap\ at $t$:}\qquad \argmax_{a\in \ac_t} \sum_{e\in N(a)}p_e\,w_{e,\alg_t}.\]
Thus far, we have assumed that the \sap\ can be solved efficiently at each arrival using value oracles for $f$ and $F$. However, the problem can be quite challenging to solve in various applications due to the following reasons: 
\begin{enumerate}[$(i)$]
	\item \emph{\sap\ may be NP-hard:}  In certain settings, the set of feasible actions $\ac_t$ is defined implicitly and can be exponentially large in the size of the instance, making the \sap\ NP-hard. For example, in online assortment optimization, the \sap\ is equivalent to the problem of finding a revenue optimal assortment, which is an NP-hard problem for many choice models \citep{assorthard}.
	\item \emph{Computing the weights may be difficult:} To solve the \sap\ at arrival $t$, we may need to compute the weight $w_{e,\alg_t}$ for each $e\in N_t$. 	{\color{black} In some settings, such as online assortment optimization, the weights can be computed in polynomial-time (see Appendix \ref{appx:sapeg})}. However, in general, $w_{e,\alg_t}$ is a sum of exponentially many terms, making it computationally infeasible to compute the exact value. 
\end{enumerate}	

These challenges are common in combinatorial optimization, both in online and offline settings \citep{negin, asadpour2016maximizing}. The standard approach to addressing them is to find an approximately optimal solution to the \sap, which we will discuss in more detail below. Although this may reduce the competitive ratio of the online algorithm, we show that by extending the reduction technique, the impact on the competitive ratio remains the same in both settings with and without stochastic outcomes.  

In the following discussion, let $w_e(P)=f(e\mid P)$ for all $P\in \curn(\alg_t)$. We assume that $w_e(P)$ is easy to compute given $P$. For brevity, we omit the subscript $\alg_t$ and use the shorthand $w_e=\sum_{P\in \curn(\alg_t)} \gamma(P)\,f(e\mid P)$. Note that $w_e=E[w_e(P)]$, where the expectation is taken w.r.t.\ the stochastic outcomes. 

\subsubsection*{Approximation Oracle for SAP.}

When \sap\ is NP-hard, we assume the availability of an oracle that outputs an $\eta$-approximate solution for the \sap\ with $\eta\in(0,1]$. Specifically, for any instance of the \sap, the oracle returns a solution $a_{t}$ such that:
\[\text{$\eta$-approximation for \sap:}\qquad \sum_{e\in N(a_t)}p_e\,w_{e}\,\geq\, \eta\, \left(\max_{a\in \ac_t} \sum_{e\in N(a)}p_e\,w_{e}\right).\] 
There exists a large body of literature on approximation algorithms for the \sap\ in various settings, including assortment optimization and stochastic rewards with patience. For further references, we refer to \cite{negin} and \cite{brubach2}.  

\subsubsection*{Sample Average Approximation of $w_e$.} 
When $w_{e}$ cannot be computed exactly, we use a sample average approximation (SAA) of $w_e$. Specifically, we sample $J$ (partial) mappings $P^j\in \curn(\alg_t)$ for $j\in[J]$ and compute, 
\[\text{SAA of $w_e$:}\qquad \wapx_e=\frac{1}{J}\sum_{j\in [J]} w_e(P^j)\qquad \forall e\in N_t,\]
Note that each partial mapping $P^j$ only includes the realized outcomes of the $t-1$ actions in the set $\alg_t$. Using a standard Chernoff bound, we can derive the following result. 
\begin{lemma}\label{chernoff}
	For $\delta_t\in(0,1)$, given $J_t=\frac{4\log \frac{|N_t|}{\delta_t}}{\delta_t^2}$ independent samples from $\curn(\aalg_{t})$, we have
	\[P\big(|\wapx_e -w_e| < \delta_t f(e)\,\, \forall e\in N_t \big)\geq 
	1-\delta_t.\]
\end{lemma}

We include the proof in Appendix \ref{appx:saproof}. 
Later, we discuss the appropriate values of $\{J_t\}_{t\in [T]}$ to ensure that $\sum_{t\in[T]} \delta_t $ is as small as desired, without requiring any knowledge of the number of arrivals $T$. 

\subsubsection*{The Approximate-\grd\ Algorithm.} 
Let Apx-\grd\ denote the modified Greedy algorithm where we select an approximately optimal solution to the \sap\ instance with approximate weights at each arrival. Specifically, at arrival $t\in[T]$, we compute the approximate weights $\wapx_e$ using $J_t$ samples and choose the action determined by the $\eta$-approximation oracle for the \sap\ with weights $\{\wapx_e\}_{e\in N_t}$. 
Let $\beta(\eta,\sum_{t\in T} \delta_t)$ denote the competitive ratio of Apx-\aalg. Note that Apx-\aalg\ is a randomized algorithm and the competitive ratio compares the expected performance of the algorithm with the optimal offline. We expect the performance of Apx-\aalg\ to degrade as the quality of the \sap\ solution deteriorates, i.e., as $\eta$ decreases and $\sum_{t\in [T]} \delta_t$ increases. We establish the following competitive ratio guarantees with deterministic and stochastic outcomes. 
\begin{theorem}\label{sap1}
	Given an $\eta$-approximate oracle for the \sap\ and $J_t=\frac{4\log \frac{N_t}{\delta_t}}{\delta_t^2}$ independent samples of the partial mapping at every arrival $t\in [T]$, the \apx\grd\ algorithm is at least $\left(\frac{\eta}{1+\eta}-\frac{3}{2}\,\sum_{t\in [T]}\delta_t\right)$-competitive for OSOW in the adversarial model.
\end{theorem}
\begin{theorem}\label{sap2}
	In the adversarial, RO, and UIID models, \apx\grd\ algorithm has the same competitive ratio both with and without stochastic outcomes. 
\end{theorem}
The proofs of Theorem \ref{sap1} and Theorem \ref{sap2} can be found in Appendix \ref{appx:sap1} and Appendix \ref{appx:sap2}, respectively. Note that $\frac{\eta}{1+\eta}> \frac{\eta}{2}$ for all $\eta\in (0,1)$. 
For deterministic outcomes, we believe that results similar to the one in Theorem \ref{sap1} should hold in the other arrival models as well as for the Greedy-like algorithms discussed in Section \ref{sec:gla}. Formally deriving these individual results is beyond the scope of this paper.

\subsubsection*{Choosing $J_t$ (and $\delta_t$).} \label{deltat} Ideally, we want $\sum_{t\in [T]}\delta_t=O(\epsilon)$, where $\epsilon\in(0,1)$ is a tunable parameter and the big-$O$ suppresses constant factors. This condition holds when $\delta_t=\frac{\epsilon}{T}\,\, \forall t\in [T]$, i.e., when we draw $\frac{4}{\epsilon^2}T^2\log \frac{T\,|N_t|}{\epsilon}$ samples at each arrival. When $T$ is unknown, it suffices to set $\delta_t=\frac{\epsilon}{t^2}$, yielding
\[J_t=\frac{4}{\epsilon^2}t^{4}\log \frac{t^2\,|N_t|}{\epsilon} \qquad \forall t\in[T],\]
since $\sum_{t=1}^{+\infty}\frac{\epsilon}{t^2}= \frac{\epsilon\,\pi^2}{6}$. We note that our goal was not to determine the minimum number of required samples; a smaller number may suffice with further refinements.

For illustration,  in the context of online two-sided assortment optimization, where each action corresponds to a subset of $I$ (the set of resources) and each action has at most $|I|+1$ outcomes since an arrival chooses at most one resource, we have $|N_t|=O(|I|\,2^{I})$. In the case of stochastic rewards with patience, we get $|N_t|=O(2^{I}\times |I|!)$, as there are $O(|I|!)$ possible actions, each corresponding to a permutation of the resources. In both cases, $\log |N_t|=O(|I|\log |I|)$, which remains polynomial in the problem size.

\subsection{The Single Arrival Problem in Special Cases of OSW-SO}\label{appx:sapeg}
As mentioned earlier, the Single Arrival Problem, \sap, given by
\[\argmax_{a\in \ac_t}\, F(a\mid S),\]
can be a computationally challenging optimization problem. The marginal value $F(a\mid S)$ is the expected increase in total reward due to action $a$. In the following, we examine the complexity of computing $F(a\mid S)$ in two settings.

\emph{Online assortment optimization:} In this setting, each action corresponds to a subset of the set of resources $I$. For the sake of simplicity, consider the unit capacity setting, i.e., $c_i=1$ for all $i\in I$. This is without loss of generality because we can always split a resource with more than 1 unit of capacity into several identical resources each with unit capacity. Now,  when assortment $a\subseteq I$ is shown to arrival $t$, the arrival chooses resource $i\in a$ with probability $\phi_t(i,a)$. This generates a reward of $r_i$ if no arrival before $t$ has chosen resource $i$. Let $p_i(t)$ denote the probability of the event that no arrival before $t$ chose resource $i$. Then,
\begin{eqnarray*}
	F(a\mid S)
	&=& \sum_{e\in N(a)} p_e \left(\sum_{P\in \curn(S)} \gamma(P)\, \sum_{i\in I} f_i(e\mid P)\right),\\
		&=& \sum_{e\in N(a)} p_e \left(\sum_{P\in \curn(S)} \gamma(P)\, \sum_{i\in I} r_i \mathbbm{1}(i \text{ unmatched in } P \text{ and } e_{i}=1)\right),\\
			&=& \sum_{i\in I} r_i \sum_{e\in N(a)\mid e_i=1} p_e \left(\sum_{P\in \curn(S)} \gamma(P) \mathbbm{1}(i \text{ unmatched in } P)\right),\\
					&=& \sum_{i\in I} r_i\, p_i(t)\sum_{e\in N(a)\mid e_i=1} p_e,\\
						&=& \sum_{i\in I} r_i\, p_i(t)\, \phi_t(i,a),\\
	&=&\sum_{i\in a} r_i\,\phi_t(i,a)\, p_i(t).
\end{eqnarray*}
Here, we use the fact that $p_i(t)=\sum_{P\in \curn(S)} \gamma(P) \mathbbm{1}(i \text{ unmatched in } P)$ and $\phi_t(i,a)=0\,\, \forall i\not\in a$. Thus, the \sap\ problem in assortment optimization is equivalent to finding the revenue optimal assortment over ground set $I$, with item prices given by $\{r_i\,p_i(t)\}_{i\in I}$. The probabilities $p_i(t)$ are easy to compute. Notably, if assortment $a_t$ is shown to arrival $t$, the probabilities for arrival $t+1$ are given by, 
\[p_i(t+1)=p_i(t)(1-\phi_t(i,a_t))\quad \forall i\in I.\]
Clearly, $p_i(1)=1$ for all $i\in I$ and it is straightforward to update the probabilities after each arrival. Therefore, the \sap\ problem for assortment optimization is as hard or as easy as the problem of finding a revenue optimal assortment. 

\emph{Two-sided assortment optimization:} As we will see, computing the \sap\ weights exactly can be much more challenging in two-sided assortment optimization. At a high level, this is because the \sap\ objective aims to maximize the expected total gain in the probability that some resource selects arrival $t$, and this probability may depend on the choices made by all previous arrivals. Formally, let $T_{i,t}$ denote the set of arrivals up to and including $t$ that choose resource $i\in I$. Let $\phi_i(T_{i,t})$ denote the probability that resource $i$ chooses an arrival from the set $T_{i,t}$. Let $a_{\tau}$ denote the (non-adaptive) assortment shown to arrival $\tau\in[t-1]$. Then the marginal increase in objective from showing assortment $a_t$ to arrival $t$ is as follows:
\[F(a_t\mid \{a_1,\cdots, a_{t-1}\})=\sum_{i\in I} \left[\sum_{T_{i,t}\subseteq [t]}\left(\prod_{\tau\in T_{i,t}}\phi_{\tau}(i,a_{\tau})\prod_{\tau\in [t]\backslash T_{i,t}}(1-\phi_{\tau}(i,a_{\tau}))\right) \big(\phi_i(T_{i,t})-\phi_i(T_{i,t}\backslash \{t\})\big)\right].\]
Here, $\prod_{\tau\in T_{i,t}}\phi_{\tau}(i,a_{\tau})\prod_{\tau\in [t]\backslash T_{i,t}}(1-\phi_{\tau}(i,a_{\tau}))$ is the probability that $T_{i,t}\subseteq [t]$ is the set of arrivals that choose $i$ and $\phi_i(T_{i,t})-\phi_i(T_{i,t}\backslash \{t\})$ is the increase in the probability that resource $i$ chooses an arrival, conditioned on the set $T_{i,t}$. Note that $\phi_i(T_{i,t})-\phi_i(T_{i,t}\backslash \{t\})=0$ when $t\not\in T_{i,t}$. 
\subsection{Proof of Lemma \ref{chernoff}}\label{appx:saproof}
\begin{repeatlemma}[Lemma \ref{chernoff}.]
	Given $J_t=\frac{4\log \frac{|N_t|}{\delta_t}}{\delta_t^2}$ independent samples from $\curn(\aalg_{t})$, we have,
	\[P\big(|\wapx_e -w_e| < \delta_t f(e)\,\, \forall e\in N_t \big)\geq 
	1-\delta_t.\]
\end{repeatlemma}
\begin{proof}{Proof.} 	We first note that all weights $w_e$ and their approximations $w_e^{\mathrm{apx}}$ lie in the interval $[0,f(e)]$. This holds because
	\[
	0 \le f(e \mid P) \le f(e)
	\]
	for any monotone arrival-consistent submodular-order (and submodular) function $f$, for every $e \in N_t$ and every $P \subseteq \bigcup_{\tau < t} N_\tau$. 
	We use the following version of the Chernoff-Hoeffding bound (Lemma 9 in \cite{vondraktut}):
	
	\emph{(Chernoff Bound)} Let $X_1,X_2,\cdots, X_n$ be $n$ independent random variables such that $X_i\in[0,b]$ for some $b\in(0,1)$ and for all $i$. Let $X=\frac{1}{n}\sum_{i\in[n]} X_i$. Then, for all $\delta>0$,
	\begin{equation}\label{chernof}
	P\left(\left|\frac{1}{n}\sum_{i\in[n]}X_i -X\right| \geq \delta b\right)<e^{-\frac{n\delta^2}{4}}.
\end{equation}
	
	Fix an arbitrary outcome $e\in N_t$. Using inequality \eqref{chernof} with $\delta=\delta_T$ and $n=\frac{4\log \frac{|N_t|}{\delta_t}}{\delta_t^2}$, we have,
	\[P\left(|\wapx_e -w_e| > \delta_t f(e)\right)\leq e^{-\frac{4\log \frac{|N_t|}{\delta_t}\,\delta_t^2}{4\delta_t^2}}=\frac{\delta_t}{|N_t|}.\]
	Applying the union bound, we have,
	\[P\left(|\wapx_e -w_e| \geq \delta_t f(e)\,\, \forall e\in N_t \right)\leq \delta_t,\]
	as desired.

	\hfill\Halmos	\end{proof}

\subsection{\apx\grd\ in the Adversarial Model}\label{appx:sap1}
\begin{repeattheorem}[Theorem \ref{sap1}.]
	Given an $\eta$-approximate oracle for the \sap\ and $J_t=\frac{4\log \frac{N_t}{\delta_t}}{\delta_t^2}$ independent samples of the partial mapping at every arrival $t\in [T]$, the \apx\grd\ algorithm is at least $\left(\frac{\eta}{1+\eta}-\frac{3}{2}\,\sum_{t\in [T]}\delta_t\right)$-competitive for OSOW in the adversarial model.
\end{repeattheorem}
\begin{proof}{Proof.}
	Consider an arbitrary instance $G$ of OSOW. Recall that \nopt\ is the offline benchmark for OSOW. Let $\opt_{T+1}$ denote the output of \nopt\ on instance $G$. Let $o_t=\opt_{T+1}\cap A_t$ denote the action selected at $t$ by \nopt\ and let $\opt_t=\{o_1,\cdots,o_{t-1}\}$ for every $t\in [T]$. Note that \apx\grd\ is a randomized algorithm due to the randomness in estimating the weights at each arrival. Let $\bm{\salg}_t$ denote the (random) action chosen by \apx\grd\ at arrival $t$ and let $\bm{R}_t=\{\bm{\salg}_1,\cdots, \bm{\salg}_{t-1}\}$ denote the set of actions chosen prior to arrival $t$. Let us fix an arbitrary sample path (by fixing all the samples used for SAA).  Let $\bm{\salg}_t=\salg_t$ and let $\bm{R}_t=\alg_t$ on this sample path. 
	We have,
	\begin{eqnarray}
		F (\opt_{T+1})&\leq &F\left(\alg_{T+1}\cup \opt_{T+1}\right),\nonumber\\
		&\leq & F\left(\cup_{t\in [T]} \{\salg_t\}\right)+ \sum_{t\in [T]} F\left(\{o_{t}\}\backslash \{\salg_{t}\}\mid \cup_{\tau\in[t-1]} \{\salg_\tau\}\right),\nonumber\\
		&\leq & F(\alg_{T+1})+ \sum_{t\in [T]} F(o_{t}\mid \alg_t),\label{apxgrd1}
	\end{eqnarray}
	here the first and the last inequalities follow from monotonicity of $F$. The second inequality follows from the submodular order property and Lemma \ref{SOineq}. Next, we will consider the expectation over randomness in \apx\alg\ on the RHS of inequality \eqref{apxgrd1}. Let $E[\cdot]$ denote the expectation w.r.t.\ randomness in \apx\grd. Using Lemma \ref{helpapx} (which is stated subsequently), we will show that, 
	\begin{equation}\label{apxgrd1+}
		E[F(\bm{\salg}_t\mid \bm{R}_t)] \geq \eta\,E[F(o_{t}\mid \bm{R}_t)]- (1+2\eta)\,\delta_t\max_{a\in \ac_t}F(a).
	\end{equation}
	First, we prove the main claim under the assumption that inequality \eqref{apxgrd1+} holds for all $t\in[T]$. We have,
	\begin{eqnarray}
		F (\opt_{T+1})&\leq & E[F(\bm{R}_{T+1})]+ \sum_{t\in [T]} E[F(o_{t}\mid \bm{R}_t)],\nonumber\\
		&\leq& E[F(\bm{R}_{T+1})]+ \sum_{t\in [T]} \frac{1}{\eta}\left(E[F(\bm{\salg}_t\mid \bm{R}_t)]+(1+2\eta)\delta_t \max_{a\in \ac_t} F(a)\right),\nonumber\\
		&=& \left(1+\frac{1}{\eta}\right)E[F(\bm{R}_{T+1})]+\left(2+\frac{1}{\eta}\right) \sum_{t\in [T]}\delta_t\,\left( \max_{a\in \ac_t} F(a)\right)\nonumber,\\
		&\leq &\left(1+\frac{1}{\eta}\right)E[F(\bm{R}_{T+1})]+\left(2+\frac{1}{\eta}\right)\,F(\opt_{T+1})\, \sum_{t\in[T]}\delta_t. \nonumber
	\end{eqnarray}
	The first inequality follows by taking expectation over randomness in \apx\grd\ on both sides of inequality \eqref{apxgrd1}. The second inequality follows from \eqref{apxgrd1+}. The (first) equality follows from the fact that $F(\bm{R}_{T+1})=\sum_{t\in[T]}F(\bm{\salg}_t\mid \bm{R}_t)$. To see the final inequality, observe that $F(\opt_{T+1})\geq \max_{t\in[T],\,a\in \ac_t}F(a)$ because the singleton action set $\{a\}$ is a feasible solution for the offline problem. Rearranging the terms on both sides of the final inequality, we obtain,
	\begin{eqnarray}
		\left(1+\frac{1}{\eta}\right)E[F(\bm{R}_{T+1})]&\geq&\left(1-\left(2+\frac{1}{\eta}\right)\sum_{t\in[T]}\delta_t\right)		F (\opt_{T+1}) ,\nonumber\\
		E[F(\bm{R}_{T+1})]&\geq &\left(\frac{\eta}{1+\eta}-\left(\frac{2\eta+1}{\eta+1}\right)\sum_{t\in[T]}\delta_t\right)		F (\opt_{T+1}),\nonumber\\
		&\geq&\left(\frac{\eta}{1+\eta}-\frac{3}{2}\sum_{t\in[T]}\delta_t\right)		F (\opt_{T+1}).\nonumber
	\end{eqnarray}
	Here, we used the fact that $\frac{2\eta+1}{\eta+1}\leq \frac{3}{2}$ for all $\eta\in(0,1]$.

	To complete the proof, we need to establish inequality \eqref{apxgrd1+}. Fix an arbitrary arrival $t\in[T]$. Since \apx\agrd\ draws independent samples for the SAA at every arrival, it suffices to show that conditioned $\bm{R}_t=\alg_t$, we have 
	\begin{equation}\label{apxgrd4}
		E_t[F(\bm{\salg_t}\mid \alg_t)]\geq\eta\,\max_{a'\in \ac_t}F(a'\mid \alg_t)- (1+2\eta)\,\delta_t\max_{a\in \ac_t}F(a),
	\end{equation}	
	here $E_t[\cdot]$ denotes the conditional expectation w.r.t.\ randomness in \apx\grd\ at arrival $t$, given a fixed sample path prior to $t$. Observe that $\max_{a'\in \ac_t}F(a'\mid \alg_t)\geq F(o_t\mid \alg_t)$. 
	Now, the inequality \eqref{apxgrd4} follows directly from Lemma \ref{helpapx} (stated subsequently), thereby completing the proof.
	
	\hfill\Halmos	\end{proof}
\begin{lemma}\label{helpapx}
	Let $a^*$ denote an optimal solution (action) for the \sap\ with weights $w_e$ and 
let $\aapx$ denote the (random) output of an $\eta$-approximation oracle for \sap\ for the (random) instance with weights $\wapx_e$. Then, we have the following inequality,
	\[E_t\left[\sum_{e\in N(\aapx)}p_e\,\wapx_e\right]\geq \eta\,  \sum_{e\in N(a^*)}p_e\,w_e-(1+2\eta)\delta_t\,\max_{a\in \ac_t} F(a).\]
	Here, the expectation is taken over the randomness in $\wapx_e$ at arrival $t$.
\end{lemma}
\begin{proof}{Proof.}

	Consider the event that,
	\[|\wapx_e-w_e|\leq \delta_t \, f(e)\quad \forall e\in N_t.\]
	We use $\textsc{Good}$ to denote this event. From Lemma \ref{chernoff}, the probability of this event is at least $1-\delta_t$. 
	Conditioned on $\textsc{Good}$, we have  w.p.\ 1,
	\begin{eqnarray}
		\sum_{e\in N(a^*)}p_e\, w_e
		&\leq& \sum_{e\in N(a^*)}p_e\, (\wapx_e + \delta_t\, f(e)),\nonumber\\
		&=& \sum_{e\in N(a^*)}p_e\, \wapx_e + \delta_t\sum_{e\in N(a^*)}p_e\, f(e),\nonumber\\
		&\leq& \max_{a'\in \ac_t}\sum_{e\in N(a')}p_e\, \wapx_e + \delta_t\max_{a\in \ac_t}\left(\sum_{e\in N(a)}p_e\, f(e)\right),\nonumber\\
		&\leq& \frac{1}{\eta}\sum_{e\in N(\aapx)}p_e\, \wapx_e + \delta_t\max_{a\in \ac_t}F(a),\nonumber\\
		&\leq& \frac{1}{\eta}\sum_{e\in N(\aapx)}p_e\, (w_e+\delta_t\, f(e)) + \delta_t\max_{a\in \ac_t}F(a),\nonumber\\
		&=& \frac{1}{\eta}\sum_{e\in N(\aapx)}p_e\, w_e + \left(1+\frac{1}{\eta}\right)\,\delta_t\max_{a\in \ac_t}F(a).\label{apxgrd2}
	\end{eqnarray}
	
	Observe that,
	\begin{eqnarray*}
		&&E\left[\sum_{e\in N(\aapx)}p_e\,\wapx_e\right]\\
		&&=  P(\textsc{Good})\,\,E\left[\sum_{e\in N(\aapx)}p_e\,\wapx_e\mid\,\, \textsc{Good}\right]+ (1-P(\textsc{Good}))\,\,	E\left[\sum_{e\in N(\aapx)}p_e\,\wapx_e\mid\,\, \neg\, \textsc{Good}\,\right],\\
		&&\geq  (1-\delta_t)\left[\eta\sum_{e\in N(a^*)}p_e\, w_e - \left(1+{\eta}\right)\,\delta_t\max_{a\in \ac_t}F(a)\right],
	\end{eqnarray*}
	here we used inequality \eqref{apxgrd2} and the fact that $P(\textsc{Good})\geq 1-\delta_t$ to lower bound the first term in the summation and we used the non-negativity of $\wapx_e$ to lower bound the second term by 0. Finally, note that,
	\begin{eqnarray*}
		(1-\delta_t)\left[\eta\sum_{e\in N(a^*)}p_e\, w_e - \left(1+{\eta}\right)\,\delta_t\max_{a\in \ac_t}F(a)\right]	&\geq& (1-\delta_t)\,\eta\sum_{e\in N(a^*)}p_e\, w_e - \left(1+{\eta}\right)\,\delta_t\max_{a\in \ac_t}F(a),\\
		&\geq& \eta\sum_{e\in N(a^*)}p_e\, w_e - \left(1+2{\eta}\right)\,\delta_t\max_{a\in \ac_t}F(a).
	\end{eqnarray*}
	The second inequality follows from the fact that $\sum_{e\in N(a)}p_e\, w_e\leq \sum_{e\in N(a)}p_e\, f(e)\leq F(a)$ for all $a\in \ac_t$.

	\hfill\Halmos\end{proof}

\subsection{Applying the Reduction Technique to \apx\grd}\label{appx:sap2}

\begin{repeattheorem}[Theorem \ref{sap2}.]
	In the adversarial, RO, and UIID models, \apx\grd\ algorithm has the same competitive ratio both with and without stochastic outcomes. 
\end{repeattheorem}

	\begin{proof}{Proof.} 
As mentioned earlier in Remark \ref{beyondg}, to use the reduction technique for an algorithm \aalg, it suffices to establish the invariance property for \aalg. Specifically, for \apx\grd, we need to show that $E[\apx\grd(G)]=E[\apx\grd(\gc)]$, where the expectation is w.r.t.\ the randomness in \apx\grd. At first glance, this may seem challenging because the approximation oracle for the \sap\ and the sample average approximation of weights may lead to the selection of actions from $\h{\curac}$; potentially violating the invariance property. We show that this issue has a simple fix. The key observation is that 
\apx\grd\ is at least $\beta(\eta,\epsilon)$-competitive for \emph{every possible} $\eta$-approximation oracle for the \sap. 	Specifically, it suffices to show that there exists a choice of $\eta$-feasible solutions for the \sap\ instances arising at each arrival of $\gc$ such that $\apx\grd(\gc)=\apx\grd(G)$.


We compare $\apx\grd(G)$ and $\apx\grd(\gc)$ at the sample-path level. To this end, we define a natural coupling of the executions of \apx\grd\ on instances $G$ and $\gc$. For each $t\in[T]$, let $\mathcal{P}_t$ denote a sufficiently large collection of i.i.d.\ random samples from the set $\curn\!\left(\bigcup_{\tau<t} \ac_{\tau}\right)$ of partial mappings. Each mapping in $\curn\!\left(\bigcup_{\tau<t} \ac_{\tau}\right)$ specifies the realized outcome of every action in $\bigcup_{\tau<t} \ac_{\tau}$. We couple the two executions by using the same collection $\mathcal{P}_t$ to compute the SAA weights $w_e^{\mathrm{apx}}$ for all outcomes $e\in N_t$, for each $t\in[T]$, in both instances $G$ and $\gc$. Although the mappings in $\mathcal{P}_t$ may not realize every possible outcome prior to arrival $t$, this is not required. It suffices that the resulting estimates $w_e^{\mathrm{apx}}$ are identical across the two instances and provide sufficiently accurate approximations of $w_e$, which follows from Lemma~16.
		
		Now, fix an arbitrary sample path of the algorithm. For instance $G$, let $\sap_t$ denote the \sap\ instance at arrival $t$ and let $\salg_t$ denote the action selected at $t$. 
		Similarly, let $\h{\sap}_t$ and $\h{\salg}_t$ denote the \sap\ instance and the action selected at arrival $t$ of $\gc$. We will use induction to show that for all $t\in[T]$, there exists an $\eta$-approximate solution to $\h{\sap}_t$ such that $\h{\salg}_t=\salg_t$.  
		
		\paragraph{Base case ($t=1$):} The instances $\h{\sap}_1$ and $\sap_1$ have identical weights but $\h{\sap}_1$ has additional feasible solutions. In the adversarial and RO models, using Lemma \ref{key}, we know that $\h{a}_1$ is not an optimal solution of $\h{\sap}_1$. Therefore, action $\salg_1$ is an $\eta$-approximate solution for $\h{\sap}_1$, and we set $\h{\salg}_1=\salg_1$. In the UIID model, we have several new actions at arrival 1. However, by applying Lemma \ref{keyuiid} -- the UIID specific extension of Lemma \ref{key} -- we reach the same conclusion.
		
		\paragraph{Inductive step:} Suppose that for some $\tau$, we have $\h{\salg}_t=\salg_t$ for all $t<\tau$. Now consider the instance $\h{\sap}_\tau$ and $\sap_\tau$, which have the same weights. Using Lemma \ref{key} with $S=\{\salg_t,\cdots,\salg_{\tau-1}\}$, we have that $\h{a}_\tau$ is not an optimal solution of $\h{\sap}_\tau$. Therefore, action $\salg_\tau$ is an $\eta$-approximate solution of $\h{\sap}_\tau$. By the induction hypothesis, we have $\h{\salg}_t=\salg_t$ for all $t$. Similar to the base case, we use Lemma \ref{keyuiid} to reach this conclusion in the UIID model.
		\hfill\Halmos	
	\end{proof}

	\section{Missing Details for Greedy-like Algorithms}
	\subsection{Extending the Reduction Technique to \wpa}\label{appx:gla}

Recall that Balance is a family of deterministic algorithms for \wpa-DO that favors actions generating the highest \emph{perturbed rewards}. Formally, the algorithm selects the following action at arrival $t$:
	\[\text{Balance:}\quad \argmax_{a\in \ac_t}\, \sum_{i\in I\mid a_i=1} r_i\, u\left(c_i,y_i(t)\right),\]
here $y_i(t)=\sum_{a\in\alg_t}a_i$ is the total capacity of resource $i$ allocated prior to $t$. In the presence of stochastic outcomes, $y_{P,i}(t)$ represents the total capacity of $i$ allocated prior to $t$ under a given (partial) action to outcome mapping $P_t\in \curn(\alg_t)$, and  is given by
\[y_{P_t,i}(t)=\sum_{a\in \alg_t} \sum_{e\in N(a)\cap P_t} e_i.\] 
The non-adaptive verion of Balance for \wpa\ selects the following action $t$,
	\[\text{Non-adaptive Balance:}\quad  \argmax_{a\in \ac_t}\, \sum_{P_t\in \curn(\alg_t)}\gamma(P_t)\sum_{e\in N(a)}p_e\sum_{i\in I\mid e_i=1} \,r_i\,u\left(c_i, y_{P_t,i}(t)\right).\]

	\cite{wholepage} showed that an instance of Balance (Algorithm 1 in \cite{wholepage}) is $(1-1/e)$-competitive for \wpa-DO when $c_{\min}\to +\infty$. Using the reduction technique, we show that non-adaptive Balance has the same competitive ratio for \wpa\ and \wpa-DO. Crucially, this equivalence holds in the large capacity regime because the \wpa-DO instance $\gc$ constructed by the reduction has the same minimum resource capacity $c_{\min}$ as the original \wpa\ instance $G$. As a result, non-adaptive Balance is asymptotically $(1-1/e)$-competitive for \wpa\ in the large capacity regime.
`	We restate our result for (non-adaptive) Balance.

	\begin{repeattheorem}[Theorem \ref{thm:ib}.]
		Any (non-adaptive) Balance algorithm has the same asymptotic competitive ratio for \wpa\ and \wpa\emph{-DO}. 
\end{repeattheorem}	



%
 \begin{proof}{Proof.} 
 
  Let \aalg\ denote an instance of Balance with competitive ratio $\beta$ for \wpa-DO. 
 Consider an arbitrary instance $G$ of \wpa. Since \wpa-DO is a special case of \wpa, it suffices to show that,
 	\[\frac{\aalg(G)}{\opt(G)}\geq \beta.\]
 		We prove this result by extending the reduction technique. Specifically, we show that there exists an instance $\gc$ of \wpa-DO with the same minimum resource capacity such that $\opt(\gc)\geq \optc(G)$ (dominance) and $\aalg(\gc)=\aalg(G)$ (invariance). To prove invariance, we show that \aalg\ solves an instance of the \sap\ at each arrival, relying on the fact that Lemmas \ref{key1} and \ref{key2} hold for arbitrary functions, without requiring monotonicity or submodularity.
 		
 		The key difference is that we must now ensure that $\gc$ is an instance of \wpa-DO, rather than OSW or OSOW. As in the original construction, $\gc$ is defined using an augmented action set
$\no = \curac \cup \hat{\curac}$. However, the objective function requires a
slight modification: we introduce a customized objective function $\hat{F}$
tailored to the \wpa-DO setting. We begin by defining this new objective function and then specify the resulting
instance $\gc$ of \wpa-DO.

\paragraph{Objective Function $\h{F}$:} Let $r_{P_1,P_2,i}=\gamma(P_1)\,\alpha^c(P_2)\, r_i$ for all $P_1\in \curn$, $P_2\in \cal{X}$, and $i\in I$. Consider the objective function,
\begin{equation}\label{hFdef}
	\h{F}(S_1\cup S_2)= \sum_{P_1\in \curn}\sum_{P_2\in \cal{X}}\sum_{i\in I} r_{P_1,P_2,i}\,\min\left\{c_i,\sum_{e\in N(S_1)\cap P_1} e_i+\sum_{ e\in N(S_2)\cap P_2} e_i\right\} 
	\quad \forall S_1\subseteq \curac,\,\, S_2\subseteq \h{\curac},
\end{equation}
where $N(S_2)=\cup_{t\mid \h{a}_t\in S_2} N_t$. We examine the differences between $\h{F}$ and the function $F$ defined in \eqref{def1} in Remark \ref{diff}, which follows this proof. Observe that, 
\begin{eqnarray}
	\h{F}(\h{\curac})&=&\sum_{P_1\in \curn}\sum_{P_2\in \cal{X}}\sum_{i\in I} r_{P_1,P_2,i}\,\min\left\{c_i,\sum_{ e\in P_2} e_i\right\}=\sum_{P_2\in \cal{X}}\sum_{i\in I} \alpha^c(P_2)\,f_i(P_2)=\optc(G),\label{foropt}\\
	\h{F}(S)&=&\sum_{P_1\in \curn}\sum_{P_2\in \cal{X}}\sum_{i\in I} r_{P_1,P_2,i}\,\min\left\{c_i,\sum_{ e\in N(S)\cap P_1} e_i\right\}=\sum_{P\in \curn(S)}\gamma(P)\,f(P) = F(S) \quad \forall S\subseteq \curac.\label{foraalg}
\end{eqnarray}
It remains to show that $\h{F}$ is the objective for an instance of \wpa-DO with action set $\no$.

\paragraph{Instance of \wpa\emph{-DO}:}  
Consider an instance of \wpa-DO with the set of resources $\h{I}=\{(P_1,P_2,i)\mid P_1\in \curn,\, P_2\in \cal{X},\,i\in I \}$, where the resource $(P_1,P_2,i)\in \h{I}$ has capacity $c_i$ and per unit reward $r_{P_1,P_2,i}$. 
The (deterministic) outcome of an action $a\in \no$ in this instance is a binary vector  $\h{e}_a=(\h{e}_{a,P_1,P_2,i})_{(P_1,P_2,i)\in \h{I}}$ with one component per resource in $\h{I}$. Fix $P_1,P_2,$ and $i$, then for all $a\in \curac$, the component $\h{e}_{a,P_1,P_2,i}=e_i$, where $e_i$ is the resource $i$ component of the outcome $e\in N(a)\cap P_1$ in instance $G$. Here, $N(a)\cap P_1$ is a singleton. Similarly, for every action $\h{a}_t\in \h{\curac}$, we have $\h{e}_{\h{a}_t,P_1,P_2,i}=e_i$ where $e_i$ is the resource $i$ component of the outcome $e\in N_t\cap P_2$. Given these definitions, the reward function of resource $(P_1,P_2,i)$ in this instance of \wpa-DO is given by,
\[f_{P_1,P_2,i}(X)=r_{P_1,P_2,i}\min\left\{c_i,\sum_{a\in \no \mid \h{e}_a\in X}\h{e}_{a,P_1,P_2,i}\right\}.\]
It follows that the total reward across all resources is given by the function $\h{F}$ defined in \eqref{hFdef}. Specifically, for all action sets $S_1\subseteq \curac$ and $S_2\subseteq \h{\curac}$,
\begin{eqnarray*}
	\sum_{P_1\in \curn}\sum_{P_2\in \cal{X}}\sum_{i\in I}f_{P_1,P_2,i}(\cup_{a\in S_1\cup S_2} \h{e}_{a})&=&\sum_{P_1\in \curn}\sum_{P_2\in \cal{X}}\sum_{i\in I} r_{P_1,P_2,i}\,\min\left\{c_i,\sum_{e\in N(S_1)\cap P_1} e_i+\sum_{ e\in N(S_2)\cap P_2} e_i\right\},\\
	&=&\h{F}(S_1\cup S_2) 
\end{eqnarray*}


Now, observe that $\opt(\gc)\geq \h{F}(\h{\curac})$, because $\h{\curac}$ is a feasible solution to $\gc$. From \eqref{foropt}, we have that $\h{F}(\h{\curac})=\optc(G)$. Thus, $\opt(\gc)\geq \optc(G)$ and we have the dominance property. 
If the invariance property is also true, then we have,
\[\frac{\aalg(G)}{\opt(G)}=\frac{\aalg(\gc)}{\opt(G)}\geq \frac{\aalg(\gc)}{\opt(\gc)}\geq \beta.\]
\paragraph{Proof of the Invariance Property:} To show that $\aalg(G)=\aalg(\gc)$, it suffices to show that \aalg\ produces the same output on both $G$ and $\gc$. This is sufficient because $\h{F}$ and $F$ are identical when restricted to the ground set $\curac$ (see \eqref{foraalg}). For the randomized PG algorithm, we assume that the random values $u_i\sim D_i$ have been fixed arbitrary for all $i\in I$, and we prove invariance on every sample path.

 Let $t\in[T]$ be an arbitrary arrival, and suppose that, prior to arrival $t$, \aalg\ selects the same set of actions $S\subseteq \curac$ on both $G$ and $\gc$. We will show that \aalg\ selects the same action at arrival $t$ in both cases. Assuming this holds, $\aalg(G)=\aalg(\gc)$ follows by induction over $t$. 

We begin by describing the action selected by  (non-adaptive) \aalg\ on instance $G$. At arrival $t$, the algorithm selects
\[a_G=\argmax_{a\in \ac_t} \sum_{P\in \curn(S)}\gamma(P) 
\sum_{e\in N(a)} p_e\sum_{i\in I\mid e_i=1} \,r_i\,u\left(c_i, y_{P,i}(t)\right).\]
For brevity, let
\[\tilde{f}(e\mid P)= \sum_{i\in I\mid e_i=1} \,r_i\,u\left(c_i, y_{P,i}(t)\right).\] 
Using Lemma \ref{key1}, which holds for any function $\tilde{f}$, we rewrite $a_G$ as follows,
\[a_G=\argmax_{a\in \ac_t} \sum_{e\in N(a)} p_e \, \left(\sum_{P\in \curn(S)}\gamma (P) \tilde{f}(e\mid P)\right).\] 

We now turn to the action selected by \aalg\ on the constructed instance $\gc$. At arrival $t$, \aalg\ selects
\begin{eqnarray*}
	a_{\gc}&=&\argmax_{a\in \ac_t\cup \{\h{a}_t\}}\sum_{P_1\in \curn}\sum_{P_2\in \cal{X}}\sum_{i\in I\,\mid\, \h{e}_{a,P_1,P_2,i}\,=\,1}r_{P_1,P_2,i}\,u\left(c_{P_1,P_2,i},\sum_{a'\in S}\h{e}_{a',P_1,P_2,i}\right),\\
	&=&\argmax_{a\in \ac_t\cup \{\h{a}_t\}}\sum_{P_1\in \curn}\sum_{P_2\in \cal{X}}\gamma(P_1)\alpha^c(P_2)\sum_{i\in I\,\mid\, \h{e}_{a,P_1,P_2,i}\,=\,1}r_{i}\,u\left(c_{i},\sum_{a'\in S}\h{e}_{a',P_1,P_2,i}\right),\\
		&=&\argmax_{a\in \ac_t\cup \{\h{a}_t\}}\sum_{P_1\in \curn}\sum_{P_2\in \cal{X}}\gamma(P_1)\alpha^c(P_2)\sum_{i\in I\,\mid\, \h{e}_{a,P_1,P_2,i}\,=\,1}r_{i}\,u(c_{i},y_{P_1\cap N(S),i}(t)),
\end{eqnarray*}
where the second equality follows from the construction of $\gc$. The third equality follows from the fact that $S\subseteq \curac$, which implies
\[\sum_{a'\in S} \h{e}_{a',P_1,P_2,i}=\sum_{a'\in S} \sum_{e\in N(a')\cap P_1} e_i=y_{P_1\cap N(S),i}(t).\]

To conclude that $a_{\gc}=a_G$, it suffices to show that the marginal value of the auxiliary action $\h{a}_t$ is a convex combination of the marginal values of actions in $\ac_t$. This essentially follows from Lemma \ref{key2} and we repeat the argument below for completeness. 
\begin{eqnarray*}
&\sum_{P_1\in \curn}\sum_{P_2\in \cal{X}}&\gamma(P_1)\alpha^c(P_2)\sum_{i\in I\,\mid\, \h{e}_{\h{a}_t,P_1,P_2,i}\,=\,1}r_{i}\,u(c_{i},y_{P_1\cap N(S),i}(t))\\
&=&\sum_{P_1\in \curn}\sum_{P_2\in \cal{X}}\gamma(P_1)\alpha^c(P_2)\sum_{i\in I}(\h{e}_{\h{a}_t,P_1,P_2,i})\,r_{i}\,u(c_{i},y_{P_1\cap N(S),i}(t))\\
&=&\sum_{P_1\in \curn}\gamma(P_1)\sum_{i\in I}r_{i}\,u(c_{i},y_{P_1\cap N(S),i}(t))\sum_{P_2\in \cal{X}}\alpha^c(P_2)\,\h{e}_{\h{a}_t,P_1,P_2,i}\\
&=&\sum_{P_1\in \curn}\gamma(P_1)\sum_{i\in I}r_{i}\,u(c_{i},y_{P_1\cap N(S),i}(t))\sum_{P_2\in \cal{X}}\sum_{e\in N_t\cap P_2}\alpha^c(P_2)e_i\\
&=&\sum_{P\in \curn(S)}\gamma(P)\sum_{i\in I}r_{i}\,u(c_{i},y_{P,i}(t))\sum_{e\in N_t\mid e_i=1}\left(\sum_{P_2\in \cal{X}\mid P_2\ni e}\alpha^c(P_2)\right)\\
&=&\sum_{P\in \curn(S)}\gamma(P)\sum_{a\in \ac_t,\, e\in N(a)}\sum_{i\in I\mid e_i=1}r_{i}\,u(c_{i},y_{P,i}(t))\left(y^c_a p_e\right)\\
&=&\sum_{a\in \ac_t}y^c_a\sum_{e\in N(a)}p_e\left(\sum_{P\in \curn(S)}\gamma(P)\tilde{f}(e\mid P)\right)\\
&\leq &\sum_{e\in N(a_G)} p_e \, \left(\sum_{P\in \curn(S)}\gamma (P) \tilde{f}(e\mid P)\right)\\
	&=& \sum_{P\in \curn(S)}\gamma(P)\sum_{e\in N(a_G)}p_e\sum_{i\in I\mid e_i=1}r_{i}\,u(c_{i},y_{P,i}(t))\\
		&=& \sum_{P_1\in \curn}\gamma(P_1)\sum_{e\in N(a_G)\cap P_1}\sum_{i\in I\mid e_i=1}r_{i}\,u(c_{i},y_{P_1\cap N(S),i}(t))\\
	&=& \sum_{P_1\in \curn}\sum_{P_2\in \cal{X}}\gamma(P_1)\alpha^c(P_2)\sum_{e\in N(a_G)\cap P_1}\sum_{i\in I\mid e_i=1}r_{i}\,u(c_{i},y_{P_1\cap N(S),i}(t))\\
		&=&\sum_{P_1\in \curn}\sum_{P_2\in \cal{X}}\gamma(P_1)\alpha^c(P_2)\sum_{i\in I\,\mid\, \h{e}_{a_G,P_1,P_2,i}\,=\,1}r_{i}\,u(c_{i},y_{P_1\cap N(S),i}(t)).
\end{eqnarray*}
\hfill\Halmos\end{proof}
\begin{remark}[Difference between $\hat{F}$ and $F$]\label{diff}
	Recall that
	\[
	\hat{F}(S_1 \cup S_2)
	=
	\sum_{P_1 \in \curn}
	\sum_{P_2 \in \cal{X}}
	\sum_{i \in I}
	r_{P_1,P_2,i}
	\min\left\{
	c_i,\;
	\sum_{e \in N(S_1) \cap P_1} e_i
	+
	\sum_{e \in N(S_2) \cap P_2} e_i
	\right\},
	\]
	whereas, from~\eqref{def1}, we have
	\[
	F(S_1 \cup S_2)
	=
	\sum_{P_1 \in \curn}
	\sum_{P_2 \in \cal{X}}
	\sum_{i \in I}
	r_{P_1,P_2,i}
	\min\left\{
	c_i,\;
	\sum_{e \in (N(S_1) \cap P_1)\,\cup\, (N(S_2) \cap P_2)} e_i
	\right\}.
	\]
	
	Observe that
	\[
	\sum_{e \in (N(S_1) \cap P_1)\,\cup\, (N(S_2) \cap P_2)} e_i
	\;\le\;
	\sum_{e \in N(S_1) \cap P_1} e_i
	+
	\sum_{e \in N(S_2) \cap P_2} e_i.
	\]
	Consequently,
	\[
	F(S_1 \cup S_2) \;\le\; \hat{F}(S_1 \cup S_2).
	\]
	
	This distinction is crucial. While $\hat{F}$ corresponds to the objective
	function of a valid instance of \wpa-DO, the function $F$ does not, in general,
	correspond to any feasible \wpa-DO instance.
	
	To illustrate this, consider a single resource $i$ with capacity $c_i = 2$, and
	singleton sets $S_1 = \{a\}$, $S'_1=\{a'\}$ and $S_2 = \{\hat{a}\}$. Suppose there exist mappings
	$P_1$ and $P_2$ such that
	$N(S_1) \cap P_1 = N(S_2) \cap P_2 = \{e\}$ and $N(S'_1)\cap P_1=\{e'\}$. Then,
	\[
	\hat{F}(S_1 \cup S_2)
	=\hat{F}(S_1 \cup S'_1)=\hat{F}(S'_1 \cup S_2)
	=\,2r_i,
	\]
	whereas
	\[
	F(S_1 \cup S_2)
	=
	r_i, \quad \text{ and } \quad
	{F}(S_1 \cup S'_1)={F}(S'_1 \cup S_2)
	=\,2r_i.
	\]
	In a valid instance of \wpa-DO, the objective behaves linearly when there is sufficient remaining resource capacity. The subadditive behavior exhibited by $F$ in this example cannot arise from any valid instance of
	\wpa-DO, highlighting the difference between $\hat{F}$ and $F$.
\end{remark}

{\color{black}
\subsection{Competitive Ratio Upper Bounds Using Reduction Technique}\label{appx:upb}
Given a set $I$ of resources, let $F=\sum_{i\in I} F_i$ denote the objective function in an instance of OSW-SO, where $F_i(S)=\sum_{P\in \curn(S)} \gamma(P) f_i(P)$ is the expected total reward from resource $i\in I$. We assume that $f=\sum_{i\in I}f_i$ and $F=\sum_{i\in I} F_i$ are monotone submodular functions but make no assumptions on the component functions $f_i$ and $F_i$. 

We are interested in \gla s for OSW that are greedy w.r.t. the perturbed objective function $\tilde{F}=\sum_{i\in I}\tilde{F}_i$. Specifically, at arrival $t$, the algorithm selects	
\[\argmax_{a\in \alg_t} \sum_{i\in I}\tilde{F}_i(a\mid \alg_t), \]
where
\[\tilde{F}_i(a\mid \alg_t) = u_i\left(\{F_j(X)\}_{j\in I,X \subseteq \alg_t}\right) F_i(a\mid \alg_t),\]
and each $u_i$ is an arbitrary (possibly randomized) function of the current state of resource $i$, represented by the collection $\{F_j(X)\}_{j\in I,X\subseteq \alg_t}$. 
\begin{repeattheorem}[Theorem \ref{thm:upb}.]
	The family of \gla s defined above has competitive ratio at most 0.5 for  OSW.
	\end{repeattheorem}}
\begin{proof}{Proof.} {\color{black}
	To prove this theorem we leverage Theorem 5 from \cite{deb}, which states that when $c_{\min}=1$, no non-adaptive algorithm has a competitive ratio higher than 0.5 for the problem of online matching with stochastic rewards (OMSR), a special case of  OSW-SO. By applying the reduction technique in the reverse direction, we show that this result imposes an upper bound of 0.5 on the competitive ratio of all \gla s for  OSW.

	Formally, we argue by contradiction. Suppose that there exists a \gla\ for  OSW with competitive ratio $0.5+\epsilon$ for some $\epsilon>0$. In the first part of the proof, we use the reduction technique to show that this implies the existence of a non-adaptive \gla\ with competitive ratio $0.5+\epsilon$ for  OSW-SO against the benchmark $\optc$. Therefore, there exists a non-adaptive algorithm with competitive ratio $0.5+\epsilon$ (against $\optc$) for OMSR. At this point, one might expect an immediate contradiction with Theorem~5 of \cite{deb}, which rules out competitive ratios above 0.5 for non-adaptive OMSR algorithms when compared against an LP benchmark that upper bounds \opt. However, since the result of \cite{deb} is stated with respect to an LP benchmark, an additional step is required. In the second part of the proof, we strengthen their result by showing that the same 0.5 upper bound continues to hold even when performance is measured against $\optc$. 

\paragraph{Part 1:} In this part, we use the reduction technique to show that a non-adaptive \gla\ has the same competitive ratio for  OSW-SO and OSW. Let $G$ be an instance of  OSW-SO with objective functions $f=\sum_{i\in I} f_i$ and $F=\sum_{i\in I} F_i$. Using the original construction from Lemma \ref{redprop}, we construct an instance $\gc$ of OSW with an extended objective function defined over the enlarged ground set $\no$. 

Specifically, the extended objective function is
\[F(S)=\sum_{P_1\in \curn}\sum_{P_2\in \mathcal{X}} \gamma(P_1) \, \alpha^c(P_2)\, f\left((N(S_1)\cap P_1)\cup (N(S_2)\cap P_2)\right),\]
for all $S\subseteq \no$, where we define $S_1=S\cap \curac$ and $S_2= S\cap \h{\curac}$. We retain the same resource set $I$ and define corresponding extended component functions $F_i$ so that 
\[F_i(S)=\, \sum_{P_1\in \curn}\sum_{P_2\in \mathcal{X}} \gamma(P_1) \, \alpha^c(P_2)\, f_i\left((N(S_1)\cap P_1)\cup (N(S_2)\cap P_2)\right)\quad \forall i\in I.\]
This construction defines a valid extension of the original component functions: for every set $S\subseteq \curac$< the value of $F_i(S)$ remains unchanged, and the extended objective decomposes as $F=\sum_{i\in I} F_i$. The dominance property holds by construction of $\gc$. To complete the reduction, it remains to establish the invariance property holds for the family of \gla s. As in Lemma \ref{redprop}, we show that the \gla\ selects the same sequence of actions on both $G$ and $\gc$.

To this end, consider the \gla\ applied to the instance $\gc$, which operates using the extended perturbed marginal value functions of the form
\[
\tilde{F}_i(a \mid S)
=
u_i\!\left(\{F_j(X)\}_{j \in I,\; X \subseteq S}\right)\, F_i(a \mid S),
\qquad \forall\, S \subseteq \no,\; i \in I.
\]
When the functions $u_i$ are randomized, we define a natural coupling between their instantiations on the two instances. Specifically, we couple the random choices underlying $u_i$ so that, for every set $S \subseteq \curac$, the value
\[
u_i\!\left(\{F_j(X)\}_{j \in I,\; X \subseteq S}\right)
\]
is identical when evaluated on $G$ and on $\gc$. This coupling is well defined because, for every $S \subseteq \curac$, the collection of values $\{F_j(X)\}_{j \in I,\; X \subseteq S}$ coincide across the two instances. The restriction to subsets of $S$ is essential: if $u_i$ were allowed to depend on all actions revealed at (and prior to) arrival $t$, i.e., $\cup_{\tau\leq t} (\ac_\tau\cup \{\h{a}_{\tau}\})$, then the presence of additional actions in $\gc$ could change the perturbation values, breaking the coupling and preventing us from establishing invariance. 

It suffices to show that for any set $S\subseteq \curac\backslash \ac_t$,
\begin{equation}\label{tildeconv}
	\tilde{F}(\h{a}_t\mid S)= \sum_{a\in \ac_t} y_a \tilde{F}(a\mid S).
\end{equation}
As an immediate consequence, we obtain $\tilde{F}(\h{a}_t\mid S)\leq \max_{a\in \ac_t} \tilde{F}(a\mid S)$, and the invariance property follows by induction over $t$, exactly as in the proof of Lemma \ref{redprop}. We now verify \eqref{tildeconv},
\begin{eqnarray*}
		\tilde{F}(\h{a}_t\mid S)&=& \sum_{i\in I} 	\tilde{F}_i(\h{a}_t\mid S),\\
		&=& \sum_{i\in I} u_i(\{F_j(X)\}_{j\in I,X\subseteq S}) F_i(\h{a}_t\mid S),\\
		&=& \sum_{i\in I} u_i(\{F_j(X)\}_{j\in I,X\subseteq S}) \left(\sum_{a\in \ac_t} y_a F_i(a\mid S)\right),\\
		&=& \sum_{i\in I} \sum_{a\in \ac_t} y_a u_i(\{F_j(X)\}_{j\in I,X\subseteq S})F_i(a\mid S),\\
		&=&\sum_{a\in \ac_t} y_a \sum_{i\in I} \tilde{F}_i(a\mid S).
	\end{eqnarray*}
All equalities follow directly from the definitions, except the third, which follows from Lemma \ref{key2} applied to the functions $f_i$ and $F_i$. Importantly, Lemma \ref{key2} holds for arbitrary functions that are not necessarily non-negative, monotone, or submodular. 
}

  \paragraph{Part 2:} In this part, we extend Theorem 5 of \cite{deb} and show that the upper bound of 0.5 holds even when we compare non-adaptive algorithms against $\optc$. \cite{deb} showed the upper bound using the following family of `hard' instances of OMSR. Consider a bipartite graph between the set of resources, $[n]$, and arrivals $[T]$. Let $E$ denote the set of edges. Every edge in $E$ has success probability $p$ and assume that $\frac{1}{p}$ is an integer. The first $\frac{1}{p}$ arrivals have an edge to every resource. The second $\frac{1}{p}$ arrivals have an edge to every resource except resource 1. In general, arrival $t\in \{ \frac{i-1}{p}+1,  \frac{i}{p} \}$ has an edge to every resource in $\{i,\cdots,n\}$, and $T=\frac{n}{p}$. 
 In a simplified view, they showed that as $p\to0$,  the total expected reward of a non-adaptive algorithm cannot exceed $0.5n +o(n)$. 	However, they compare online algorithms for OMSR with the following LP upper bound on \opt.
	\begin{eqnarray*}
		\max_{x_{i,t}\,\, \forall (i,t)\in E}	&&	\sum_{(i,t)\in E} p\, x_{i,t},\\
		&& \sum_{i \mid (i,t)\in E} x_{i,t} \leq 1\quad \forall t\in [T],\\
		&& \sum_{t \mid (i,t)\in E} p\, x_{i,t} \leq 1\quad \forall i\in [n],\\
		&&x_{i,t}\geq 0\quad \forall (i,t)\in E.
	\end{eqnarray*}	
	It is easy to see that, for the family of instances described above, the LP has optimal value $n$ which is achieved by the following optimal solution: for all $i\in[n]$ we can set $x_{i,t}=1$ if $t$ is one of the last $\frac{1}{p}$ arrivals that have an edge to $i$ and $x_{i,t}=0$ for all other arrivals.
	
	We show that the benchmark $\optc$ has value at least $n$ for these instances by finding a feasible solution with value $n$ for the optimization problem in $\optc$. For $k\in[\frac{1}{p}]$, let $X_k$ denote the outcome where edge $(i,\frac{i-1}{p}+k)$ succeeds for all $i\in [n]$ and all other edges fail. Consider the probability distribution $\alpha$ such that $\alpha(X)=p$ if $X=X_k$ for some $k\in [\frac{1}{p}]$ and $\alpha(X)=0$ otherwise. Observe that,
	\[\sum_{X\in \cal{X}} \alpha(X)=\sum_{k\in [\frac{1}{p}]} \alpha(X_k)=1.\]
	For all $i\in [n]$, let $y_{(i,t)}=1$ if $t\in \{ \frac{i-1}{p}+1,  \frac{i}{p} \}$ and let $y_{(i,t)}=0$ otherwise. Clearly, $(y_{(i,t)})_{(i,t)\in E}\in \Delta(E)$ and 
	\[\sum_{k\in [\frac{1}{p}] \mid (i,t)\in X_k}\alpha(X_k)=y_{(i,t)}\, p\quad \forall (i,t)\in E. \]
	We have $f(X_k)=n$ for all $k\in [\frac{1}{p}]$ because every resource is matched in $X_k$. Thus,  $\sum_{X\in \cal{X}}\alpha(X) f(X)=\frac{1}{p} \, p\, n=n$, which concludes the proof.

\hfill\Halmos\end{proof}
{\color{black}
\subsection{Perturbed Greedy (PG) Algorithm}
\subsubsection{PG Algorithms Captured by \gla.}\label{appx:pg1}
Inspired by the RANKING algorithm of \cite{kvv}, the Perturbed Greedy algorithm for vertex-weighted online bipartite matching (OBM) \citep{goel} matches each arrival $t$ to an available resource with the highest \emph{perturbed reward}, given by
\[
r_i \bigl(1 - e^{y_i}\bigr),
\]
where each $y_i$ is sampled independently and uniformly at random from $[0,1]$. This algorithm fits naturally within the family of \gla s introduced in Section~\ref{sec:upb}, with perturbation functions $u_i = 1 - e^{y_i}$ for all $i \in I$. 

Several works have studied natural generalizations of this approach to settings with non-unit inventories, stochastic rewards, and possibly non-binary bids \citep{albers,vazirani,unknown,stochrew}, all of which are captured as special cases of the \gla\ framework. More broadly, this framework also encompasses the Perturbed Greedy algorithm of \cite{jin}, which applies to settings where the reward associated with each resource is given by a matroid rank function.

\subsubsection{Perturbed Greedy (PG) Algorithm for \wpa-DO.}\label{appx:pg2}
	Consider a natural extension of the family of PG algorithms for the \wpa-DO formulation. Given the set of previously selected actions $\alg_t$, a PG algorithm selects the following action at arrival $t$:
	\begin{eqnarray}\label{pg}
		\argmax_{a\in \ac_t}\, \sum_{i\in I\mid a_i=1} u_i\,r_i\,\min\{1,c_i-y_i(t)\}.
	\end{eqnarray}
	where $\{u_i\}_{i\in I}$ are independent random perturbations drawn from a distribution $U$. We recover the Greedy algorithm when $u_i=u=1$ for each resource $i$.	
	
	A natural question is whether there exists a PG algorithm with competitive ratio exceeding 0.5 for \wpa-DO. We answer this question in the negative by presenting a simple instance showing that the family of PG algorithms defined by \eqref{pg} has competitive ratio at most 0.5. 

	Without loss of generality, assume $u_i\in[0,1]$ for all $i\in I$; this can always be achieved by appropriately scaling the rewards $r_i$. We further assume that the distribution $U$ is regular in the sense that $E[u_i]\coloneqq\mu_i>0$. Now,  consider an instance of \wpa-DO with two disjoint groups of resources, $I_1$ and $I_2,$ each containing $m$ resources, where $m$ is sufficiently large. All resources have unit capacity. Each resource in $I_1$ has per-unit reward $1+\epsilon$, for some arbitrarily small $\epsilon>0$, while each resource in $I_2$ has per-unit reward 1. There are two arrivals. The first arrival has two feasible actions $\{a_1,a_2\}$. Action $a_1$ uses a unit of every resource in $I_1$, and action $a_2$ uses a unit of every resource in $I_2$. The second arrival has a single feasible action $a'_1$ that uses a unit of every resource in $I_1$. The optimal offline solution selects actions $a_2$ and $a'_1$, achieving a total reward of $2m$.
		
		Under a PG algorithm, the perturbed reward of action $a_1$ is $(1+\epsilon)\sum_{i\in I_1} u_i$. By the law of large numbers, this quantity concentrates  within a factor of $(1\pm \delta)$ of $(1+\epsilon)mE[u_i]$ with high probability, where $\delta\to0$ as $m\to \infty$. Consequently, for any fixed $\epsilon$, there exists a sufficiently large $m$ such that the PG algorithm selects action $a_1$ with high probability. As a result, the algorithm obtains total reward of $(1+\epsilon)m$, yielding a competitive ratio of at most $(1+\epsilon)\, 0.5$, which can be made arbitrarily close to $0.5$.

Finally, we note that applying the reduction technique for \wpa\ to provide upper bounds on the performance of the general family of \gla s for \wpa-DO leads to a violation of the invariance property. Specifically, the reduction for \wpa\ introduces new resources (see Appendix \ref{appx:gla}), and without more restrictive assumptions on the perturbation functions $u_i$, a \gla\ may no longer select the same sequence of actions on the original instance $G$ and the reduced instance $\gc$. As a result, the extended reduction technique developed in Appendix \ref{appx:gla} cannot be directly used to transfer performance guarantees for general \gla s from \wpa\ to \wpa-DO.

	}
			
	

\section{Missing Details for Non-monotone Functions}\label{appx:onsw}

\subsection{Non-Adaptive Algorithm: Challenges and Potential Approaches}\label{appx:onswnonadap}
In this section, we discuss the challenges of extending the reduction technique to analyze a non-adaptive algorithm for ONSW-SO. First, as noted in Remark \ref{redfail}, Cascade Sampling does not satisfy the invariance property. To address this issue, we consider a simpler algorithm called Greedy Sampling (Algorithm \ref{sinsam}).
\smallskip

\begin{algorithm}[H]
	\SetAlgoNoLine
	\KwIn{Parameter $p\in[0,1]$;}
	Set $\alg_1 = \emptyset$\;
	\For{every arrival $t\in[T]$}{
		Find element with maximum marginal value, $a_t=\argmax_{a\in \ac_t}\, F(a\mid \alg_t)$\;
		Choose action $a_t$ w.p.\ $p$ and choose the null action w.p.\ $1-p$\;
		Let $\salg_t$ denote the chosen action and set $R_{t+1}=R_{t}\cup \{\salg_t\}$\;
	}
	\caption{Greedy Sampling (Non-adaptive)}
	\label{sinsam}
\end{algorithm}
\smallskip

At every arrival, Algorithm \ref{sinsam} selects the greedy action $a_t$ w.p.\ $p$ and selects the null action ($0_t$) w.p.\ $1-p$. The action $a_t$ has non-negative marginal value because $F(a_t\mid R_t)\geq F(0_t\mid \alg_t)=0$. The sampling probability $p$ is a parameter that influences the competitive ratio and can be set to any value in $[0,1]$.   
Choosing $p=1$ gives us the deterministic Greedy algorithm which has a competitive ratio of 0 for ONSW 
\citep{nuti}. 

\begin{remark}
	Algorithm \ref{sinsam} was introduced by \cite{harshaw2022power}, who showed that it is $\frac{p(1-p)}{1+p}$-competitive for ONSW. Note that $\max_{p\in[0,1]}\frac{p(1-p)}{1+p}=3-2\sqrt{2}$. In Appendix \ref{appx:sinsam}, we provide a short and simplified proof of this result by extending the proof template for Theorem \ref{advmain}, which may be of independent interest. 
\end{remark}

\subsubsection*{Challenges with Extending the Reduction Technique:} 
Given an instance $G$ of ONSW-SO, consider the instance $\gc$ as defined in Section \ref{sec:pas}. It can be verified that Greedy Sampling satisfies the invariance property. Specifically, when it selects an action with a positive value, it does so greedily. By extending Lemma \ref{key}, we can show that the greedy action $a_t$ lies in the set $\curac$ for each arrival $t\in[T]$.  

However, as mentioned earlier, another challenge in extending the reduction technique to ONSW-SO is that the objective $F$ (as defined in \eqref{def1}) may not be submodular when $f$ is non-monotone. Specifically,  in the proof of Lemma \ref{gcprop1} (included in Appendix \ref{appx:gcprop1}), we rely on both the \emph{monotonicity} and submodularity of $f$ to establish that $F$ is submodular. 

We believe that this a fundamental issue. To pinpoint the source of the problem, consider the extreme case where $G$ has deterministic outcomes. In this case, the optimal offline solution, \opt, is a deterministic subset of $\curac$. Let $o_t$ denote the action selected at arrival $t\in[T]$ in $\opt$.  To define instance $\gc$, we extend $F$ over $\curac\cup \h{\curac}$ such that, for each arrival $t$, the marginal values of the new action $\h{a}_t$ mimic those of action $o_t$. This ensures that \grd\ (and Greedy-like algorithms) do not select actions in $\h{\curac}$, while also guaranteeing that $F(\h{\curac})=\opt(G)$. 

When $F$ is monotone and submodular on the ground set $\curac$, we can accomplish this extension while preserving the monotonicity and submodularity properties by setting,
\[F(\h{a}_t\mid X)=F(o_t\mid X)\quad \forall X\subseteq \curac.\]
In particular, we have $F(\h{a}_t\mid X\cup o_t)=0$. However, when $F$ is non-monotone, we must have $F(\h{a}_t\mid X\cup o_t)\leq F(\h{a}_t\mid X)$ in order to preserve submodularity over the expanded ground set. This presents a problem: if $F(\h{a}_t\mid X)$ is negative and $F(X\cup o_t)=0$, we cannot guarantee the non-negativity of $F$ on the expanded ground set (see Example \ref{fineg}). In other words, for non-monotone functions, it is unclear whether we can extend $F$ over the new ground set in a way that preserves both its non-negativity and submodularity, while also ensuing that the marginal values of the new actions mimic those of certain existing actions.
\begin{example}\label{fineg}
	
	Consider a ground set $\curac=\{1,2\}$ and let $F(\emptyset)=F(1)=0$, $F(2)=1$, and $F(\{1,2\})=0$. 
	We add a new element $\h{1}$ and extend $F$ such that, $F(\h{1}\mid 2)=F(1\mid 2)=-1$. To ensure submodularity of $F$, we need that, $F(\h{1}\mid \{1,2\})\leq F(\h{1}\mid 2)=-1$. Observe that the resulting function has a negative value, $F(\{1,\h{1},2\})\leq F(2)+2F(1\mid 2)=-1$.
	
\end{example}

\subsection{Simplified Analysis of Greedy Sampling}\label{appx:sinsam}
In the analysis of many of the algorithms in this paper, we use monotonicity to argue that the (set) union of the optimal solution and the solution of an online algorithm has at least as much value as the optimal solution alone. This is not true for non-monotone functions and 
the following lemma is the main ingredient that fills the technical gap created in the absence of monotonicity. 
\begin{lemma}[Lemma 2.2 in \cite{buchbinder2014submodular}]\label{subsample}
	Consider a non-negative submodular function $w$ over ground set $W$. Let $W(p)$ denote a random subset of $W$ where each element appears with probability at most $p$ (not necessarily independently). Then, 
	\[E[w(W(p)\cup S)]\geq (1-p)\,w(S)\quad \forall S\subseteq W.\]
\end{lemma}
	Lemma \ref{subsample} does not hold for submodular order functions and finding a suitable substitute for this lemma appears to be a challenging technical problem.
\begin{theorem}\label{onsw1}
	Algorithm \ref{sinsam} with sampling probability $p\in[0,1]$ is $\frac{p(1-p)}{1+p}$-competitive for ONSW with adversarial arrivals. 
\end{theorem}

\begin{proof}{Proof.}
	Let $\alg_{T+1}$ denote the final (random) set of actions output by Algorithm \ref{sinsam}. Observe that each action $a\in \curac$ appears in $\alg_{T+1}$ with probability at most $p$. Let $\opt_{T+1}$ denote the optimal offline solution, i.e., $\opt_{T+1}=\argmax_{a_t\in \ac_t\,\, \forall t\in [T]} F(\cup_{t\in [T]} \{a_t\} )$. Finally, for $t\in [T]$, let $\{\salg_t\}=\alg_{T+1}\cap \ac_t$,\, $\alg_t=\{\salg_1,\cdots,\salg_{t-1}\}$, and $\{o_t\}=\opt_{T+1}\cap \ac_t$. Let $E[\cdot]$ denote expectation over the randomness in Algorithm \ref{sinsam}. 
	\begin{eqnarray*}
		(1-p)\,F (\opt_{T+1})&\leq &E[F\left(\alg_{T+1}\cup \opt_{T+1}\right)]\\
		&=&	E\left[F\left(\alg_{T+1} \cup \left( \cup_{t\in [T]} \{o_{t}\}\backslash \{\salg_{t}\}\right)\right)\right],\\
		&\leq & E[F(\alg_{T+1})]+ \sum_{t\in [T]} E[F\left(\{o_{t}\}\backslash \{\salg_{t}\}\mid \alg_{t}\right)],\\
		&\leq &E[F(\alg_{T+1})]+ \sum_{t\in [T]}  \frac{1}{p}E[F(\salg_t\mid \alg_{t})],\\
		&=& \left(1+\frac{1}{p}\right)\,E[F(\alg_{T+1})].
	\end{eqnarray*}
	Assuming correctness, this completes the proof. Now, the first inequality follows from Lemma \ref{subsample} with $w=F$, $W(p)=\alg_{T+1}$ and $S=\opt_{T+1}$. The second inequality follows by linearity of expectation and submodularity of $F$. The third inequality follows from the following claim,
	\[E[F(\salg_t\mid \alg_{t})]\geq \,p\,E[F\left(\{o_{t}\}\backslash \{\salg_{t}\}\mid \alg_{t}\right)]\quad \forall t\in [T].\]
	To show this claim, we fix $t\in [T]$ and let $a_t=\argmax_{a\in \ac_t} F(a\mid \alg_t)$. Note that $F(a\mid \alg_t)\geq \max\{0,\, F\left(\{o_{t}\}\backslash \{\salg_{t}\}\mid \alg_{t}\right)\}$ because $0_t\in A_t$. Let $E_{t}[\cdot]$ denote expectation w.r.t.\ the random sampling of action $a_t$. We have, 
	\[E_{t}[F(\salg_t\mid \alg_t)]=\, p\,F(a_t\mid \alg_t) + (1-p) F(0_t\mid \alg_t)=\, p\, F(a_t\mid \alg_t)\, \geq\, p\,F\left(\{o_{t}\}\backslash \{\salg_{t}\}\mid \alg_{t}\right). \]
	\hfill\Halmos	\end{proof}

\subsection{Analysis of Cascade Sampling} \label{appx:cassam}
\begin{repeattheorem}[Theorem \ref{nonmadapt}.]
	Cascade Sampling (Algorithm \ref{cassam}) is $0.25$-competitive for ONSW-SO in the adversarial model. 
\end{repeattheorem}
\begin{proof}{Proof.} 
	Since Cascade Sampling is both randomized and adaptive, we need substantial new notation for this proof that will be introduced at various points as needed to clarify the argument. 
	
	Let \aalg\ denote the Cascade Sampling algorithm and let \opt\ denote the optimal adaptive offline benchmark. Fix an arbitrary realized mapping $P\in \curn$ for both \aalg\ and \opt. 
	With $P$ fixed, every action maps to a unique outcome and \opt\ outputs a deterministic set of actions. 
	Let $e^{O}_t$ and $e^{A}_t$ denote the realized outcomes of the actions chosen by \aalg\ and \opt\ at arrival $t$. Let $\cal{E}^{A}=\{e^A_1,e^A_2,\cdots,e^A_T\}$ and $\cal{E}^{O}=\{e^O_1,e^O_2,\cdots,e^O_T\}$ denote the set of all outcomes in \aalg\ and \opt\ respectively. Let $\cal{E}^A_t=\{e^A_1,\cdots, e^A_{t-1}\}$ denote the set of realized outcomes in \aalg\ prior to arrival $t$. 
 
	Let $E[\cdot]$ denote expectation w.r.t.\ the randomness in \aalg. Note that every outcome $e\in P$ appears in $\cal{E}^A$ with probability at most 0.5 (not necessarily independently) because (excluding the null action) $\aalg$ does not choose any action w.p.\ more than 0.5. Thus,
	\begin{eqnarray}
		0.5\,f (\cal{E}^O)&\leq &E[f\left(\cal{E}^A\cup \cal{E}^O\right)],\nonumber\\
		&=&	E\left[f\left(\cal{E}^A \cup \left( \cup_{t\in [T]} \{e^O_{t}\}\backslash \{e^A_{t}\}\right)\right)\right],\nonumber\\
		&\leq & E[f(\cal{E}^A)]+ \sum_{t\in [T]} E\left[f\left(\{e^O_{t}\}\backslash \{e^A_{t}\}\mid \cal{E}^A_t\right)\right]\label{onswinterim}
	\end{eqnarray}
	We use Lemma \ref{subsample} with $w=f$, $W(0.5)=\cal{E}^A$, and $S=\cal{E}^O$ to get the first inequality. The second inequality follows by linearity of expectation and submodularity of $f$.
	
	Thus far, we fixed $P$ and followed the proof template of Theorem \ref{advmain} to derive upper and lower bounds on $f(\cal{E}^A\cup \cal{E}^O)$. 
	Now, we consider inequality \eqref{onswinterim} in expectation w.r.t.\ the randomness in the mapping $P$. We begin with some notation. Let $\bm{P}$ denote a random mapping from actions to outcomes. We decompose $\bm{P}$ into $\bm{P}_t$ and $\bm{P}_{-t}$, where $\bm{P}_t$ is the random mapping from $\ac_t$ to $N_t$ and $\bm{P}_{-t}$ is the random mapping from $\curac\backslash \ac_t$ to $N\backslash N_t$. Let $E_{\bm{P}}[\cdot], E_{\bm{P}_t}[\cdot], E_{\bm{P}_{-t}}[\cdot]$ denote expectation w.r.t.\ the randomness in $\bm{P}$, $\bm{P}_t$, and $\bm{P}_{-t}$ respectively. Note that the sets $\cal{E}^A$ and $\cal{E}^O$ are now random sets but but we continue to use the original notation for brevity. Also, the set of outcomes $\cal{E}^A_{t}$ is independent of the mapping $\bm{P}_{t}$.  Taking expectation w.r.t.\ $\bm{P}$ on both sides of inequality \eqref{onswinterim}, we have,
	\begin{eqnarray}
		0.5 E_{\bm{P}}[f(\cal{E}^O)]	&\leq & E_{\bm{P}}[E[f(\cal{E}^A)]]+ \sum_{t\in [T]} E_{\bm{P}}\left[E\left[f\left(\{e^O_{t}\}\backslash \{e^A_{t}\}\mid \cal{E}^A_t\right)\right]\right],\nonumber\\
		&\leq & E_{\bm{P}}[E[f(\cal{E}^A)]]+ \sum_{t\in [T]} E_{\bm{P}_{-t}}\left[E_{\bm{P}_t}\left[E\left[f\left(\{e^O_{t}\}\backslash \{e^A_{t}\}\mid \cal{E}^A_t\right)\right]\right] \big|\, \bm{P}_{-t}=P_{-t}\right],\label{lawtot1}\\
		&\leq & E_{\bm{P}}[E[f(\cal{E}^A)]]+ \sum_{t\in [T]} E_{\bm{P}_{-t}}\left[E_{\bm{P}_t}\left[E\left[f\left( e^A_{t}\mid \cal{E}^A_t\right)\right]\right] \big|\, \bm{P}_{-t}=P_{-t}\right],\label{onswinterim2}\\
		&= & E_{\bm{P}}[E[f(\cal{E}^A)]]+ \sum_{t\in [T]} E_{\bm{P}}\left[E\left[f\left( e^A_{t}\mid \cal{E}^A_t\right)\right]\right] ,\label{lawtot2}\\
		&=& 2\, E_{\bm{P}}[E[f(\cal{E}^A)]].\nonumber
	\end{eqnarray}
	First, note that the final inequality gives us the desired lower bound on competitive ratio of \aalg\ because $E_{\bm{P}}[f(\cal{E}^O)]$ is the objective value of the offline benchmark and $E_{\bm{P}}[E[f(\cal{E}^A)]]$ is the objective value of \aalg. Inequality \eqref{lawtot1} and identity \eqref{lawtot2} follow from the law of total expectation (also called the tower rule). To complete the proof, we need to prove inequality \eqref{onswinterim2}. To this end, it suffices to show that the following inequality holds for every $t\in [T]$, conditioned on $\bm{P}_{-t}=P_{-t}$:
	\begin{equation}\label{onswinterim3}
		E_{\bm{P}_t}\left[E\left[f\left(\{e^O_{t}\}\backslash \{e^A_{t}\}\mid \cal{E}^A_t\right)\right]\right] \leq E_{\bm{P}_t}\left[E\left[f\left( e^A_{t}\mid \cal{E}^A_t\right)\right]\right].
	\end{equation}
	\paragraph{Proof of \eqref{onswinterim3}:} Let $o_t$ denote the action chosen by $\opt$ at arrival $t$. Since \opt\ is non-anticipatory, the selection of $o_t$ is independent of $\bm{P}_t$. Let $\{{a}_1,{a}_2,\cdot,{a}_{m_t}\}$ denote the ordered set of actions with non-negative marginal value at arrival $t$ in \aalg, sorted in descending order of marginal values. Note that inequality \eqref{onswinterim3} holds trivially when $o_t\notin \{a_1,\cdots, a_{m_t}\}$, since  the left-hand-side of \eqref{onswinterim3} is at most zero, while the right-hand-side is non-negative. Similarly, if $m_t=0$, i.e., there are no actions with non-negative marginal value at $t$, then the right-hand-side is zero whereas the left-hand-side is at most zero. So we assume that $m_t\geq 1$ and  $o_t={a}_\ell$ for some $\ell\in[m_t]$.
	
	Recall that \aalg\ chooses action ${a}_\ell$ with probability $2^{-\ell}$. Interchanging the order of expectations on the right-hand-side of inequality \eqref{onswinterim3} we get,
	\[E_{\bm{P}_t}\left[E\left[f\left(e^A_{t}\mid \cal{E}^A_t\right)\right]\right] =E\left[E_{\bm{P}_t}\left[f\left(e^A_{t}\mid \cal{E}^A_t\right)\right]\right] =\,\sum_{j\in[m_t]} 2^{-j}\sum_{e\in N({a}_j)}p_e\, f(e\mid \cal{E}^A_t),\]
	here we ignored the null action because it contributes no value. Now, we make two important observations. First, with probability $1-1/2^{\ell}$, \aalg\ does not choose ${a}_\ell$ (which is the same as $o_t$), and the set $\{e^O_t\}\backslash \{e^A_t\}=\{e^O_t\}= \bm{P}_t\cap N({a}_\ell)$. With probability $1/2^{\ell}$, \aalg\ chooses ${a}_\ell$ and we have $\{e^O_t\}\backslash \{e^A_t\}=\emptyset$. Thus,
	\[E_{\bm{P}_t}\left[E\left[f\left(\{e^O_{t}\}\backslash \{e^A_{t}\}\mid \cal{E}^A_t\right)\right]\right]= (1-2^{-\ell}) \sum_{e\in N({a}_\ell)}p_e\, f(e\mid \cal{E}^A_t).\]
	The second observation is simply that, 
	\[\sum_{j\leq \ell}2^{-j}=1-2^{-\ell}.\]
	Using these two observations along with the fact that every action in $\{a_1,a_2,\cdots,a_{m_t}\}$ has non-negative marginal value we obtain, 
	\begin{eqnarray*}
		E_{\bm{P}_t}\left[E\left[f\left(\{e^O_{t}\}\backslash \{e^A_{t}\}\mid \cal{E}^A_t\right)\right]\right]&=& (1-2^{-\ell}) \sum_{e\in N({a}_\ell)}p_e\, f(e\mid \cal{E}^A_t),\\
		&=&  \sum_{j\in[\ell]} \left[2^{-j} \sum_{e\in N({a}_\ell)}p_e\, f(e\mid \cal{E}^A_t)\right],\\
		&\leq &  \sum_{j\in[\ell]} \left[2^{-j} \sum_{e\in N({a}_j)}p_e\, f(e\mid \cal{E}^A_t)\right],\\
		&\leq &  \sum_{j\in[m_t]} \left[2^{-j} \sum_{e\in N({a}_j)}p_e\, f(e\mid \cal{E}^A_t)\right],\\
		&=& E_{\bm{P}_t}\left[E\left[f\left(\{e^A_{t}\}\mid \cal{E}^A_t\right)\right]\right].
	\end{eqnarray*}
 The two inequalities follow from the fact that, $\sum_{e\in N({a}_1)}p_e\, f(e\mid \cal{E}^A_t) \geq\sum_{e\in N({a}_2)}p_e\, f(e\mid \cal{E}^A_t)\geq \cdots\geq \sum_{e\in N({a}_\ell)}p_e\, f(e\mid \cal{E}^A_t)\geq \cdots \geq \sum_{e\in N({a}_{m_t})}p_e\, f(e\mid \cal{E}^A_t)\geq 0.$
	\hfill\Halmos	\end{proof}
\begin{remark}[\agrd\ is 0.5-competitive for OSOW-SO]\label{agrd0.5}
 For several special cases of OSOW-SO, it is known that \agrd\ achieves a competitive ratio of 0.5. However, in Example~\ref{algequiv}, we presented an instance of OSW-SO in which \agrd\ attains only half the expected reward of \grd. In fact, a refinement of the analysis of Cascade Sampling shows that the competitive ratio of \agrd\ is exactly $0.5$. Specifically, when $f$ is  monotone and admits an arrival-consistent submodular order, we obtain the following strengthened version of inequality \eqref{onswinterim}:
\begin{equation}\label{agrdinterim}
	f (\cal{E}^O)\leq  f(\cal{E}^A)+ \sum_{t\in [T]} f\left(\{e^O_{t}\}\backslash \{e^A_{t}\}\mid \cal{E}^A_t\right).
\end{equation}
Moreover, by the definition of \agrd, we have
\[E_{\bm{P}_t}\left[f\left(\{e^O_{t}\}\backslash \{e^A_{t}\}\mid \cal{E}^A_t\right)\right]\leq  E_{\bm{P}_t}\left[f\left(\{e^A_{t}\}\mid \cal{E}^A_t\right)\right].\]
Taking expectation over $\bm{P}$ on both sides of \eqref{agrdinterim}, we conclude that \[E_{\bm{P}}[f(\cal{E}^O)]\leq 2\, E_{\bm{P}}[f(\cal{E}^A)],\] 
as desired.
	\end{remark}


%

\end{APPENDICES}

\end{document}